\documentclass[a4paper,11pt]{article}
\usepackage{graphicx}
\usepackage{amsmath}
\usepackage{amssymb}
\usepackage{amsthm}
\usepackage{csquotes}
\usepackage{fullpage}
\usepackage{tabularx}
\usepackage{booktabs}
\usepackage[bookmarks,linkcolor=black,hidelinks]{hyperref}
\usepackage{xspace}
\usepackage{xcolor}
\usepackage[font=small,labelfont=bf,labelsep=period]{caption}
\usepackage[font=small,labelfont=normalfont]{subcaption}
\usepackage{lineno}
\usepackage{xspace}

\title{Minimum Rectilinear Polygons\\
        for Given Angle Sequences
  \thanks{A preliminary version of this paper appeared in the
  	Proceedings of the 18th Japan Conference on Discrete and
        Computational Geometry and Graphs (JCDCGG
        2015)~\cite{efkssw15}.  Note that this preliminary version had
        inaccuracies in the NP-hardness proof.  Another version of
        this article is contained in the dissertation of K.~Fleszar~\cite{thesisFleszar}.}}
\newcommand{\email}[1]{\texttt{#1}}
\newcommand{\affil}[1]{\emph{#1}}
\author{%
  William S.\ Evans%
  \thanks{\affil{Department of Computer Science, University of British Columbia, Canada}, Supported by NSERC Discovery grant.}
  \thanks{\email{will@cs.ubc.ca}}%
  \and
  Krzysztof Fleszar%
  \thanks{\affil{Institute of Informatics, University of Warsaw, Poland},
  	Supported by CONICYT Grant PII 20150140 and by ERC consolidator grant TUgbOAT no.\ 772346.
  		  \email{kfleszar@mimuw.edu.pl}, \href{https://orcid.org/0000-0002-1129-3289}{\tt orcid.org/0000-0002-1129-3289}.}%
  \and
  Philipp Kindermann%
  \thanks{\affil{Lehrstuhl f\"ur Informatik~I, Universit\"at W\"urzburg, Germany},
  \href{www1.informatik.uni-wuerzburg.de/en/staff}{www1.informatik.uni-wuerzburg.de/en/staff},
  \email{firstname.lastname@uni-wuerzburg.de}}
  \footnote{\href{https://orcid.org/0000-0001-5764-7719}{\tt
		orcid.org/0000-0001-5764-7719}.}%
\and
  Noushin~Saeedi%
\footnotemark[2]
\thanks{\email{noushins@cs.ubc.ca}}%
  \and
  Chan-Su Shin%
  \thanks{\affil{Division of Computer and Electronic Systems Engineering, Hankuk Univ.\ of Foreign Studies, South Korea},\newline \email{cssin@hufs.ac.kr}, \href{https://orcid.org/0000-0003-3073-6863}{\tt orcid.org/0000-0003-3073-6863}}%
  \and
  Alexander Wolff
  \footnotemark[5]
  \footnote{\href{https://orcid.org/0000-0001-5872-718X}{\tt
  		orcid.org/0000-0001-5872-718X}.}
}

\usepackage{amsthm}
\theoremstyle{plain}
\newtheorem{theorem}{Theorem}
\newtheorem{observation}{Observation}
\newtheorem{lemma}{Lemma}
\newtheorem{definition}{Definition}
\newtheorem{corollary}{Corollary}
\newtheorem{proposition}{Proposition}

\graphicspath{{pictures/}}

\usepackage{enumitem}
\newlist{listArabic}{enumerate}{1}
\setlist*[listArabic,1]{label=(\arabic*),itemsep=0ex}
\newlist{listRoman}{enumerate}{1}
\setlist*[listRoman,1]{label=(\roman*),itemsep=0ex}
\newlist{inlinelistArabic}{enumerate*}{1}
\setlist*[inlinelistArabic,1]{label=(\arabic*),itemsep=0ex}
\newlist{inlinelistAlph}{enumerate*}{1}
\setlist*[inlinelistAlph,1]{label=(\alph*),itemsep=0ex}
\newlist{inlinelistRoman}{enumerate*}{1}
\setlist*[inlinelistRoman,1]{label=(\roman*),itemsep=0ex}
\newlist{listAlph}{enumerate}{1}
\setlist*[listAlph,1]{label=(\alph*),itemsep=0ex}

\newcommand{\xSpaceSE}{~} %Extra space for simple equations, that is, containing only one (in)equality-sign (all others have extra space ~ already in formula)
\newcommand{\xSpaceASD}{\quad} %Extra space when defining an angle sequence
\newcommand{\thSuffix}{-th\xspace}

\newcommand{\bigOh}{\mathcal{O}}
\newcommand{\problemName}[1]{\textsc{#1}}
\newcommand{\runtimeclass}[1]{\ensuremath{\mathsf{#1}}\xspace}
\newcommand{\NP}{\runtimeclass{NP}}
\newcommand{\APX}{\runtimeclass{APX}}

\newcommand{\len}[1]{\left\|#1\right\|}
\newcommand{\floor}[1]{\left\lfloor #1 \right\rfloor}
\newcommand{\ceil}[1]{\left\lceil #1 \right\rceil}
\newcommand{\minStyle}[1]{\min\left\{#1\right\}}
\newcommand{\maxStyle}[1]{\max\left\{#1\right\}}

\newcommand{\MRPfGASlong}{\problemName{Minimum Rectilinear Polygon for Given Angle Sequence}\xspace}
\newcommand{\LS}{\texttt{L}}
\newcommand{\RS}{\texttt{R}}
\newcommand{\area}[1]{\textrm{area}(#1)}
\newcommand{\peri}[1]{\textrm{peri}(#1)}

\newcommand{\TL}{\ensuremath{\mathit{TL}}\xspace}
\newcommand{\TR}{\ensuremath{\mathit{TR}}\xspace}
\newcommand{\BR}{\ensuremath{\mathit{BR}}\xspace}
\newcommand{\BL}{\ensuremath{\mathit{BL}}\xspace}

\newcommand{\rlength}{\text{\RS-length}}

\newcommand{\lenBL}{b}
\newcommand{\lenTR}{a}
\newcommand{\segTR}{\tau}
\newcommand{\sarea}{\textrm{area}}
\newcommand{\BB}[1]{B_{#1}}%Boundind Box of objects in the part of the paper about efficient algorithms for monotone special cases of our problem
\newcommand{\Btl}{\BB{\TL}}
\newcommand{\Bbr}{\BB{\BR}}

\newcommand{\Popt}{P^*}
\newcommand{\rnum}[1]{\operatorname{r}(#1)}

\newcommand{\exTR}{\widehat{\TR}}
\newcommand{\exBL}{\widehat{\BL}}
\newcommand{\partHullTop}{P_{\top}}
\newcommand{\partHullBot}{P_{\bot}}

\newcommand{\rectR}{R} %The (WxH)-rectangle of FitBoundingBox
\newcommand{\rojo}{\rho} %The winding number in the 2nd part

\newcommand{\ez}[1]{e_{#1}^z}
\newcommand{\ex}[1]{e_{#1}^{x}}
\newcommand{\ey}[1]{e_{#1}^{y}}
\newcommand{\edgeIn}{e_{\mathrm{in}}}
\newcommand{\edgeOut}{e_{\mathrm{out}}}

\newcommand{\subSeqName}[2]{\mathrm{#1}_{#2}} %E.g.: 1=Spiral, 2=in
\newcommand{\subSeqNamePlus}[3]{\mathrm{#1}^\mathrm{#2}_{#3}} %E.g.: 1=Spiral, 2=in, 3=j
\newcommand{\inVarSpiral}[1]{\subSeqNamePlus{spiral}{in}{#1}}
\newcommand{\outVarSpiral}[1]{\subSeqNamePlus{spiral}{out}{#1}}
\newcommand{\inSpiral}{\inVarSpiral{i}}
\newcommand{\outSpiral}{\outVarSpiral{i}}
\newcommand{\snail}[1]{\subSeqName{snail}{#1}}
\newcommand{\inOuterLadder}{\subSeqNamePlus{outerLadder}{in}{i}}
\newcommand{\outOuterLadder}{\subSeqNamePlus{outerLadder}{out}{i}}
\newcommand{\varInnerDoubleLadder}[1]{\subSeqName{innerDoubleLadder}{#1}}
\newcommand{\innerDoubleLadder}{\varInnerDoubleLadder{i}}
\newcommand{\ladder}{\mathsf{defLadder}} %Current ladder (e.g., the ladder defining the order)
\newcommand{\BigInOuterLadder}{\subSeqNamePlus{outerLadder}{in}{}}
\newcommand{\BigOutOuterLadder}{\subSeqNamePlus{outerLadder}{out}{}}
\newcommand{\BigInInnerLadder}{\subSeqNamePlus{innerLadder}{in}{}}
\newcommand{\BigOutInnerLadder}{\subSeqNamePlus{innerLadder}{out}{}}
\newcommand{\BigInSpiral}{\subSeqNamePlus{spiral}{in}{}}
\newcommand{\BigOutSpiral}{\subSeqNamePlus{spiral}{out}{}}
\newcommand{\BigOutSnail}{\subSeqNamePlus{snail}{out}{}}
\newcommand{\BigInSnail}{\subSeqNamePlus{snail}{in}{}}

\newcommand{\lowerValue}{lower value\xspace} % The notion we use to refer to \lbEdge{..}
\newcommand{\lowerValues}{lower values\xspace}
\newcommand{\lbEdge}[1]{\mathrm{low}(#1)} %Lower bound on spiral edge #1
\newcommand{\lbSpirals}{\mathrm{lowSpirals}} %Lower bound on all spiral edges
\newcommand{\lbOuterLadders}{\mathrm{lowLadders}} %Lower bound on both outer ladders (2nd part)
\newcommand{\lowerBoundAreaBB}{\mathrm{lowBBArea}} %Lower bound on the bounding box area (2nd part)
\newcommand{\vUBcenter}{\mathrm{maxCenterCost}} %Upper bound on the area in the center

\newcommand{\bboxGeneral}[1]{\textrm{BB}(#1)} %Bounding box of object #1
\newcommand{\bboxSpiral}[1]{\textrm{BB}_{#1}} %Spiral box (Bounding box of a spiral until level #1)
\newcommand{\inbboxSpiral}[1]{\bboxSpiral{#1}^\textrm{in}} 
\newcommand{\outbboxSpiral}[1]{\bboxSpiral{#1}^\textrm{out}}

\date{}

\begin{document}
\maketitle

\begin{abstract}
A \emph{rectilinear} polygon is a polygon whose edges are
axis-aligned.  Walking counterclockwise on the boundary of such a
polygon yields a sequence of left turns and right turns.  The number
of left turns always equals the number of right turns plus~$4$.  It is
known that any such sequence can be realized by a rectilinear
polygon.  

In this paper, we consider the problem of finding
realizations that minimize the perimeter or the area of the polygon
or the area of the bounding box of the polygon.  
We show that all three problems are \NP-hard in general.
This answers an open question of Patrignani [CGTA 2001],
who showed that it is \NP-hard to minimize the area of the bounding
box of an orthogonal drawing of a given planar graph.
We also show that realizing polylines with minimum bounding box area is \NP-hard.
Then we consider the special
cases of $x$-monotone and $xy$-monotone rectilinear polygons.  For
these, we can optimize the three objectives efficiently.
\end{abstract}

	\section{Introduction}
In this paper, we consider the problem of computing, for a given
\emph{rectilinear angle sequence}, a ``small'' rectilinear polygon that 
realizes the sequence.  A rectilinear angle sequence~$S$ is a
sequence of left ($+90^\circ$) turns, denoted by~$\LS$, and right ($-90^\circ$) turns, denoted by~$\RS$.
We write~${S=(s_1,\dots,s_n)\in\{\LS,\RS\}^n}$, where~$n$ is the
\emph{length} of~$S$. 
As we consider only rectilinear angle
sequences, we usually drop the term ``rectilinear.''  
A polygon~$P$
\emph{realizes} an angle sequence~$S$ if there is a counterclockwise (\emph{ccw})
walk along the boundary of~$P$ such that the turns at the vertices
of~$P$, encountered during the walk, form the sequence~$S$.  The turn
at a vertex~$v$ of~$P$ is a left or right turn if the interior angle
at~$v$ is~${90^\circ}$ ($v$ is convex) or, respectively,~${270^\circ}$
($v$ is reflex).
We call the problem \MRPfGASlong.

In order to measure the size of a polygon, we only consider polygons
that lie on the integer grid.  In this context, the \emph{area} of a polygon~$P$
corresponds to the number of grid cells that lie in the interior
of~$P$.  The \emph{bounding box} of~$P$ is the smallest axis-parallel
enclosing rectangle of~$P$.  The \emph{perimeter} of~$P$ is the sum of
the lengths of the edges of~$P$.  The task is, for a given angle
sequence~$S$, to find a simple\footnote{We use the following strong notion of simplicity: A polyline is \emph{simple} if it visits every grid point at most once. Thus, neither crossings nor revisits of a same point are allowed. Similarly, a polygon is simple if the (closed) polyline realizing its boundary is simple.} polygon that realizes~$S$ and minimizes
\begin{listAlph}
	\item its bounding box, 
	\item its area, or 
	\item its perimeter.
\end{listAlph}
Thereby, \emph{minimizing the bounding box} is short for \emph{minimizing the area of the bounding box}. Figure~\ref{fig:example} shows that, in general, the three criteria
cannot be minimized simultaneously.
\begin{figure}
	\centering \subcaptionbox{\label{fig:example-A}Area~$11$, perimeter~$20$.}{~~~\includegraphics[page=1]{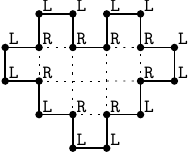}~~~} \hfil \subcaptionbox{\label{fig:example-B}Area~$10$, perimeter~$22$.}{~~~\includegraphics[page=2]{example2}~~~}
	\caption{Two polygons realizing the same angle sequence. The
		bounding box of both polygons has area~$20$,
                but~(\subref{fig:example-A}) shows a polygon of
                minimum perimeter and~(\subref{fig:example-B}) one of
                minimum area.}
	\label{fig:example}
\end{figure}

Obviously, the angle sequence of a polygon is unique (up to rotation),
but the number of polygons that realize a given angle sequence is
unbounded.  The formula for the angle sum of a polygon implies that,
in any angle sequence,~${n=2r+4}$, where~$n$ is the length of the sequence and~$r$ is the number of right
turns. In other words, the number of right turns is exactly four less
than the number of left turns. 

\paragraph{Related Work.}
Bae et al.~\cite{bos-abrpr-cgta19} considered, for a given angle sequence~$S$,
the polygon~${P(S)}$ that realizes~$S$ and minimizes its area. They studied the 
following question: Given a number~$n$, find an angle sequence~$S$ of
length~$n$ such that the area of~${P(S)}$ is minimized, 
or maximized. 
Let~${\delta(n)}$ denote the minimum area and let~${\Delta(n)}$ denote the maximum area for~$n$.
They showed 
\begin{listRoman}
	\item ${\delta(n)=n/2-1}$ if~${n\equiv 4 \bmod 8}$,~${\delta(n)=n/2}$ otherwise, and 
	\item ${\Delta(n)=(n-2)(n+4)/8}$ for any~$n$ with~${n\geq 4}$. 
\end{listRoman}
The result for~${\Delta(n)}$ tells us that any angle sequence~$S$ of length~$n$
can be realized by a polygon with area at most~${(n-2)(n+4)/8}$. 

Several authors have explored the problem of realizing a turn sequence.
Culberson and Rawlins~\cite{cr-socg85} and Hartley~\cite{h-dpgas-IPL89}
described algorithms that, given a sequence of exterior angles summing
up to~${2\pi}$, construct a simple polygon realizing that angle sequence.
Culberson and Rawlins' algorithm, when constrained to~${\pm 90^\circ}$ angles, produces polygons with no colinear
edges, implying that any~$n$-vertex polygon can be drawn with area
approximately~${(n/2-1)^2}$.  However, as Bae et
al.~\cite{bos-abrpr-cgta19} showed, the bound is not tight.
In his PhD thesis, Sack~\cite{Sack84} introduced label sequences
(which are equivalent to turn sequences) and, among others, developed
a grammar for label sequences that can be realized as simple
rectilinear polygons.
Vijayan and Wigderson~\cite{VijayanWigderson85} considered the problem
of drawing \emph{rectilinear graphs}, of which
rectilinear polygons are a special case, using an edge labeling that
is equivalent to a turn sequence in the case of paths and cycles.

In graph drawing, the standard approach to drawing a graph of maximum
degree~$4$ orthogonally (that is, with rectilinear edges) is the
topology--shape--metrics approach of Tamassia~\cite{t-oeggm-SICOMP87}:
\begin{listArabic}
	\item Compute a planar(ized) embedding, that is, a circular order of edges around each vertex that admits a crossing-free drawing; 
	\item\label{tamassia:angle_seq} compute an \emph{orthogonal representation}, that is, an angle
	sequence for each edge and an angle for each vertex;
	\item\label{tamassia:drawing}  \emph{compact} the graph, that is, draw it inside a bounding box
	of minimum area.
\end{listArabic}
Step~\ref{tamassia:drawing} is
  NP-hard\footnote{Patrignani~\cite{WADS99Titto}
      claimed that Step~\ref{tamassia:drawing} is APX-hard for planar
      graphs.  He has, however, withdrawn his
      claim~\cite{PrivComTitto}, so it is still open whether this
      step admits a PTAS.} for planar graphs as shown by
  Patrignani~\cite{p-ococ-CG01}.
For non-planar graphs, it is even inapproximable within a polynomial factor unless~${\runtimeclass{P}=\NP}$ as shown by Bannister et al.~\cite{bannister12}.
Note that an orthogonal representation
computed in step~\ref{tamassia:angle_seq} is essentially an angle sequence for each face of
the planarized embedding, so our problem corresponds to step~\ref{tamassia:drawing} in
the special case that the input graph is a simple cycle.

Another related work contains the reconstruction of a simple
(non-rectilinear) 
polygon from partial geometric information.  Disser et
al.~\cite{dmw-cgta11} constructed a simple polygon in~${\bigOh(n^3\log n)}$
time from an ordered sequence of angles measured at the vertices
visible from each vertex.  Chen and Wang~\cite{cw-cgta12} showed how
to solve the problem in~${\bigOh(n^2)}$ time,
which is optimal in the worst-case.  Biedl et
al.~\cite{bds-tcs11} considered polygon reconstruction from points
(instead of angles) captured by laser scanning devices.
Very recently, Asaeedi et al.~\cite{asaeedi18} 
encloses
a given set of points by a simple polygon 
whose vertices are a subset of the points and 
that optimizes some criteria (minimum area, maximum perimeter or maximum number of vertices). 
The vertex angles are constrained to lie below a threshold.

\paragraph{Our Contribution.}
First, we show that finding a minimum polygon that realizes a given
angle sequence is \NP-hard for any of the three measures: bounding box
area, polygon area, and polygon perimeter; see Section~\ref{sec:nph}.
This hardness result extends the one of Patrignani~\cite{p-ococ-CG01} and settles
an open question that he posed. 
We note that in an extended abstract~\cite{efkssw15} of this paper
there were some inaccuracies in our proof that now have been addressed.
As a corollary, we infer that realizing an angle sequence as a \emph{polyline} within a given rectangle is \NP-hard.

In this paper, we also give efficient algorithms for special types of angle sequences,
namely~$xy$- and~$x$-\emph{monotone sequences}, 
which are realized by~$xy$-monotone and~$x$-monotone
polygons, respectively.  
For example, Figure~\ref{fig:example} depicts an~$x$-monotone polygon realizing the~$x$-monotone
sequence~{\LS\LS\RS\RS\LS\LS\RS\LS\LS\RS\LS\RS\LS\LS\RS\LS\RS\LS\LS\RS}. 
Our algorithms for these angle sequences
minimize the bounding box and the area (Section~\ref{sec:area-algo}) and the perimeter
(Section~\ref{sec:peri-algo}).
For an overview of our results, see
Table~\ref{tbl:summary}. 
Throughout this paper, a \emph{segment} is always an axis-aligned line segment.

\begin{table}
	\centering
	\caption{Summary of our results.}
	\label{tbl:summary}
	
	\begin{tabular}{l@{\qquad}c@{\qquad}c@{\qquad}c}
		\toprule
		Type of sequences & Minimum area & Minimum bounding
                box & Minimum perimeter \\
		\midrule
		general & \NP-hard  & \NP-hard  & \NP-hard \\
                $x$-monotone & $\bigOh(n^4)$ & $\bigOh(n^3)$ &
                $\bigOh(n^2)$ \\ $xy$-monotone & $\bigOh(n)$ &
                $\bigOh(n)$ & $\bigOh(n)$ \\
		\bottomrule
	\end{tabular}
\end{table}

\section{NP-Hardness of the General Case}\label{sec:nph}
In contrast to the special cases that we efficiently solve in later
sections, the general case of our problem turns out to be \NP-hard.
In two steps, we show \NP-hardness for all three objectives:
minimizing the perimeter of the polygon, the area occupied by the
polygon, and the area of the bounding box.  
First, in Section~\ref{sec:np-fitboundingbox}, we consider the
base problem defined below from whose \NP-hardness we then derive the
three desired results in Section~\ref{sec:np-fitboundingbox}.
As a warm-up, we show via a reduction from our base
  problem that realizing angle sequences by polylines that minimize
  the area of their bounding box 
	is \NP-hard; see Section~\ref{sec:np-polylines}.

The setting of the base problem is a little different from the general case. 
Given an angle sequence~$S$, we do not look for a polygon that realizes it while minimizing one of the three objectives, but for a polyline that lies within some given rectangle.
We say that a (simple rectilinear) polyline~$P$ \emph{realizes}~$S$ if we can walk along~$P$ from one of its endpoints to the other one and
observe exactly the same angle sequence as~$S$. Note that the endpoints of~$P$ do not have angles, hence, a polyline that realizes~$S$ has~${|S|+2}$ vertices (including its endpoints) and~${|S|+1}$ edges. 
Furthermore, if all edges of~$P$ have unit length, then~${\peri{P}=|S|+1}$. 
For polylines in this section, we define~${|P|}$ as the number of inner vertices of~$P$, that is,~${|P|=|S|}$. 
Throughout the section, we will interchangeably use the names of angle sequences to refer to fixed polylines realizing them. For example,~$\peri{S}$ would denote the perimeter of a polyline realizing~$S$ that we fixed before. 

\medskip\noindent
\problemName{FitBoundingBox}: %
An instance~${\langle S,W,H \rangle}$ of this problem consists of an angle sequence~$S$ and positive even integers~$W$ and~$H$.
A \emph{feasible drawing} of~$S$ (with respect to~$W$ and~$H$) is a simple rectilinear polyline~$P$ realizing~$S$ within an axis-parallel rectangle of
width~${W+10}$ and height~${H+10}$ 
such that the first and last edge 
of~$P$ are horizontal and such that~$P$ can be extended to a simple polygon (not necessarily within the rectangle) by connecting its endpoints with a simple rectilinear polyline. 
An instance is called a \emph{no-instance} if there is no feasible drawing of~$S$.
An instance is called a \emph{yes-instance} if there is a feasible drawing of~$S$ within 
an (even smaller) axis-parallel rectangle~$\rectR$ of width~$W$ and height~$H$ such that the first
vertex of~$P$ lies on the upper-left corner of~$\rectR$ and the last vertex
of~$P$ lies on the lower-right corner of~$\rectR$.
An instance is \emph{valid} if it is a yes- or a no-instance (note that
not every instance is valid).
The problem is to decide whether a given valid instance of~\problemName{FitBoundingBox} 
is a yes- or a no-instance.

The \enquote{gap} between yes- and no-instances will help us to differentiate among them in our hardness proof: To classify a valid instance as a yes-instance, it suffices to show that it admits a feasible drawing, and, similarly, to classify it as a no-instance, it suffices to show that it does not admit a feasible drawing in the smaller~${(W\times H)}$-rectangle.

\subsection{NP-Hardness of Drawing a Polyline with Given Angle Sequence}
\label{sec:np-polylines}

In this section, we show, by reducing from \problemName{FitBoundingBox},
that it is \NP-hard to draw a polyline that realizes a given angle sequences within a given rectangle. 
As a corollary, we show that it is also \NP-hard to draw a realizing polyline that minimizes the area of its bounding box.
The reduction will be a warm-up for the upcoming sections, where the
technical details are more involved.

Recall that our task in \problemName{FitBoundingBox} is to realize a polyline within a given rectangle.
The polyline has to obey certain constraints; for instance, it needs to be drawn such that one can complete it to a simple polygon, which is not possible in a drawing where the endpoints are \enquote{blocked}. 
Thus, a natural question is whether the problem is also \NP-hard without this constraint.
Formally, is it \NP-hard to decide whether a given angle sequence~$S$ can be realized by a polyline within an axis-parallel rectangle of given width~$W$ and height~$H$?
The question has been answered in the affirmative by Noushin
Saeedi---one of the authors of this article---in her upcoming
PhD-thesis where she also shows that finding a polyline of minimum
perimeter is \NP-hard.
In the following, we present an alternative proof by a reduction from \problemName{FitBoundingBox} whose hardness will be proven in the next section.

\begin{theorem}
  \label{thm:polylineNPhard}
  Deciding whether 
  there is
  a polyline
  realizing a given angle sequence
  that 
  can be drawn within a given
  axis-parallel rectangle
  is
  \NP-hard.
 The problem remains \NP-hard 
  even if constrained to angle sequences for which
  the bounding box of any realizing polyline contains the given rectangle (up to rotation by 90 degrees). 
\end{theorem}
\begin{proof}
	Let~${\langle S,W,H \rangle}$ be a valid instance of \problemName{FitBoundingBox} and recall that~$W$ and~$H$ are even integers. 
	Let~${W'=W+2}$ and~${H'=H+4}$.
Consider the angle sequence
\[(\RS\RS\LS\LS)^{{W}/{2}}\] 
that we call the \emph{horizontal ladder}. 
Observe that any polyline realizing it is~$x$-monotone and its width is at least~$W$ provided that the first (and last) segment is drawn vertically. 
The \emph{double ladder} is defined as
\[(\RS\RS\LS\LS)^{{W}/{2}}~\RS\RS\LS\RS~(\RS\LS\LS\RS)^{(H+2)/{2}}~.\]
Any polyline realizing the double ladder such that the first edge 
is drawn vertically consists of an~$x$-monotone part 
realizing~$(\RS\RS\LS\LS)^{{W}/{2}}~\RS\RS\LS$ and 
a~$y$-monotone part realizing~$\LS\RS~(\RS\LS\LS\RS)^{(H+2)/{2}}$; the two parts overlap in two segments. 
The width of the first part is at least~$W+2$, that is,~$W'$, and 
the height of the second part is at least~$H+4$, that is,~$H'$. 

We define a new angle sequence~$S'$ by placing~$S$ between the horizontal ladder and the double ladder: 
        \[S' \xSpaceASD=\xSpaceASD (\RS\RS\LS\LS)^{{W}/{2}}\,\LS
        \quad S \quad \RS (\RS\RS\LS\LS)^{{W}/{2}}~\RS\RS\LS\RS~(\RS\LS\LS\RS)^{{(H+2)}/{2}}~.\]
	Figure~\ref{fig:np-polyline} depicts a possible realization when~${\langle S,W,H \rangle}$ is a yes-instance.
	The double ladder guarantees that any drawing realizing~$S'$ has width at least~$W'$ and height at least~$H'$ (or vice versa). 
	As proven later, the two ladders 
	enforce the endpoints of~$S'$ to lie on the border of the bounding box if its size is exactly~$ W'\times H'$.
	\begin{figure}[t]
		\centering
		\includegraphics[page=1]{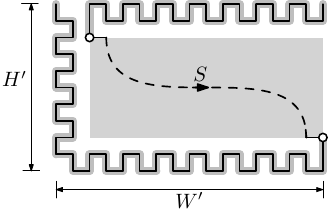}
		\caption{%
			A realization of~$S'$ within a~${( W' \times  H')}$-rectangle. 
			The gray~${(W \times H)}$-rectangle contains the subpolyline corresponding to~$S$ whose endpoints are depicted as discs. The ladders are highlighted in gray; the horizontal ladder has width~$W$, the double ladder has width~$W'=W+2$ and height~$H'=H+4$. Note that the first and last edge of~$S$ are horizontal if and only if the first and last edge of each of the two ladders are vertical.
		}
		\label{fig:np-polyline}
	\end{figure}

	We claim that~$S'$ can be realized as a polyline within an axis-parallel rectangle of width~$W'$ and height~${H'}$ if and only if~${\langle S,W,H \rangle}$ is a yes-instance. 
	The theorem follows from this claim as it is \NP-hard to decide whether a \emph{valid} instance of \problemName{FitBoundingBox} is a yes- or no-instance; see Section~\ref{sec:np-fitboundingbox}.
	
	In the first direction of the proof, we assume that~${\langle S,W,H \rangle}$ is a yes-instance. By definition, there is a realization of~$S$ within a $(W\times H)$-rectangle such that the first (and last) edge is horizontal and the first vertex lies on the upper-left corner and the last vertex lies on the lower-right corner. 
	Consider such a realization. 
	From~$S$, we realize the two ladders as depicted in Fig.~\ref{fig:np-polyline}: we set the length of the ladder edges incident to the endpoints of~$S$ (which are vertical) to~$2$ and the remaining ladder edge lengths to~$1$. In total, we obtain a polyline realizing~$S'$ whose bounding box has width~$W'$ and height~${H'}$. 
		
	In the other direction of the proof, we assume that there is a polyline realizing~$S'$ within a~${( W'\times  H')}$-rectangle~$R'$. 
	To show that~${\langle S,W,H \rangle}$ is a yes-instance, it suffices to show that 
	\begin{listRoman}
		\item\label{polyl:prop:rec}~$S$ can be realized as a polyline within a $((W+10)\times(H+10))$-rectangle such that
		\item\label{polyl:prop:hor}~it's first (and last) edge is horizontal and such that
		\item\label{polyl:prop:compl}~the polyline can be completed to a (simple) polygon.
	\end{listRoman} 
	Since~$S$ is a subsequence of~$S'$ and since~${W'=W+2}$ and~$H'=H+4$, 
	we immediately get Property~\ref{polyl:prop:rec}. 
	Suppose Property~\ref{polyl:prop:hor} does not hold. Thus, the first (and last) 
	segment of~$S$ is drawn vertically; consequently, the first segment of 
	the double ladder is horizontal as it is orthogonal to the last (and first) segment of~$S$.  
	Hence, as discussed above, the double ladder has width at least~$H'$ 
	and height at least~$W'$. 
	If~${W'\not=H'}$, then the polyline 
	does not fit into the~${( W'\times  H')}$-rectangle~$R'$; a contradiction. 
	If~${W'=H'}$, we can rotate the drawing by~${90^\circ}$ and obtain the desired property. 
	Regarding Property~\ref{polyl:prop:compl}, we show that one of the two point pairs lie on the boundary of~$R'$: both endpoints of~$S'$, or the first vertex of~$S$ and the last vertex of~$S'$. In either case, we can connect these two points in the exterior of 
	the rectangle via a polyline which gives us a polygon containing~$S$ as desired. 
	Consider the part of the polyline realizing the~$y$-monotone part of 
	the double ladder and call it~$P$. Since~$P$ is~$y$-monotone and has 
	height at least~$H'$ which is also the height of~$R'$, 
	both endpoints of~$P$ necessarily lie on opposite sides of~$R'$. 
	Hence, the last vertex of~$S'$ lies on the border 
	of~$R'$ as claimed in both cases. Since~$P$ cuts~$R'$ into two parts 
	(where one is possibly empty) 
	the remaining part of the polyline has to lie entirely in one 
	of the two parts.
	Since~$P$ uses at least two 
	grid points from every horizontal grid line within~$R'$, 
	the bounding box,~$B$, of the remaining part has width at most~$W$.
	Consider the part~$Q$ of the polyline realizing 
	the horizontal ladder. 
	Its last segment is vertical as it is orthogonal to the first segment of~$S$; thus~$Q$ has width at least~$W$.
	Since~$Q$ is also~$x$-monotone,
	both its endpoints necessarily lie on opposite (vertical) 
	sides of~$B$. Consequently, one of the two endpoints also 
	lies on the border of~$R'$. Our claim follows as the endpoints of~$Q$ are identical to the first vertex of~$S'$ and the first vertex of~$S$.
\end{proof}

	Note that Theorem~\ref{thm:polylineNPhard} immediately implies \NP-hardness of the 
 optimization problem where we want to 
 realize a given angle sequence by a polyline whose bounding box has minimum area.
Indeed, consider any polyline realizing a constrained angle sequence 
as defined in Theorem~\ref{thm:polylineNPhard}.
By assumption, the bounding box of the polyline
contains the given rectangle. Thus, its area cannot be smaller
than that of the rectangle. Furthermore, both areas are equal 
if and only if the bounding box is identical to the rectangle (up to translation).
\begin{corollary}
  It is~\NP-hard to draw a polyline realizing a given angle sequence
  such that the area of the polyline's bounding box is minimized.
\end{corollary}

\subsection{NP-Hardness of FitBoundingBox}\label{sec:np-fitboundingbox}
To show the \NP-hardness of \problemName{FitBoundingBox}, we reduce 
from \problemName{3-Partition}: 
Given a multiset~$A$ of~$3m$ integers~${a_1, \dots, a_{3m}}$ 
with~${\sum_{i=1}^{3m} a_i = Bm}$, is there a partition of~$A$ into~$m$
subsets 
such that, for each subset~${A'}$,~${\sum_{a \in A'} a = B}$?
It is known that \problemName{3-Partition} is \NP-hard even 
if~$B$ is polynomially bounded in~$m$ and, for every~${a \in A}$, we have~${B/4 < a < B/2}$, which implies that every subset 
must contain exactly three
numbers~\cite{GJ79}.

\begin{figure}[t]
	\centering
	\includegraphics[page=1]{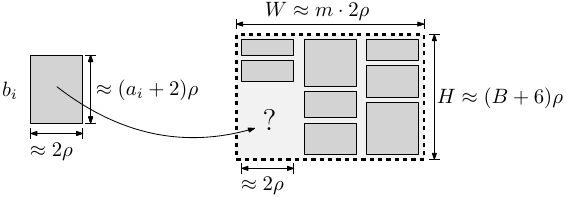}
	\caption{%
		Overview of our reduction for~${m=3}$. The boxes are shaded.}
	\label{fig:np-rectangle}
\end{figure}

Equivalently, we can ask the question whether we can pack~$3m$ boxes, where 
the~$i$\thSuffix box has width~$1$ and height~$a_i$, into a rectangle of 
width~$m$ and height~$B$. The problem remains the same if, for some~$\rho$, we 
scale the boxes and the rectangle horizontally by~${2\rho}$ and vertically 
by~$\rho$. Assuming~${B/4<a_i}$, the vertical space allows at most three boxes 
above each other. Hence, the problem remains equivalent even if we further 
add~${2\rho}$ to the height of each box and~${3\cdot 2\rho}$ to the height of the 
rectangle. Thus, the question is, can we pack~$3m$ boxes~${b_1,\dots, b_{3m}}$, 
where~$b_i$ has width~${2\rho}$ and height~${a_i\rho + 2\rho}$, into a rectangle
of width~${W\approx m \cdot 2\rho}$ and height~${H\approx (B+6)\rho}$? See 
Fig.~\ref{fig:np-rectangle} for an overview of our reduction. We 
create an angle sequence~$S$ that contains, for each~$b_i$, a subsequence called 
a \emph{snail} whose minimum bounding box is~$b_i$. By ensuring that the snails 
are \enquote{more or less} disjoint, \problemName{3-Partition} reduces to 
\problemName{FitBoundingBox} via the following question: Can we draw~$S$ inside 
a~${(W \times H)}$-rectangle?

\begin{theorem}\label{thm:fitNPhard}
	\problemName{FitBoundingBox} is \NP-hard.
\end{theorem}

We now prove Theorem~\ref{thm:fitNPhard}. 
Let~$c_W$ and~$c_H$ be sufficiently big even constants that we discuss at the end of the proof.
Given an instance for \problemName{3-Partition} as defined above, we set~${\rho=4B^3 m^7}$ and assume that~$m$ is larger than a sufficiently big constant depending on~$c_W$ and~$c_H$. 
We set
\[{W'\xSpaceSE=\xSpaceSE 2m \rho + c_W m^2}\] and \[{H'\xSpaceSE=\xSpaceSE (B+6) \rho + c_H m}\] and choose~${W=W'-10}$ and~${H=H'-10}$ 
for our~${(W \times H)}$-rectangle~$\rectR$ (note that~$W$ and~$H$ are even as desired).  
In the following, we create step by step an angle sequence~$S$ for \problemName{FitBoundingBox} consisting of~$3m$ subsequences, 
called \emph{snails}, each corresponding to an (integer) number~$a_i$ in~$A$.
We will show that~${\langle S, W, H \rangle}$ is a valid instance with
the property that it is a yes-instance of \problemName{FitBoundingBox} if and only if the \problemName{3-Partition} instance is a yes-instance. 
The number of angles in~$S$ (as well as the time to construct~$S$)
will be polynomially bounded in~$m$.

Before we define~$S$, 
let us consider the snails.
For~${a_i\in A}$, a snail has the property that if we draw it with minimum perimeter, then its bounding box has roughly width~${2\rho}$ and height~${(a_i + 2)\rho}$. Observe that~$W$ provides enough width to draw~$m$ snails next to each other along a horizontal line, but not more than that (for sufficiently big values of~$m$).
Furthermore,~$H$ provides enough height to draw three snails above each other (along a vertical line) if and only if 
the corresponding numbers in~$A$ 
add up to at most~$B$, 
provided~$m$ is sufficiently big; 
see Fig.~\ref{fig:np-rectangle}.
By forcing that, in any feasible drawing of~$S$, each snail is drawn with roughly minimum perimeter, we will get the property that all the bounding boxes of the snails are basically disjoint and drawn in one of~$m$ \enquote{columns} with three boxes per column. Hence, given the heights of the boxes, this will allow us to directly \enquote{read} a solution to the underlying \problemName{3-Partition}-instance.

\begin{figure}[t]
	\centering
	\includegraphics[page=3]{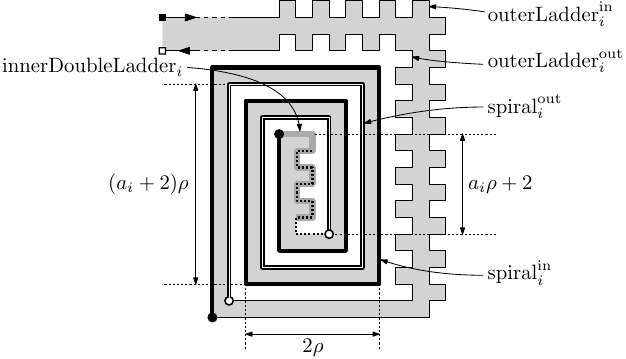}
	\caption{%
		A compact drawing of~$\snail{i}$,~${i\in\{1,\dots,3m\}}$. Here,~${\rho=4}$. The shaded area depicts the exterior of the polygon that one obtains by interconnecting the endpoints of~$S$.
		Note that~$\innerDoubleLadder$ consists of two overlapping inner ladders (bold gray and bold dashed) that we define later.}
	\label{fig:snail}
\end{figure}  

We now describe a snail in more detail.  Let~$a_i$ be the 
corresponding number in~$A$.  The \enquote{heart} of a snail
is its \emph{inner double ladder}. This is a~$y$-monotone sequence built by alternating
two left and two right turns, and it has minimum width~$1$ and minimum height~${a_i\rho+2}$.
It consists of two overlapping \emph{inner ladders} that we define later. 
The inner double ladder is
placed in the center of two \emph{spirals} that wind around it approximately~$\rho$ times (later we'll give a precise definition of winding that depends on the number of horizontal spiral edges); see Fig.~\ref{fig:snail}. 
The bounding box of the spirals will have width at least~${2\rho}$ and
height at least~${(a_i +2)\rho}$. In order to ensure that a spiral is
winding around the inner double ladder, each spiral will end with a so called
\emph{outer ladder}, a combination of a~$y$- and an~$x$-monotone
sequence of some large minimum length. A spiral winding around its
outer ladder is thus much longer than when winding around its inner
ladder. Hence, it consumes many more grid points than in the other
case. By ensuring that the number of grid points in a~${(W\times H)}$-rectangle (and even in a~${(W'\times H')}$-rectangle) is just
big enough for every spiral being drawn with its minimum perimeter
(that is, winding around its inner double ladder), there won't be enough free
grid points for any spiral to be drawn in a \enquote{bad way} (several
times around the outer ladder).
We will use the following upper bound on the number of grid points in a~${(W'\times H')}$-rectangle (that also bounds its area~${W' \cdot H'}$):

\begin{alignat}{2}
(W'+1) \cdot (H'+1)&=&~&(2m \rho + c_W m^2) \cdot ((B+6) \rho + c_H m) + W'+ H'\nonumber\\
&\le&& 2(B+6)m \rho^2 + c_{W\!H} B m^2 \rho \label{formula:gridarea}
\end{alignat}
where~${c_{W\!H}}$ is some sufficiently large constant.

Formally, we define the angle sequence~$\snail{i}$.  The superscript
\enquote{in} denotes that the respective angle sequence comes in~$S$
before the inner double ladder, whereas the superscript \enquote{out}
denotes that the sequence comes after the ladder.  (Recall that the
first and last vertex of a polyline have no angles and are omitted in
the angle sequence.) 
Note that~$\snail{i}$ has exactly two more right turns than left ones.

\begin{eqnarray*}
	\snail{i} &=&(\inOuterLadder) ~\RS~ (\inSpiral) 
	~\RS~ (\innerDoubleLadder)  ~\LS~ 
	(\outSpiral) ~\LS~ (\outOuterLadder)~,
\end{eqnarray*}%\[\snail{i} =  \inOuterLadder ~ \RS ~ \inSpiral  ~\RS ~\innerDoubleLadder ~ \LS ~ \outSpiral ~ \LS ~\outOuterLadder,\]
where 
\begin{eqnarray*}
	\inOuterLadder &=& \LS \RS (\RS\LS\LS\RS)^{\rho} \RS (\LS\RS\RS\LS)^{(a_i + 2)\rho/2 + 3} \RS~, \\ 
	\inSpiral & = & \RS^{2\rho}, \\
	\innerDoubleLadder & =& (\RS\RS\LS\LS)^{a_i\rho/2+1}~,\\
	\outSpiral & = & \LS^{2\rho-2}, \\
	\outOuterLadder &=&  \LS (\LS\RS\RS\LS)^{(a_i + 2)\rho/2+1} \LS (\RS\LS\LS\RS)^{\rho}~.   
\end{eqnarray*}

Independently of how we complete our definition of the whole angle sequence~$S$, as long as all spiral drawings are rotated such that the inner double ladders start (and end) with horizontal edges oriented to the right (as in Fig.~\ref{fig:snail}), we can
prove the following, using a number of intermediate lemmas. 

\begin{proposition}\label{prop:fitbb-to-3part}
	Given a feasible drawing of~$S$, 
	we can efficiently decode a solution to the underlying \problemName{3-Partition} instance. In other words, if~${\langle S, W, H \rangle}$ is a yes-instance, then the underlying \problemName{3-Partition} instance is a yes-instance.
\end{proposition}

Fix some feasible drawing of~$S$ inside an axis-aligned~${(W'\times H')}$-rectangle~${\rectR'}$.  (Recall that we have~${W'=W+10}$ and~${H'=H+10}$).
Let~${\kappa=1/(Bm^2)}$, and let~${\iota=(1-3\kappa)}$.  Note that for
increasing~$m$,~$\kappa$ gets arbitrarily close to~$0$ and~$\iota$
arbitrarily close to~$1$. 
Let the \emph{center} of an inner double ladder denote the center
point of its bounding box.  For~${1\le i \le 3m}$, let~$R_i$ be the
box of width~${2\iota\rho}$ and height~${2\iota\rho + a_i \rho}$
centered at the center of~$\innerDoubleLadder$.  Let~$\mathcal{R}$
denote the set of all these boxes.

Observe that, by definition of~$\iota$, a box~$R_i$ has width slightly smaller than~${2\rho}$ and height slightly smaller than~${(a_i + 2)\rho}$. 
Later we use this fact to prove that these boxes are pairwise disjoint if the drawing is feasible. 
If the boxes were slightly bigger, they possibly might overlap.

We now show a special case of Proposition~\ref{prop:fitbb-to-3part}.
Later we will see that we can always assume this ``special case'',
which will prove Proposition~\ref{prop:fitbb-to-3part}.

\begin{lemma}\label{lem:boxes-to-3part}
	If all boxes in~$\mathcal{R}$ are pairwise disjoint and lie completely inside~${\rectR'}$, then we can efficiently decode a solution to the underlying \problemName{3-Partition} instance.
\end{lemma}
\begin{proof}
	We place the origin on the upper-left corner of~${\rectR'}$ and,
	for~${1\le j\le m}$, we place a vertical line at~$x$-coordinate~${(2j-1)\rho}$.
	
	First, suppose that there is a box~${R_i\in\mathcal{R}}$ not intersected by any of these vertical lines.
	Then~$R_i$ lies between two vertical lines as it is too wide ($2\iota\rho$) to fit before the leftmost or after the rightmost vertical line (which offer only~$\rho$ and~${\rho+c_W m^2}$ horizontal space, respectively).
	Recall that the distance between two vertical lines is~${2\rho}$.
	Let~$j$ be the number of vertical lines to the left of~$R_i$.  Observe
	that the distance between the left edge of~$R_i$ and the~$j$\thSuffix
	vertical line from the left is at most~${2\rho - 2\iota\rho = 2(1-\iota)\rho = 6\kappa\rho}$.
	Hence, the distance between the left edge of~$R_i$ and the left edge
	of~${\rectR'}$ is at most~${(2j-1+6\kappa)\rho}$.
	Consider any horizontal line that intersects~$R_i$.
	The number of boxes to the left of~$R_i$ intersected by this line is
	at most~${j-1}$ since~$j$ boxes have total width
	\[ j \cdot 2\iota\rho ~=~ 2j(1-3\kappa)\rho ~=~ (2j-6j\kappa)\rho ~>~  (2j-1+6\kappa)\rho~.\]
	
	By repeating the same argument for the right side of~$R_i$, we observe that any horizontal line intersecting~$R_i$ intersects at most~${m-1}$ boxes including~$R_i$. 
	Consider the parts of such a line not covered by the boxes. Their total length inside~${\rectR'}$ is at least~${W'-(m-1)2\iota\rho \ge 2\rho}$. 
	Given that~$R_i$ has height~${2\iota\rho + a_i \rho \ge 2\rho}$, 
	taking the integral over the uncovered parts of all horizontal lines intersecting~$R_i$ gives us~${2\rho \cdot 2\rho = 4 \rho^2}$ 
	as 
	a lower bound on the total area to the left and to the right of~$R_i$ (inside~${\rectR'}$) that is not covered by the boxes. However, this is a contradiction as the total area of~${\rectR'}$ is
	\[{2(B+6)m \rho^2 + c_{W\!H} B m^2 \rho}\] by Inequality~\ref{formula:gridarea}, 
	and the total area occupied by all the boxes is at least
	\begin{alignat*}{2}
	&&~&\sum\limits_{i=1}^{3m} 2\iota\rho \cdot (2\iota\rho + a_i \rho)\\
	&=&&\sum\limits_{i=1}^{3m} 4\iota^2\rho^2 ~+ ~2\iota\rho^2 \sum\limits_{i=1}^{3m} a_i \\
	&=&& 12m(1-6\kappa + 9\kappa^2)\rho^2 ~+~ 2Bm(1-3\kappa)\rho^2\\
	&\ge&& 2(B+6)m \rho^2 - 72m\kappa\rho^2 - 6 Bm\kappa\rho^2\\
	&\ge&& 2(B+6)m \rho^2 - 78 \rho^2/m
	\end{alignat*}
	implying an upper bound
	of~${c_{W\!H} B m^2 \rho + 78 \rho^2/m \le 79\rho^2/m}$ (if~$\rho$ is
	sufficiently large) on the total area not covered by the boxes.  This,
	however, is less than the area of~${4\rho^2}$ that we lose to the left
	and to the right of~$R_i$ if~$R_i$ is not intersected by a vertical
	line.  Consequently, each box in~$\mathcal{R}$ is intersected by one
	of the~$m$ vertical lines.
	
	Next, assume that each vertical line intersects exactly three
        boxes  
	of~$\mathcal{R}$.  These three boxes correspond to three numbers
	in~$A$, so each vertical line corresponds to a subset of~$A$ of
	cardinality~3.  Since there are~$3m$ boxes and~$m$ vertical lines, and
	since each box is intersected by at least one vertical line, these
	subsets form a partition of the numbers in~$A$.
	We claim that in each such subset~${A'}$, the numbers sum up to at most~$B$. 
	This holds as otherwise~${\sum_{a \in A'} a \ge B + 1}$ 
	and, thus, the total height of the three corresponding boxes would be at least
	\begin{alignat}{1}
	\sum_{a \in A'} \left( 2\iota\rho + a\rho\right) ~\ge~ 6\iota\rho + (B+1)\rho ~>~ (B+6)\rho + c_H m ~=~  H' \label{formula:toTall}
	\end{alignat}
	(by using~${\iota > 5/6 + 1/7}$ and~${6/7 \rho > c_H m}$), which would
	be strictly greater than the height of~${\rectR'}$; a contradiction.
	Hence, given the total sum~$Bm$ of all numbers in~$A$, the numbers of each subset sum up to exactly~$B$. Our partition is therefore a feasible solution to the underlying \problemName{3-Partition} instance.
	
	Finally, suppose that a vertical line intersects four boxes. 
	Recall that for any~$a_i$ in~$A$ we have~${B/4 < a_i < B/2}$. Hence, the numbers corresponding to these four boxes sum up to a value strictly larger than~$B$. 
	By a similar calculation as in Inequality~\eqref{formula:toTall}, 
	we can show that the total height of the four boxes is strictly greater 
	than the height of~${\rectR'}$; a contradiction.
\end{proof}

We will now show that the boxes in~$\mathcal{R}$ are indeed pairwise disjoint and lie inside~${\rectR'}$. 
We begin with a simple observation about feasible drawings of~$S$.

\begin{figure}
	\centering
	\includegraphics[page=2]{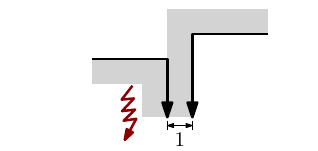}
	\caption{Two parallel edges with the same orientation must
		have distance at least~$2$. The interior of the polygon is
		shaded.}
	\label{fig:forbidden}
\end{figure} 

\begin{observation}\label{obs:minDistance}
	In a feasible drawing of~$S$, if two edges are oriented in the same
	way and their projections on a line parallel to both edges overlap,
	then their distance is at least~$2$.
\end{observation}
\begin{proof}
	By assumption, the endpoints of~$S$ can be appropriately connected
	to obtain a drawing of a simple polygon.  The orientation of an edge
	determines on which side the interior of the polygon lies.  Now
	consider two edges that are oriented in the same way and whose
	projections on a line parallel to both edges overlap; see
	Fig.~\ref{fig:forbidden}.  If the distance of the two edges was~$1$
	then, for one of the two edges, the interior of the polygon would
	lie on both of its sides; a contradiction.
\end{proof}

To facilitate the arguments in the following proofs, we introduce several notions. 
An inner double ladder consists of two overlapping \emph{inner ladders} that we obtain by either removing the first two or last two vertices from it; see Fig.\ref{fig:snail}. 
Thus, each inner ladder is incident to one spiral and its minimum height is~${a_i \rho + 1}$ where~$a_i$ is the corresponding number in~$A$.
In the context of a fixed spiral, the inner ladder refers always to the inner ladder incident to the spiral.
Furthermore, we use the following notation concerning the edges of the (fixed) spiral. 
Let~${\rho'_x}$ and~${\rho'_y}$ denote the number of its horizontal and vertical edges, respectively ($\rho'_x=\rho'_y-1$ and~${\rho'_y\in\{\rho,\rho+1\}}$). 
Let~${z\in\{x,y\}}$ and consider all edges of the spiral parallel to the~$z$-axis.
We define two orders on 
the edges along the spiral. 
In the \emph{inner order}, the first edge is incident to the inner ladder, in the \emph{outer order}, the first edge is incident to the outer ladder.
For a given order, let~${\ez{1}, \dots, \ez{\rho'_z}}$ denote the~$z$-axis-parallel spiral edges in this order and let~$\ladder$ 
denote the ladder defining the order. 
For an edge~$\ez{i}$, we call~$i$ its \emph{level} with respect to the order,
and denote its length by~${\len{\ez{i}}}$. 

Throughout the proof, we fix a spiral 
and a~${z\in\{x,y\}}$.
Note that all claims shown hold for any spiral and coordinate axis.

\begin{figure}[t]
	\centering 
	\begin{subfigure}[b]{.47\textwidth}
		\centering
		\includegraphics[page=3]{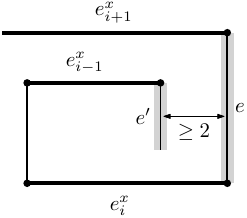}
		\caption{$\ex{i-1}$ is closer to~$\ex{i}$ than~${\ex{i+1}}$ to~$\ex{i}$.}
		\label{fig:minlen-order-spiral}
	\end{subfigure}
	\hfil
	\begin{subfigure}[b]{.47\textwidth}
		\centering
		\includegraphics[page=4]{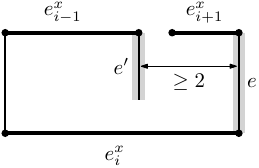}
		\caption{$\ex{i-1}$ and~${\ex{i+1}}$ have the same distance to~$\ex{i}$.}
		\label{fig:minlen-order-two-spirals}
	\end{subfigure}
	\caption{At least one of both,~${\ex{i-1}}$ and~${\ex{i+1}}$, is shorter than~$\ex{i}$. The edges~$e$ and~${e'}$ (highlighted) are parallel and have the same orientation. Therefore, their distance is at least~$2$.}
	\label{fig:minlen}
\end{figure}

\begin{lemma}\label{lem:greaterEqualTwo}
	Let~${1< i < \rho'_z}$. We have~${\len{\ez{i}}\ge\len{\ez{i-1}}+2}$  or~${\len{\ez{i}}\ge\len{\ez{i+1}}+2}$. 
\end{lemma}
\begin{proof}
	Assume that~${\ez{i-1}}$ has the smallest distance to~$\ez{i}$ among~${\{\ez{i-1},\ez{i+1}\}}$; see Fig.~\ref{fig:minlen}.
	We have~${\len{\ez{i-1}}<\len{\ez{i}}}$, as otherwise~${\ez{i-1}}$ would intersect the edge~$e$ connecting~$\ez{i}$ with~${\ez{i+1}}$; a contradiction to the drawing being simple. 
	Furthermore, consider the edge~${e'}$ that is incident to~${\ez{i-1}}$ and not incident to~$\ez{i}$. 
	The edges~$e$ and~${e'}$ are parallel and oriented in the same way.
	Since~${\ez{i-1}}$ has the smallest distance to~$\ez{i}$ among~${\{\ez{i-1},\ez{i+1}\}}$, 
	the projection
	of~${e'}$ on the line through~$e$ is contained in~$e$.
	Thus, by Observation~\ref{obs:minDistance}, the distance between~$e$ and~${e'}$ is at least~$2$; hence,~${\len{\ez{i}}\ge\len{\ez{i-1}}+2}$.  
	By repeating the same argument for the case that~${\ez{i+1}}$ is closer to~$\ez{i}$ than~${\ez{i-1}}$, the claim follows.
\end{proof}

Note that both inequalities of Lemma~\ref{lem:greaterEqualTwo} can be fulfilled for at most one edge since, in a cascading manner, it forces all other edges to satisfy exactly one of the two inequalities. 
Consequently, one of the following three cases holds (see Fig.~\ref{fig:spiral-configs}).

\begin{corollary}\label{cor:spiralConfigs}
	One of the following three cases holds:
	\begin{enumerate}
		\item ${\len{\ez{1}}< \dots < \len{\ez{\rho'_z-1}}}$, or
		\item ${\len{\ez{2}}> \dots > \len{\ez{\rho'_z}}}$, or
		\item\label{spiralcase:three}
		there is an~$i$ with~${1< i < \rho'_z-1}$ such that
		\[\len{\ez{1}}< \dots < \len{\ez{i}} \qquad\textrm{ and }\qquad \len{\ez{i+1}}> \dots > \len{\ez{\rho'_z}} \qquad\textrm{ holds.}\]
	\end{enumerate}
\end{corollary}

If~${\len{\ex{i}}>\len{\ex{i-1}}}$, 
then we say that the spiral \emph{winds~$i$ times} around~$\ladder$. 
Although we use this definition only for horizontal edges, note that~${\len{\ex{i}}>\len{\ex{i-1}}}$ implies~${\len{\ey{i}}>\len{\ey{i-1}}}$. 

\begin{observation}\label{obs:biggerThanLadder}
	Let~$b_x$ and~$b_y$ 
	denote the width and height of the bounding box of~$\ladder$, respectively.
	Let~${1 < i \le \rho'_z}$. If~${\len{\ez{i}}\ge\len{\ez{i-1}}+2}$, 
	then~${\len{\ez{i}}\ge 2i + b_z}$. 
\end{observation}

\begin{figure}[t]
	\centering 
	\includegraphics{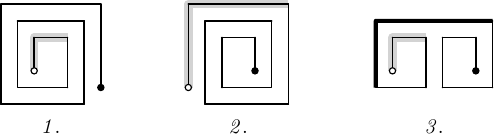}
	\caption{The three cases of Corollary~\ref{cor:spiralConfigs}~
		(gray:~$\ex{1}$,~$\ey{1}$; bold:~$\ex{i}$,~$\ey{i}$ of case~\eqref{spiralcase:three}).
	}
	\label{fig:spiral-configs}
\end{figure}

\begin{proof}
	By Corollary~\ref{cor:spiralConfigs} and Lemma~\ref{lem:greaterEqualTwo}, we 
	have~${\len{\ez{j}}\ge \len{\ez{j-1}} + 2}$ for~${1 < j \le i}$.
	Hence, we have~${\len{\ez{i}} \ge \len{\ez{1}} + 2(i-1)}$.
	We now show~${\len{\ez{1}}\ge b_z + 2}$ and the claim will follow.
	
	Let~${e_1, e_2, e_3, e_4}$ denote the first four edges of the spiral in the order defined by~$\ladder$; see Fig.~\ref{fig:ladderForcesMinLen}. 
	Note that~$e_1$ is vertical, so~${\ey{1}=e_1}$ and~${\ex{1}=e_2}$.
	Recall~${\len{e_3}={\len{\ey{2}}>\len{\ey{1}}}=\len{e_1}}$.
	Thus, by monotonicity,~$\ladder$ lies completely inside the bounding box of~$e_1$ and~$e_2$.
	Consider any horizontal edge of~$\ladder$ with smallest distance to~$e_2$.
	Observe that in the case of the outer ladder as well as in the case of the inner ladder,
	this edge lies on the border of the bounding box of~$\ladder$ and has the same orientation as~$e_2$.
	Furthermore, observe that the same holds for~$e_3$: 
	Any vertical edge of~$\ladder$ with smallest distance to~$e_3$
	lies on the border of the bounding box of~$\ladder$ and has the same orientation as~$e_3$.
	Hence, by Observation~\ref{obs:minDistance}, the bounding box of~$\ladder$ has distance at least~$2$ to~$e_2$ and to~$e_3$. 
	Now, observe that the height~$b_y$ of this bounding box and its distance to~$e_2$ sum up to exactly~${\len{e_1}}$.
	Thus,~${\len{e_1}\ge b_y +2}$ and, similarly,~${\len{e_2}\ge b_x +2}$.
\end{proof}

\begin{definition}\label{def:lowerBoundsA}
	For~${1\le i \le 3m}$, we define for every spiral edge~$e$ 
	belonging to~$\snail{i}$ its \emph{\lowerValue} as
	\begin{itemize}
		\item ${\lbEdge{e}=2j}$ if~$e$ is horizontal and 
		\item ${\lbEdge{e}=2j+a_i\rho}$ otherwise
	\end{itemize}
	where~$j$ is the level of~$e$ with respect to the inner order. 
	We denote by~$\lbSpirals$ the sum of the \lowerValues over all edges of all spirals.
\end{definition}

Now we show that the \lowerValues of the edges are proper lower bounds on their lengths.

\begin{lemma}\label{lem:lbPerimeter}
	In any feasible drawing, every spiral edge~$e$ has length at least~$\lbEdge{e}$ and the total perimeter of all spirals is at least~${\lbSpirals \ge 2(B+6)m\rho^2}$.
\end{lemma}
\begin{proof}
	Consider any spiral edge~$e$ and its spiral belonging to~$\snail{i}$. 
	For a moment, consider any order of the spiral edges and let~$b_x$ and~$b_y$ denote the width and height, respectively, of the bounding box of the ladder defining the order.
	In case of the inner order, 
	we have 
	\[b_x \ge 1 \quad\textrm{ and }\quad b_y\ge a_i \rho ~,\] 
	and in case of the outer order, we have
	\[b_x\ge 2\rho'_x  \quad\textrm{ and }\quad b_y\ge 2\rho'_y+a_i \rho ~.\] 
	To see the latter case, observe that the~$x$-monotone parts of~$\inOuterLadder$ 
	and~$\outOuterLadder$ consist of at least~${2\rho}$ horizontal edges, and 
	the~$y$-monotone parts consist of at least \[{(a_i + 2)\rho+2={2(\rho+1) + a_i \rho}}\] vertical edges.
	
	Thus, in any case, we have \[b_z \xSpaceSE\ge\xSpaceSE  \lbEdge{e} - 2j\] where~$j$ is the level of~$e$ in the respective order and~$z$ is the axis to which~$e$ is parallel.
	If, 
	for one of the two orders, we have~${\len{\ez{j}}\ge\len{\ez{j-1}}+2}$ where~${\ez{j}=e}$,
	then Observation~\ref{obs:biggerThanLadder} implies~${\len{\ez{j}}\ge 2j + b_z \ge \lbEdge{e}}$.
	Otherwise, Lemma~\ref{lem:greaterEqualTwo} implies that~$e$ is the first edge in one of the two orders and we have~${\len{\ez{1}}<\len{\ez{2}}+2}$ in that order (where~${\ez{1}=e}$). 
	Fix this order and let~${e_1,\dots,e_4}$ denote the first four edges of the spiral in the order when starting at~$e$ ($e_1=e=\ez{1}$,~${e_3=\ez{2}}$, and, by assumption,~${\len{e_1}<\len{e_3}+2}$); 
	see Fig.~\ref{fig:forcedTurn}. Let~$e_0$ be the edge before~$e_1$ (either belonging to~$\ladder$ or being adjacent to~$\ladder$).
	Observe that~$e_0$ has to stop before~$e_3$ as otherwise it would either intersect~$e_3$ (if~${\len{e_3}\ge \len{e_1}}$) or lie opposite to~$e_4$ with distance~$1$ (if~${\len{e_3} = \len{e_1} - 1}$) and thus contradict Observation~\ref{obs:minDistance}.
	Therefore, by monotonicity,~$\ladder$ lies completely in the bounding box of~$e_1$ and~$e_2$.
	As in the proof of Observation~\ref{obs:biggerThanLadder}, this containment implies~${\len{e_1}\ge b_z + 2 \ge \lbEdge{e}}$. 
	
	\begin{figure}[t]
		\begin{subfigure}[b]{.6\textwidth}
			\centering
			\includegraphics[page=5]{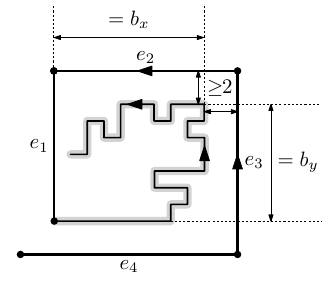}
			\caption{The spiral edges~$e_2$ and~$e_3$ are longer by~$2$ than the respective bounding box edges of the ladder (lengths~$b_x$ and~$b_y$).}
			\label{fig:ladderForcesMinLen}
		\end{subfigure}
		\hfill
		\begin{subfigure}[b]{.36\textwidth}
			\centering
			\includegraphics[page=6]{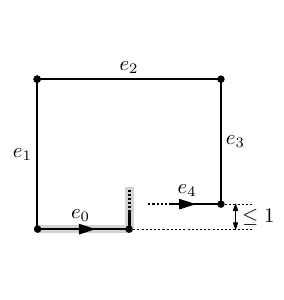}
			\caption{The edge~$e_0$ has to make a turn before reaching~$e_4$.}
			\label{fig:forcedTurn}
		\end{subfigure}
		
		\caption{The spiral winds around the ladder (highlighted). By monotonicity, the ladder cannot leave the bounding box of~$e_1$ and~$e_2$.}
	\end{figure}
	
	We are ready to show the second claim. The perimeter of the spiral is at least
	\begin{eqnarray*}
		\sum_{j=1}^{\rho'_x} \lbEdge{\ex{j}} + \sum_{j=1}^{\rho'_y} \lbEdge{\ey{j}} 
		&\ge& \sum_{j=1}^{\rho-1}2j + \sum_{j=1}^\rho(2j+a_i\rho)  \\
		&=& 2 \sum_{j=1}^{\rho-1} 2j + 2\rho + a_i \rho^2 
		~=~ (2+a_i)\rho^2  ~.
	\end{eqnarray*} 
	Recall that~${\sum_{i=1}^{3m} a_i = Bm}$ holds and that, for each~${a_i \in A}$, there are two spirals (namely~$\inSpiral$ and~$\outSpiral$).
	Thus, summing up over all spirals, we obtain
	\[\lbSpirals ~\ge~ \sum_{i=1}^{3m} 2 \cdot (2 + a_i) \rho^2 ~=~ 2(B+6)m\rho^2~.\qedhere\]
\end{proof}

\begin{definition}
	For~${1\le i \le \rho'_x}$ and the inner order, we define the \emph{spiral box}~$\bboxSpiral{i}$ as the bounding box of~$\ex{i}$ and~${\ey{i+1}}$. 
	For~${i\ge 3}$, the \emph{entrance} 
	of~$\bboxSpiral{i}$ is defined as the area between~${\ex{i-1}}$ and~${\ex{i+1}}$ (that is, as the bounding box of~${\ex{i-1}}$ and its vertical projection onto~${\ex{i+1}}$); see Fig.~\ref{fig:closedBox}.
	The \emph{height} of the entrance is the distance between~${\ex{i-1}}$ and~${\ex{i+1}}$.
	We call~$\bboxSpiral{i}$ closed if and only if its entrance has height~$2$.
	If~$\bboxSpiral{i}$ is closed, we say that the spiral is \emph{closed at level~$i$}. 
\end{definition}

By Observation~\ref{obs:minDistance}, the height of an entrance cannot be smaller than~$2$. Also observe that a spiral entering a spiral box of another spiral has to do so through the entrance. We formulate this observation as follows.
\begin{observation}\label{obs:entrance}
	Consider a spiral box~$\bboxSpiral{i}$ of a spiral. If there is a polyline distinct to the spiral containing a point inside and outside~$\bboxSpiral{i}$, then it contains a horizontal line segment intersecting the entrance of~$\bboxSpiral{i}$; see Fig.~\ref{fig:closedBox}.
\end{observation}

Recall that we set~${\kappa=1/(Bm^2)}$;~${\iota=1-3\kappa}$; and~${\rho=4B^3 m^7}$.

\begin{lemma}\label{lem:closedAndWinds}
	For every spiral, there is a~$j$ with~${\iota\rho+2 \le j \le \iota\rho+\kappa\rho}$ 
	such that the spiral  is closed at level~$j$ and winds at least~${j+\kappa\rho}$ 
	times around the inner ladder.
\end{lemma}
\begin{proof}
	Consider any spiral.
	We first show the second claim: If the spiral winds fewer than \[{\iota\rho+2\kappa\rho=\rho-\kappa\rho}\] 
	times around the inner ladder, then it winds at least~${\kappa\rho}$ times around the outer ladder. 
	Recall that the width of the bounding box of the outer ladder is at least~${2\rho'_x}$.
	Thus, by Lemma~\ref{lem:greaterEqualTwo} and Observation~\ref{obs:biggerThanLadder}, 
	for  any spiral edge~$\ex{i}$ with~${1 < i < \kappa\rho}$ in the outer order, 
	we have
	\[\len{\ex{i}} ~\ge~ 2\rho'_x + 2i ~\ge~ 2(\rho'_x-i+1) + i ~=~ \lbEdge{\ex{i}} + i~.\] 
	Hence, the perimeter of the drawing is at least 
	\begin{alignat*}{2}
	\lbSpirals + \sum_{i=2}^{\kappa\rho-1} i &\ge&~& 2(B+6)m\rho^2 + (\kappa\rho-2)^2/2  \\
	&=&& 2(B+6)m\rho^2 + (\kappa^2\rho^2-4\kappa\rho+4)/2 \\ 
	&=&& 2(B+6)m\rho^2 + \kappa^2\rho^2/4 + (\kappa^2\rho^2/2-4\kappa\rho+4)/2  \\
	&\ge&& 2(B+6)m\rho^2 + (\kappa^2\rho)\rho/4 \\
	&=&&  2(B+6)m\rho^2 + B m^3 \rho ~.
	\end{alignat*}
	However, this is strictly greater than~${2(B+6)m \rho^2 + c_{W\!H} B m^2 \rho}$
	(recall that~${c_{W\!H}}$ is a constant), which again, for a sufficiently big constant value of~${c_{W\!H}}$, is greater than the total number~${(W'+1)\cdot (H'+1)}$ of grid points offered by~${\rectR'}$ (see Inequality~\ref{formula:gridarea}); a contradiction.
	
	Next, we show the first claim. Consider the inner order. 
	If the spiral were not closed at any level
	between~${\iota\rho+2}$ and~${\iota\rho+\kappa\rho}$, 
	then, for~${1 \le i \le \kappa\rho-2}$, 
	we have
	\begin{alignat*}{2}
	\len{\ey{\iota\rho + i + 2}} &\ge&~& \len{\ey{\iota\rho + i + 1}}+3 \\
	&\ge&& \len{\ey{\iota\rho + 2}}+3i\\
	&\ge&& \lbEdge{\ey{\iota\rho + 2}} + 3i\\
	&=&& \lbEdge{\ey{\iota\rho + i + 2}} + i~.
	\end{alignat*}
	
	Again, the perimeter of the drawing is larger than~${(W'+1)\cdot (H'+1)}$; a contradiction.
\end{proof}

\begin{figure}
	\centering 
	\begin{subfigure}[b]{.53\textwidth}
		\centering
		\includegraphics[page=1]{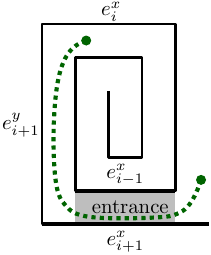}
		\caption{The dashed polyline contains a point outside and inside~$\bboxSpiral{i}$, therefore it has to go through the entrance.}
		\label{fig:closedBox}
	\end{subfigure}
	\hfil
	\begin{subfigure}[b]{.4\textwidth}
		\centering
		\includegraphics[page=2]{closedBox}
		\caption{The vertical dashed segment forces the entrance to have height at least~$3$.}
		\label{fig:openEntrance}
	\end{subfigure}  
	\caption{The shaded area depicts the entrance of the spiral box~$\bboxSpiral{i}$.}
\end{figure}

Let~$j$ be as in Lemma~\ref{lem:closedAndWinds}. We call~$\bboxSpiral{j}$ the \emph{closing box} of the respective spiral.

\begin{corollary}\label{cor:fitsIntoSpiral}
	For~${1\le i \le 3m}$, the box~$R_i$ is contained in each of the closing boxes of~$\inSpiral$ and~$\outSpiral$.
\end{corollary}
\begin{proof}
	Recall that the closing box 
	of any of the two spirals
	is closed at some level~${j \ge \iota\rho + 2}$. 
	Thus, it contains at least~${\iota\rho + 2}$ vertical and~${\iota\rho + 2}$ horizontal edges of the spiral in its interior, where at least~${\iota\rho/2 + 1}$ many of them are lying on each side (left and right, above and below) of the inner double ladder. 
	Recall that, at each of the four sides, the distance between any two neighboring parallel edges is at least~$2$ as they have the same orientation (Observation~\ref{obs:minDistance}). 
	Hence, the center of the inner double ladder lies at a distance of at least~${\iota\rho}$ to the left and to the right edge of the closing box, and at a distance of at least~${\iota\rho + a_i\rho/2}$ to the top and to the bottom edge of the closing box. Recall that~$R_i$ has width~${2\iota\rho}$ and height~${2\iota\rho + a_i \rho}$. Thus,~$R_i$ fits into the closing box when centered at the center of the inner double ladder. 
\end{proof}

\begin{lemma}\label{lem:boxesdisjoint}
	The boxes in~$\mathcal{R}$ are pairwise disjoint and lie inside~${\rectR'}$.
\end{lemma}
\begin{proof}
	The second statement follows from the fact that, by Corollary~\ref{cor:fitsIntoSpiral}, all boxes lie inside spirals and that all spirals lie inside~${\rectR'}$.
	To show the first statement, suppose for a contradiction that two boxes~$R_i$ and~$R_j$ intersect. 
	Thus, by Corollary~\ref{cor:fitsIntoSpiral}, 
	the closing boxes~$b_i$ and~$b_j$ of~$\outVarSpiral{i}$ and~$\outVarSpiral{j}$, respectively, intersect each other. 
	This intersection implies that one of the closing boxes, say~$b_i$, contains a point of the spiral corresponding to the other closing box, here~$b_j$, in its interior. 
	Consider the entrance of~$b_i$. 
	Suppose that a horizontal line segment~$s$ of~$\outVarSpiral{j}$ 
	intersects the entrance. 
	By Observation~\ref{obs:minDistance},~$s$ can be oriented only towards the entrance. 
	But then~$s$ ends with a left turn inside~$b_i$, forcing the entrance to be higher than~$2$
	(see Fig.~\ref{fig:openEntrance}). This contradicts that the entrance is closed.
	
	By Observation~\ref{obs:entrance},~$\outVarSpiral{j}$ cannot contain any point outside~$b_i$. 
	Consequently,~$\outVarSpiral{j}$ lies completely inside~$b_i$. 
	Hence, the horizontal edge~$e$ of~$\outVarSpiral{i}$ spanning~$b_i$ is longer than the longest horizontal edge of~$\outVarSpiral{j}$ whose length is at least~${2\rho'_x}$ by~Lemma~\ref{lem:lbPerimeter}. 
	Since, by Lemma~\ref{lem:closedAndWinds}, the level of~$e$ is at most~${\iota \rho + \kappa \rho \le \rho'_x-\kappa\rho}$,~$e$ is longer by at least~${2\kappa\rho}$ than its \lowerValue. 
	Also by Lemma~\ref{lem:closedAndWinds},~$\outVarSpiral{i}$ winds at least~${\kappa\rho}$ times around~$b_i$. 
	Thus, for at least~${\kappa\rho}$ edges, it holds that their length is larger by at least~${2\kappa\rho}$ than their \lowerValues.
	Thus, the perimeter of the drawing is at least 
	\[{\lbSpirals + 2(\kappa\rho)^2 ~\ge~ 2(B+6)m\rho^2 + 8 B m^3 \rho}~.\]
	However, this is strictly greater than the total number of grid points in~${\rectR'}$ (see Inequality~\ref{formula:gridarea}); a contradiction.
\end{proof}

\begin{figure}
	\centering
	\includegraphics[page=2]{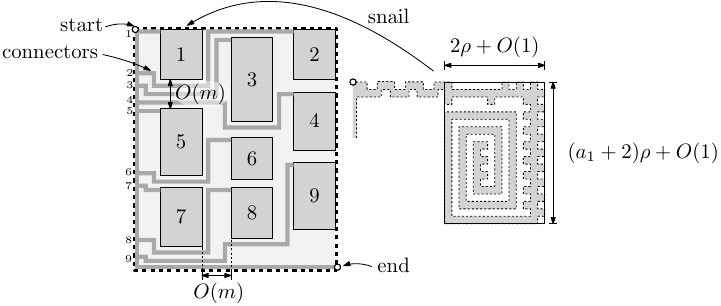}
	\caption{The polyline realizing~$S$ connects the upper-left corner
		(start) of the~${(W\times H)}$-rectangle~$\rectR$ with the
		lower-right one (end). 
		It consists of nine snails enumerated from~$1$ to~$9$.
		For readability, we shaded one part of the
		rectangle separated by the polyline in light gray, the other part
		in dark gray.  The widths of all spirals, when tightly wound
		around their inner ladders, are the same. Their heights depend on
		the corresponding numbers in the \problemName{3-Partition}
		instance. 
		Here,~${m=3}$.  All snails are packed into three
		columns, each one accommodates three snails.}
	\label{fig:hardness}
\end{figure}

Proposition~\ref{prop:fitbb-to-3part} follows immediately from
Lemmas~\ref{lem:boxes-to-3part} and~\ref{lem:boxesdisjoint}.  We will
now show the other direction of our reduction.

\begin{proposition}\label{prop:3part-to-fitbb}
	If the \problemName{3-Partition} instance is a yes-instance, 
	then there is a feasible drawing of~$S$ 
	within an axis-parallel
	rectangle~$\rectR$ of width~$W$ and height~$H$ such that,
	for the polyline~$P$ realizing~$S$,
	the first vertex of~$P$ lies on the upper-left corner of~$\rectR$ and the last vertex of~$P$ lies on the lower-right corner of~$\rectR$, that is,~${\langle S, W, H \rangle}$ is a yes-instance.
\end{proposition}
\begin{proof}
	Before we complete our definition of~$S$, let us take another look at the snails. 
	As long as we neglect on how the snails are embedded in~$S$, we can observe that every~$\snail{i}$ can be drawn inside a bounding box of width~${2\rho+\bigOh(1)}$ and height~${(a_i+2)\rho +\bigOh(1)}$ such that the first segment of the inner double ladder is horizontal and oriented to the right; see Figs.~\ref{fig:snail} and~\ref{fig:hardness}.
	The idea is now as follows. We subdivide~$\rectR$ into~$m$ columns of width~${2\rho+\Theta(m)}$ each. 
	In each column, we will draw three snails as described above, one above another. 
	The way we choose which snail to draw in which column depends on our solution to the \problemName{3-Partition} instance.
	Let~${(A_j)_{j=1}^{m}}$ be our partition of~$A$ where~${\sum_{a\in A_j} a =B}$ for every~$A_j$.
	In the~$j$\thSuffix column from the left, we draw the three snails
	corresponding to the three numbers in~$A_j$. We draw them such that their right border is aligned to the right border of the column. 
	Hence, the left part of the column of width~${\Theta(m)}$ is unused. 
	As the vertical order of the three snails, we choose the order of the corresponding numbers in the input. That is, for two snails~$\snail{i}$ and~$\snail{k}$ belonging to the same column,~$\snail{i}$ is drawn above~$\snail{k}$ if and only if~${i<k}$.
	Note that such a drawing fits into~$\rectR$: The total height of the three spirals in each column is only~${(B+6)\rho + \Theta(1)}$. 
	Hence, the total height of the unused space in 
	the columns is~${\Theta(m)}$.
	
	We will use the unused space to the left and between the snails to interconnect them. This will give us the complete angle sequence~$S$.
	We will do it by modifying our drawings of the snails by redrawing their outer ladders to connect the snails to the left edge of~$\rectR$; see the snail in Fig.~\ref{fig:hardness}. 
	More precisely, only the~$x$-monotone part of the outer ladders leaves the current bounding box of the snails. We call this part \emph{connector}.
	At the left edge  of~$\rectR$, these connector pairs will be ordered from top to bottom relative to the order of the corresponding numbers in~$A$. For any 
	connector pair, the connector oriented towards the snail will be drawn above the other one. Consecutive pairs will be connected by a vertical edge.
	The bottommost connector is connected to a vertical and then a horizontal line segment 
	allowing us to reach the lower-right corner of~$\rectR$.
	Observe that a connector consists of at least~$\rho$ up-and-down curves, hence, of enough curves to bypass all the snails encountered in the at most~${m-1}$ columns between the left edge of~$\rectR$ and the spiral it is connected to. We are ready to complete the definition of~$S$.
	\begin{eqnarray*}
		\label{eq:1}
		S       &=& \RS\LS ~(\snail{1}\LS \LS) ~\dots ~ (\snail{3m}\LS \LS)~.
	\end{eqnarray*}%
	
	Finally, observe that we can draw the connectors such that the total height of all connectors going through a column below or above the snails is~${\bigOh(m)}$. Since a connector does not need to change its~$y$-position more than once in each column (in order to bypass or to connect to a snail), 
	a total extra width of~${\bigOh(m)}$ per column is sufficient to allow the connectors to change their~$y$-positions (which happens in each column to the left of the snails). 
	Hence, we choose the constants~$c_W$ and~$c_H$ in~${W=2m \rho + c_W m^2 - 10}$ and~${H=(B+6) \rho + c_H m -10}$ as big enough even integers such that~$\rectR$ gives enough space to draw~$S$ in the way described above. 
	Also note that our drawing is feasible as it can be easily extended to a simple polygon by appropriately connecting its endpoints around~$\rectR$.
	We conclude that~${\langle S, W, H\rangle}$ is a yes-instance.
\end{proof} 

\subsection{Extension to the Optimization Versions}\label{sec:np-general}
In this section, we show for each of the three objectives (minimum perimeter, area, and bounding box) that it is \NP-hard to draw a rectilinear polygon of minimum cost that realizes a given angle sequence.  
Our proof is a reduction from \problemName{FitBoundingBox}.
Given an instance~${\langle S,W,H \rangle}$ of \problemName{FitBoundingBox}, we define an angle
sequence~$T$ (with~${|T|}$ polynomial in~${|S|}$) and, for each objective, a threshold value~$\Upsilon$ 
such that~$T$ can be drawn with cost at most~$\Upsilon$ (with respect to the objective) if and only if~$S$ is a yes-instance. 
We consider only drawings that are \emph{feasible} in the general case (without the restrictions of Section~\ref{sec:np-fitboundingbox}), that is, a simple rectilinear polygon or polyline on the grid realizing a given angle sequence.

At first glance, one might think that \problemName{FitBoundingBox} directly implies \NP-hardness for the objective of minimizing the area of the bounding box. 
However, the question of whether an angle sequence~$S$ can be drawn within a rectangle of width~$W$ and height~$H$ does not directly translate to the question of whether~$S$ can be drawn in a rectangle of area~${W \cdot H}$.
For instance, suppose that~$S$ is a no-instance that we obtained by our reduction from \problemName{3-Partition}.
Draw the snails of~$S$ as tight as possible below each other in the order of their indices and connect them on the left side. 
Observe that such a drawing fits into a rectangle of width~${2\rho + \bigOh(1)}$ and height~${(\sum_{i=1}^{3m} (a_i+2)\rho + \bigOh(1))=(B+6)m\rho + \bigOh(m)}$ (the variables are defined as in Section~\ref{sec:np-fitboundingbox}). Hence,~$S$ fits into a rectangle of area even smaller than~${W\cdot H}$.

\paragraph{Outline of the proof.}
We define~$T$ by simultaneously constructing a \enquote{reference drawing} for the case that~$S$ is a yes-instance. 
It, roughly speaking, consists of two snail subsequences with~$S$ in between, where each snail is formed by ladder and spiral subsequences similar to Section~\ref{sec:np-fitboundingbox}. The notions \emph{spirals}, \emph{snails} and \emph{ladders} throughout this section refer only to the subsequences of~$T$ excluding~$S$, unless otherwise mentioned.
After defining the thresholds,
we use the reference drawing as a certificate in the first direction of the proof that 
a cheap drawing exists if~$S$ is a yes-instance. 
In the second direction of the proof, we show, for each objective, that if a drawing does not surpass the respective threshold~$\Upsilon$, then~$S$ is a yes-instance. 
For this, we first observe that, in any drawing of~$T$, certain subsequences (for instance, spirals) have certain lower bounds on the cost of drawing them.
We use these lower bounds 
to show that a drawing respecting the threshold~$\Upsilon$ has some special structure: If it doesn't, then some part of it is very expensive and, \emph{together} with the lower bounds on the other parts, the total cost is above the threshold; a contradiction. 
Generally speaking,~$T$ consists of two long spirals. 
Step by step, we show that spiral edges are not much longer than certain lower bounds 
and that spirals wind sufficiently enough 
in the \enquote{right} direction. This again will help us to observe that the spirals interleave until the inner-most level.
Together with the upper bounds on the spiral edges, we will see that~$S$ cannot leave the center of the spirals and is closed in a box of relatively small size, which implies that~$S$ is a yes-instance.

\paragraph{Definition of the instance~$T$.}

Recall that~$W$ and~$H$ are even.
Without loss of generality, 
we assume~${\min\{W,H\} > 5}$. 
Let
\[{\vUBcenter\xSpaceSE=\xSpaceSE 2(W+7)(H+7)}\] and let 
\[{\rojo\xSpaceSE=\xSpaceSE (\vUBcenter+12)^2}~.\] 
Finally, set~
\[{w=W+2\rojo+11}
\textrm{\quad\quad and \quad\quad}
{h=H+2\rojo+11}~.\] 

We define~$T$ constructively by giving a drawing of two polylines, called \emph{snails}, whose angle sequences together with~$S$ form~$T$. 

\begin{figure}
	\begin{subfigure}[b]{.3\textwidth}
		\centering
		\includegraphics[page=1]{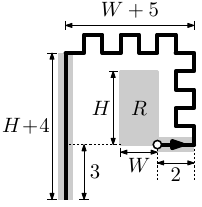}
		\caption{$\BigOutInnerLadder$ (bold) starts at the 
			lower-right corner of~$\rectR$. Its first and last edge 
			(highlighted) have lengths~$2$ and~${H+4}$, respectively. All other
			edges have unit length. Its bounding box has width~${W+5}$ and height~${H+5}$.}
		\label{fig:outSnail-inner}
	\end{subfigure}
	\hfill
	\begin{subfigure}[b]{.275\textwidth}
		\centering
		\includegraphics[page=2]{big-snail}
		\caption{$\BigOutSpiral$ (bold) starts at the endpoint 
			of~$\BigOutInnerLadder$ and winds around it with edge lengths 
			increasing in steps of~$2$. The first two edges (highlighted) have 
			lengths~${W+7}$ and~${H+7}$, respectively.}
		\label{fig:outSnail-spiral}
	\end{subfigure}
	\hfill
	\begin{subfigure}[b]{.34\textwidth}
		\centering
		\includegraphics[page=3]{big-snail}
		\caption{The last two spiral edges (dashed) have lengths~${w-4}$ and~${h-4}$, 
			respectively.~$\BigOutOuterLadder$ (bold) starts at the 
			endpoint of the spiral, and its bounding box has width~$w$ and 
			height~$h$. All but its first edge (highlighted) have unit length.}
		\label{fig:outSnail-outer}
	\end{subfigure}
	\caption{The construction of~$\BigOutSnail$.
	}
\end{figure}
We begin with the first snail that we call~$\BigOutSnail$. 
Place an axis-aligned rectangle~$\rectR$ of width~$W$ and height~$H$ in the grid. 
Starting at its lower-right corner, draw a ladder around it, as in Fig.~\ref{fig:outSnail-inner}, such that 
the first edge (horizontal) has length~$2$, the last edge (vertical) has length~${H+4}$, and all the other edges have unit length, and the bounding box of the ladder has width~${W+5}$ and height~${H+5}$. 
We call the ladder~$\BigOutInnerLadder$. 
Formally,
\[\BigOutInnerLadder \xSpaceASD=\xSpaceASD \LS (\LS\RS\RS\LS)^{\frac{H}{2}} \LS (\RS\LS\LS\RS)^{\frac{W+4}{2}} \LS~.\]
We continue the sequence by a left turn followed by a spiral, called~$\BigOutSpiral$, of~${2\rojo+1}$ left turns 
winding around the rectangle (and~$\BigOutInnerLadder$) in such a way that the 
first edge has length~${W+7}$, the second edge has length~${H+7}$, and 
the~${(i+2)}$\thSuffix edge is longer by exactly~$2$ than the~$i$\thSuffix edge; 
see Fig.~\ref{fig:outSnail-spiral}. 
Note that the spiral consists of~${\rojo+1}$ horizontal and~${\rojo+1}$ vertical edges. 
Thus, in our drawing, the last horizontal and vertical edges of~$\BigOutSpiral$ have the lengths~${W+7 + 2\rojo = w-4}$ and~${H+7+2\rojo=h-4}$, respectively.
Formally,
\[\BigOutSpiral \xSpaceASD=\xSpaceASD \LS^{2\rojo+1}~.\]
We finish the snail by a left turn and a following ladder, called~$\BigOutOuterLadder$. 
We draw the ladder around the spiral such that all edges but the first one have unit length and the bounding box of the ladder has width~$w$ and height~$h$. The length of the first edge is~${w-1}$; see Fig.~\ref{fig:outSnail-outer}.
Formally,
\[\BigOutOuterLadder \xSpaceASD=\xSpaceASD \LS (\RS\LS\LS\RS)^{\frac{h-1}{2}} \LS (\LS\RS\RS\LS)^{\frac{w-3}{2}} \LS \RS \LS ~.\]
Summarized, 
\begin{eqnarray*}
	\BigOutSnail&=& \BigOutInnerLadder ~\LS ~\BigOutSpiral ~ \LS ~ \BigOutOuterLadder~.
\end{eqnarray*}

\begin{figure}[t]
	\centering
	\includegraphics[page=4,scale=1]{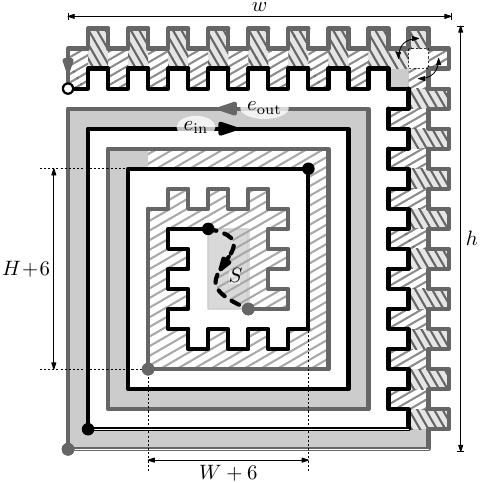}
	\caption[]%empty argument used as a hack as otherwise error message because of the enumerate environment inside the caption
	{The reference drawing of~$T$ (when~$S$ is a yes-instance). The snails~$\BigOutSnail$ (gray) and~$\BigInSnail$ (black) wind around~$\rectR$ (shaded rectangle in the center). 
		The endpoints of the ladders and spirals are depicted as nodes; the common endpoint of the snails is white. 
		The other two endpoints are connected via~$S$ (dashed curve) within~$\rectR$. 
		The bounding box containing both inner ladders has size~${(W+6)\times(H+6)}$.
		The edges of the outer ladders that are incident to the spirals (edges with white filling) are longer by~$3$ than~$\edgeIn$ and~$\edgeOut$, respectively.
		To ease the estimation of the area, the grid cells of the polygon are highlighted as follows: 
		\begin{inlinelistRoman}
			\item All grid cells within the bounding box of the first vertical edge of~$\BigOutSpiral$ and both inner ladders (size~${(W+7)\times(H+7)}$) are hatched.
			\item Almost all remaining grid cells surrounded by exactly two edges of the drawing are shaded in gray.
			\item with the exception of one grid cell (dashed white box), all other remaining grid cells are grouped into pairs that are hatched in one of two patterns (the pairs around the white box are indicated with arrows).
		\end{inlinelistRoman}
	}
	\label{fig:referenceDrawing}
\end{figure}

In a similar way, we define the second snail~$\BigInSnail$; see Fig.~\ref{fig:referenceDrawing}.
The biggest difference is that~$\BigInSnail$ winds in the other direction and ends at the upper-left corner of the rectangle~$\rectR$. 
In detail, the polyline of~$\BigInInnerLadder$ is a copy of~$\BigOutInnerLadder$ turned by~${180^\circ}$ with reversed orientation. 
Another difference is that the spiral~$\BigInSpiral$ contains only~${2\rojo-1}$ right turns (instead of~${2\rojo+1}$ turns). Thus, it consists of~$\rojo$ horizontal and~$\rojo$ vertical edges which also increase in lengths by steps of~$2$.  Therefore, in our drawing, the last horizontal and vertical edges of~$\BigInSpiral$ have lengths~${W+7 + 2(\rojo-1) = w-6}$ and~${H+7+2(\rojo-1)=h-6}$, respectively. Regarding~$\BigInOuterLadder$, 
it has width~${w-2}$, height~${h-3}$, and it starts and ends with a horizontal edge. 
Formally,
\begin{eqnarray*}
	\BigInSnail&=& \BigInOuterLadder  ~\RS~\BigInSpiral ~\RS~ \BigInInnerLadder~,\\
	\BigInOuterLadder&=& (\LS\RS\RS\LS)^{\frac{w-3}{2}} \RS (\RS\LS\LS\RS)^{\frac{h-5}{2}} \RS~, \\
	\BigInSpiral&=&\RS^{2\rojo-1}~,\\
	\BigInInnerLadder&=&\RS (\LS\RS\RS\LS)^{\frac{W+4}{2}} \RS (\RS\LS\LS\RS)^{\frac{H}{2}} \RS~.
\end{eqnarray*}

Finally, we complete our definition of~$T$ as follows:
\[T \xSpaceASD=\xSpaceASD  \BigInSnail ~ S ~ \BigOutSnail ~\LS ~. \]

Note that, if~$S$ is a yes-instance, then there exists the following drawing 
of~$T$: We draw~$\BigInSnail$ and~$\BigOutSnail$ as above and place~$S$ 
inside~$\rectR$ such that the first vertex of~$S$ touches the last vertex 
of~$\BigInSnail$ and the last vertex of~$S$ touches the first vertex 
of~$\BigOutSnail$ (in other words, the first and last edge of~$S$---which are 
horizontal---extend the first and last edge of~$\BigInSnail$ and~$\BigOutSnail$, 
respectively). Finally, we connect both snails on the outside by 
prolonging the last (vertical) edge of~$\BigOutSnail$ one unit to the bottom 
such that it touches the first vertex of~$\BigInSnail$; see 
Fig.~\ref{fig:referenceDrawing}. We call this drawing the \emph{reference drawing}.

Throughout this section, we say \emph{inner ladder} to refer 
to~$\BigInInnerLadder$ or~$\BigOutInnerLadder$, which are the ladders incident to~$S$, 
and \emph{outer ladder} to refer 
to~$\BigInOuterLadder$ or~$\BigOutOuterLadder$. 

\paragraph{Lower Bounds and Thresholds.}
Next, we provide lower bounds and thresholds 
on the cost of any feasible drawing of~$T$ 
that depend only on~$W$ and~$H$.
The thresholds will be defined on each of the three objectives.
They will be essential for our reduction:
There exists a drawing of~$T$ (in particular our reference drawing) that does not surpass the threshold of the respective objective if and only if~$S$ is a yes-instance.
In the reduction, we prove by contradiction that any 
drawing having the threshold as an upper bound has some specific properties.
Our proof will use that any drawing of~$T$ has a lower bound on the perimeter (that influences also the other objectives) that is very close to the threshold. We will see that if a drawing lacks a desired property, then it has to be much more expensive than its lower bound, and thus above the respective threshold.

We begin by providing lower bounds on the perimeter of any drawing of~$T$. We 
will first consider spiral edges, then whole spirals, and finally the ladders. 
We will see that, in the reference drawing, the respective parts meet the 
lower bound or are very close to them. (Generously, we will use the lower 
bound of~$0$ for the remaining part of~$T$, which is~$S$.) We will also give 
a lower bound on the area of the bounding box.

In the following, we use the same notation as in Section~\ref{sec:np-fitboundingbox} for the spirals and ladders of~$\BigInSnail$ and~$\BigOutSnail$.
Consider a spiral. Note that, in contrast to 
Section~\ref{sec:np-fitboundingbox},~${\len{\ey{i}}>\len{\ey{i-1}}}$ 
implies~${\len{\ex{i}}>\len{\ex{i-1}}}$ in the inner order.
Therefore, in this section, we redefine winding and
say that a spiral winds~$i$ times around the ladder defining the order 
if~${\len{\ey{i}}>\len{\ey{i-1}}}$. 
Note that Observation~\ref{obs:minDistance}, Lemma~\ref{lem:greaterEqualTwo}, 
Observation~\ref{obs:biggerThanLadder}, and Corollary~\ref{cor:spiralConfigs} 
hold also for the spirals of~$\BigInSnail$ and~$\BigOutSnail$. 

We begin with a definition similar to Definition~\ref{def:lowerBoundsA}.
\begin{definition}\label{def:lowerBoundsB}
	We define for every spiral edge~$e$ its \emph{\lowerValue} as 
	\begin{itemize}
		\item ${\lbEdge{e}=2j+X}$ if~$e$ is horizontal and 
		\item ${\lbEdge{e}=2j+Y}$ otherwise
	\end{itemize}
	where~$j$ is the level of~$e$ with respect to the inner order,~${X=W+5}$, and~${Y=H+4}$. 
	\begin{itemize}
		\item Let~$\lbSpirals$ denote the sum of the \lowerValues over all edges of both spirals.
		\item Let~${\lbOuterLadders=|\BigOutOuterLadder|+|\BigInOuterLadder|+\lbEdge{\edgeIn}+\lbEdge{\edgeOut}}$, 
		where~$\edgeIn$ and~$\edgeOut$ denote the first horizontal edge of~$\BigInSpiral$ 
		and~$\BigOutSpiral$, respectively, in the outer order. 
		\item Let~${\lowerBoundAreaBB= w \cdot h}$.
	\end{itemize}
\end{definition}

In the following, we observe that the \lowerValues defined in Definition~\ref{def:lowerBoundsB} are proper lower bounds for any feasible drawing.
First, we observe that~$X$ and~$Y$ correspond to the minimum width and height of the bounding box of the inner ladders, respectively. We also examine the width and height of the outer ladders.

\begin{figure}[t]
	\centering
	\begin{subfigure}[b]{.47\textwidth}
		\centering
		\includegraphics[page=5]{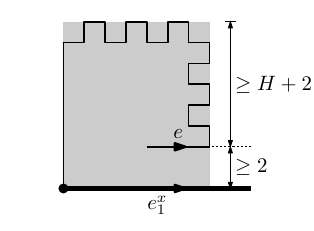}
		\caption{$\ex{1}$ (bold) leaves the bounding box}
		\label{fig:innerLadder-minBB-leaves}
	\end{subfigure}
	\hfil
	\begin{subfigure}[b]{.47\textwidth}
		\centering
		\includegraphics[page=6]{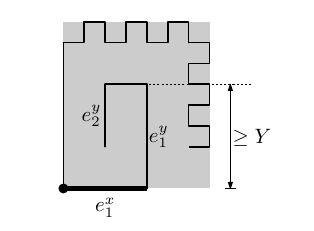}
		\caption{$\ex{1}$ (bold) stays completely inside}
		\label{fig:innerLadder-minBB-stays}
	\end{subfigure}
	\caption{In each of the two cases, the bounding box~$b$ (shaded) of~$\BigOutInnerLadder$ has height at least~${Y=H+4}$.}
	\label{fig:innerLadder-minBB}
\end{figure}

\begin{lemma}
	\label{lem:minBBLadders}	
	In any feasible drawing,~$\BigOutOuterLadder$ has width at least~$w$ and height at least~$h$, $\BigInOuterLadder$ has width at least~${w-2}$ and height at least~${h-3}$,
	and the bounding box of an inner ladder has minimum width~$X$ and minimum height~$Y$.
\end{lemma}
\begin{proof}
	A ladder consists of an~$x$-monotone and a~$y$-monotone part (that overlap).
	The width of an~$x$-monotone polyline is at least the number of its horizontal edges, 
	the height of a~$y$-monotone polyline is at least the number of its vertical edges.
	Hence, by the definition of the ladders, the first claim follows.
	
	The second claim follows only partially by this observation: The bounding box of an inner ladder has minimum width~${W+5=X}$ and minimum height~${H+2=Y-2}$. 
	We now show that the height is at least~$Y$. 
	Without loss of generality, consider~$\BigOutInnerLadder$, its bounding box~$b$, and its incident spiral in the inner order. The ladder starts with a right-oriented edge~$e$ and ends with a vertical edge that is incident to the right-oriented spiral edge~$\ex{1}$; see Fig.~\ref{fig:innerLadder-minBB-leaves}. 
	We have two cases:
	In the first case,~$\ex{1}$ leaves~$b$.
	Since the left endpoint of~$\ex{1}$ lies on the left edge of~$b$, its right endpoint has to be to the right of~$b$. Furthermore, the bottom edge of~$b$ is contained in~$\ex{1}$ since~$\ex{1}$ lies below the vertical ladder edge it is incident to.
	Thus,~$\ex{1}$ goes below~$e$. By Observation~\ref{obs:minDistance}, the vertical distance between~$e$ and~$\ex{1}$ is at least~$2$. 
	Note that the~$y$-monotone part of~$\BigOutInnerLadder$ starts at~$e$ and goes upward for at least~${H+2}$ units.
	Hence, the height of the bounding box of~$\BigOutInnerLadder$ is at least~${H+4=Y}$.
	
	In the second case,~${\ex{1} \in  b}$; see Fig.~\ref{fig:innerLadder-minBB-stays}.
	Then, also~${\ey{1}\in b}$ and, by monotonicity of the ladder,~${\len{\ey{1}}>\len{\ey{2}}}$.
	Recall that the level of~$\ey{1}$ is at least~$\rojo$ in the outer order and that the outer ladder has height at least~${h-3}$.
	Thus, Lemma~\ref{lem:greaterEqualTwo} and Observation~\ref{obs:biggerThanLadder} imply \[{\len{\ey{1}}\ge 2\rojo + (h-3) \ge Y}~.\]
	Hence, the height of~$b$ is at least~$Y$.
\end{proof}

The following lemma is a consequence of the lemma above. 

\begin{lemma}\label{lem:bbaLowB}
	For any feasible drawing, the area of its bounding box is
	at least~$\lowerBoundAreaBB$. 
\end{lemma}
\begin{proof} 
	By Lemma~\ref{lem:minBBLadders}, the bounding box~$b$ of~$\BigOutOuterLadder$
	has width and height at least~$w$ and~$h$, respectively. 
	Thus,~$b$ has area at least~${w \cdot h = \lowerBoundAreaBB}$.
	Since~$b$ is contained in the bounding box of the whole drawing, the claim follows.
\end{proof}	

By using the same arguments as in the proof of Lemma~\ref{lem:lbPerimeter}, we obtain the following lemma.
\begin{lemma}\label{lem:bigSpiralLB}
	In any feasible drawing, every spiral edge~$e$ has length at least~$\lbEdge{e}$ and the total perimeter of the spirals is at least~$\lbSpirals$.
\end{lemma}\qed

\begin{figure}[t]
	\centering 
	\includegraphics[page=7,scale=1]{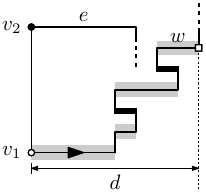}
	\caption{The distance~$d$ between 
		the spiral endpoint~$v_1$ (white node) and the right-most vertex~$w$ (white square) of the incident outer ladder (here~$\BigOutOuterLadder$) is~${d\ge \lbEdge{e}}$.
		The polyline~${v_1-w}$ has at most two more right-oriented edges (highlighted) than left-oriented ones (bold).}
	\label{fig:outerLadderLB} 
\end{figure}

{\begin{lemma}\label{lem:lowBOutLad}
		\sloppy In any feasible drawing, 
		the total perimeter of the two outer ladders is at least~$\lbOuterLadders$.
\end{lemma}}

\begin{proof}	
	Given a feasible drawing, consider a spiral and its incident outer ladder~$L$.
	A natural lower bound on~$\peri{L}$ is~${|L|+1}$ (as~$L$ consists of~${|L|+1}$ edges).
	However, this is not enough. 
	Therefore we show that some of the edges are longer than~$1$. We define the \emph{remainder} 
	of an edge~$e$ to be~${\len{e}-1}$ and we let~$r$ denote the total remainder of the edges of~$L$, that is,~${r=\peri{L}-|L|-1}$.
	In the following, we bound~$r$ from below.
	
	Let~$v_1$ and~$v_2$ denote the first two vertices (including the endpoint) of the spiral in the outer order, and let~$e$ denote the first horizontal spiral edge. 
	Furthermore, let~$w$ denote a right-most vertex of~$L$ and let~$d$ 
	denote the horizontal distance between~$w$ and~$v_1$ (and~$v_2$); see Fig.~\ref{fig:outerLadderLB}. 
	Suppose~${d\le \len{e}}$. 
	Then, by monotonicity of~$L$,~$L$ lies completely inside the bounding box 
	of~$v_1$,~$v_2$, and~$w$. 
	However, the width of this bounding box is~${d\le\lbEdge{e}\le 2(\rojo+1) + X = w-4}$ 
	and the minimum width of~$L$ is at least~${w-2}$; a contradiction. 
	Hence, we have~${d\ge \lbEdge{e}+1}$.
	
	Consider the part of~$L$ between~$v_1$ and~$w$, and orient the edges of this polyline such that it is directed from~$v_1$ to~$w$. 
	Observe that the polyline is~$y$-monotone and that it has at most 
	two right-oriented edges more than left-oriented edges. 
	Since its width is~$d$, the total length of its right-oriented edges is bigger by~$d$ 
	than the total length of its left-oriented edges. 
	Hence, the total remainder of the right-oriented edges is at least~${d-2}$.
	Thus,~${r\ge d-2 \ge \lbEdge{e}-1}$, and~${\peri{L}=r+|L|+1\ge |L|+\lbEdge{e}}$.
	We repeat the proof above for the other spiral and its outer ladder and the claim follows.
\end{proof}

\begin{definition}\label{def:threeThresholds}
	We define the following \emph{thresholds} for each objective:
	\begin{itemize}
		\item ${\Upsilon_p = \lbOuterLadders + \lbSpirals + \vUBcenter  + 2\rojo  + 12}$  for minimizing the perimeter of the drawing, 
		\item ${\Upsilon_a = \Upsilon_p/2 - 1}$ for minimizing the area of the drawing, and
		\item ${\Upsilon_b = \lowerBoundAreaBB}$ for minimizing the area of the bounding box of the drawing.
	\end{itemize}	
\end{definition}

We use the thresholds for our reduction. 
\begin{theorem}\label{thm:NPreduction}
	For each of the three objectives it holds: There is a drawing of~$T$ that does not surpass the threshold (as defined in Definition~\ref{def:threeThresholds}) of the given objective if and only if~$S$ is a yes-instance.
\end{theorem}

We first show that if~$S$ is a yes-instance, then there is a drawing of~$T$ that does not surpass the threshold of the respective objective. 
Consider the reference drawing and recall that we drew~$S$ inside the empty~${(W\times H)}$-rectangle~$\rectR$ and connected it to the two snails accordingly.
We now show that the reference drawing respects all three thresholds. 

\paragraph{Perimeter.}
First, consider~$S$ and the inner ladders.
Given that~$S$ and the inner ladders lie in a~${((W+6) \times (H+6))}$-rectangle (see Fig.~\ref{fig:referenceDrawing}), the total perimeter of~$S$ and the inner ladders is bounded from above by 
\[{2(W+7)(H+7) \xSpaceSE = \xSpaceSE  \vUBcenter}~.\]  
Next, consider the spirals.
Observe that in the reference drawing,
each horizontal spiral edge~$e$ has length~$\lbEdge{e}$,
and
each vertical spiral edge~$e$ has length~${\lbEdge{e}+1}$. 
Recall that~$\BigInSpiral$ has~$\rojo$ vertical edges and~$\BigOutSpiral$ has~${\rojo+1}$ vertical edges.
Thus, the total perimeter of the spirals is~${\lbSpirals+2\rojo+1}$. 
Finally, consider the outer ladders. 
with the exception of the edges incident to the spirals and the last edge of~$\BigOutOuterLadder$ (which has length~$2$), all edges of the outer ladders have unit length. 
The two edges incident to the spirals are exactly~$3$ units longer than the 
first horizontal edge of the respective incident spiral in the outer order. 
Hence, using the notation of Definition~\ref{def:lowerBoundsB}, the two edges have total length~${\lbEdge{\edgeIn} + \lbEdge{\edgeOut} + 6}$.
Thus, the total perimeter of the outer ladders 
is 
\[|\BigInOuterLadder| + |\BigOutOuterLadder| + \lbEdge{\edgeIn} + \lbEdge{\edgeOut} + 6 + 1 
~\le~  \lbOuterLadders+7~. \]
Summing up, the total perimeter of the reference drawing is at most  
\[\lbOuterLadders + \lbSpirals + \vUBcenter + 2\rojo + 8 \xSpaceSE<\xSpaceSE  \Upsilon_p ~.\] 

\paragraph{Area.}
Regarding the area, we subdivide the grid cells of the reference drawing into three parts: 
The first part is the intersection of our polygon with the~${((W+7) \times (H+7))}$-rectangle containing~$S$, the inner ladders, and the first vertical edge of~$\BigOutSpiral$ in the inner order.
Hence, the intersection contains at most
\[{(W+7) \cdot (H+7) \xSpaceSE=\xSpaceSE  \vUBcenter/2}\] grid cells. 
The second part consists of almost all grid cells outside this 
rectangle touching exactly two edges of 
the polyline~$P$ that realizes the spirals and the outer ladders. 
The third part consists of all the remaining grid cells. 
with the exception of one grid cell, we can group the grid cells of the third part into pairs that touch four or five edges of~$P$; see Fig.~\ref{fig:referenceDrawing}.
Hence, with the exception of one grid cell, each grid cell of the second and third part touches at least two edges of~$P$ on average. 
Since each unit-line segment of~$P$ is touched by exactly one grid cell,
the number of grid cells belonging to the second and third part is at most~${1+\peri{P}/2 \le  (\lbOuterLadders+\lbSpirals+2\rojo+10)/2}$. 
Hence, the total area of the reference drawing is at most \[\frac{\lbOuterLadders + \lbSpirals+ \vUBcenter+  2\rojo+10}{2} ~\le~ \frac{\Upsilon_p - 2}{2} ~=~ \Upsilon_a ~.\]

\paragraph{Bounding Box.}
Regarding the bounding box of the reference drawing, note that it is identical to the bounding box of~$\BigOutOuterLadder$. Following the proof of Lemma~\ref{lem:bbaLowB}, the area of the bounding box of the drawing is~${\lowerBoundAreaBB=\Upsilon_b}$.

\bigskip

Now, to prove the other direction of Theorem~\ref{thm:NPreduction}, assume that~$T$ can be drawn such that (at least) one of the three thresholds of Definition~\ref{def:threeThresholds} is not surpassed. 
We show that this fact implies that~$T$ is a yes-instance.
Until the remainder of this section, we fix such a drawing that respects a threshold and refer to it as \emph{our drawing}.
We begin by making a helpful observation that will allow us to focus only on the perimeter and the bounding box of our drawing: 
\begin{lemma}\label{lem:area-implies-per}
	If the area of our drawing is at most~$\Upsilon_a$, then the perimeter is at most~$\Upsilon_p$.
\end{lemma} 
\begin{proof}
	The claim follows from~${\Upsilon_a=\Upsilon_p/2-1}$ and the following observation that we prove below:
	For any simple rectilinear polygon~$P$ on the grid,
	\[{\area{P}\xSpaceSE\ge\xSpaceSE  \peri{P}/2 -1}~.\]
	
	We scale~$P$ by a factor of~$2$ and obtain a new polygon~${P'}$. 
	In~${P'}$, there are~${\#\LS}$ grid cells touching exactly two edge
	segments (which happens only at~$\LS$ vertices),~${\peri{P'}-2\#\LS}$
	grid cells touching exactly one edge segment, and at least~${\#\RS}$
	grid cells touching no edges (every~$\RS$ vertex is exclusively
	incident to one such grid cell due to the simplicity and upscaling of~$P$). Thus,~${\area{P'}\ge \#\LS +  \peri{P'}-2\#\LS + \#\RS = \peri{P'} - 4}$
	using~${\#\LS=\#\RS+4}$. The claim follows by
	substituting~${\area{P'}=4\cdot\area{P}}$
	and~${\peri{P'}=2\cdot\peri{P}}$.
\end{proof}

Our assumption that at least one of the three thresholds of Definition~\ref{def:threeThresholds} is not surpassed has a number of implications that we consider one by one. 

\begin{figure}[t]
	\begin{subfigure}[b]{.30\textwidth}
		\centering
		\includegraphics[page=8]{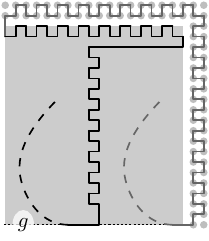}
		\caption{If~$\BigInOuterLadder$ visits the grid line~$g$ (dashed), then~$\Gamma$ consists of two disconnected polylines; a contradiction as the whole drawing is a polygon.}
		\label{fig:BBminSpirals-diconnected}
	\end{subfigure}
	\hfill
	\begin{subfigure}[b]{.35\textwidth}
		\centering
		\includegraphics[page=9]{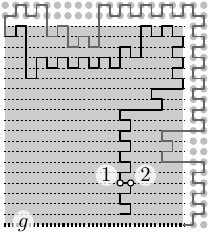}
		\caption{Within~${\bboxGeneral{\Gamma}}$, the outer~lad\-ders occupy at least two grid points from every horizontal (dashed) and vertical (not depicted) grid line with the exception of~$g$ (bold dashed).}
		\label{fig:BBminSpirals-boundedSpace}
	\end{subfigure}
	\hfill
	\begin{subfigure}[b]{.22\textwidth}
		\centering
		\includegraphics[page=10]{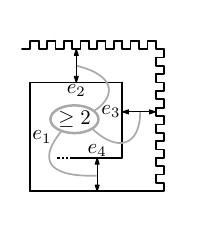}
		\caption{Given the edge orientations, the spiral edges~$e_2$,~$e_3$, and~$e_4$ have distance at least~$2$ to the outer ladder.}
		\label{fig:BBminSpirals-distance}
	\end{subfigure}
	\caption{If the bounding box of the drawing has height~$h$ and width~$w$, then
		every grid point (gray nodes) at distance at most~$1$ to the top or right border of the bounding box is visited only by~$\BigOutOuterLadder$ (gray). 
		This forces	the remaining part~$\Gamma$ of the drawing to lie in the box (shaded) of width~${w-2}$ and height~${h-2}$, which has several implications on~$\BigInOuterLadder$ (black) and the spirals, for instance, forcing~~$\BigInOuterLadder$ to have its minimum height~${h-3}$.}
\end{figure}

\begin{lemma}\label{lem:bbA-windingAndSmallEdgeLen}
	If the bounding box of the drawing has area at most~$\Upsilon_b$, then the spirals wind at least~$\rojo$ times around their inner ladders and 
	for every spiral edge~$e$,~${\len{e}\le \lbEdge{e} + 1}$. 
\end{lemma}
\begin{proof}
	Consider~$\BigOutOuterLadder$. 
	By Lemma~\ref{lem:minBBLadders}, the area of its bounding box is at least \[{w \cdot h = \lowerBoundAreaBB=\Upsilon_b}~.\]
	Thus, the bounding box of~$\BigOutOuterLadder$ is exactly the bounding box of the whole drawing (see Lemma~\ref{lem:bbaLowB}).
	Let~$\Gamma$ denote the part of the drawing that excludes the edges of~$\BigOutOuterLadder$. 
	Recall that~$\BigOutOuterLadder$ consists of an~$x$-monotone and a~$y$-monotone part of minimum width~$w$ and height~$h$, respectively.
	Given the orientation of the first and the last edge of~$\BigOutOuterLadder$,~$\Gamma$ 
	has to lie entirely to the bottom of the~$x$-monotone part and to the left of the~$y$-monotone part.
	Observe that all horizontal edges of the~$x$-monotone part as well as the vertical edges of the~$y$-monotone part have unit length.
	Thus, every grid point with distance at most~$1$ to the top or right border of the bounding box of the drawing either belongs to~$\BigOutOuterLadder$, or is not visited by the drawing; see Fig.~\ref{fig:BBminSpirals-boundedSpace}.
	Consequently, the bounding box of~$\Gamma$, which includes~$\BigOutSpiral$ as well as~$\BigInSpiral$ and~$\BigInOuterLadder$, 
	has width and height at most~${w-2}$ and~${h-2}$.
	
	Consider~$\BigInOuterLadder$.
	It cannot visit any grid point on the bottom-most grid line~$g$, as otherwise it would separate~$\BigInSpiral$ from~$\BigOutSpiral$ (see Fig.~\ref{fig:BBminSpirals-diconnected}); a contradiction as~$\Gamma$ is a (connected) polyline. 
	Thus,~$\BigInOuterLadder$ lies in a bounding box of width~${w-2}$ and height~${h-3}$.
	Given that its~$x$-monotone part has width at least~${w-2}$ and its~$y$-monotone part has height at least~${h-3}$, all horizontal line segments of the~$x$-monotone part and all vertical line segments of the~$y$-monotone part are of unit length.
	Therefore, every (vertical and horizontal) grid line that goes through~${\bboxGeneral{\Gamma}}$---with the exception of~$g$---contains at least
	two grid points within~${\bboxGeneral{\Gamma}}$ that are covered by the outer ladders\footnote{We consider both outer ladders as~$\BigInOuterLadder$ possibly visits only one grid point of the left-most vertical line; see Fig.~\ref{fig:BBminSpirals-boundedSpace}.}.
	Consequently, every vertical and horizontal grid line 
	---with the exception of~$g$---contains within~${\bboxGeneral{\Gamma}}$ at most~${h-4}$ and~${w-4}$ free grid points, respectively; see Fig.~\ref{fig:BBminSpirals-boundedSpace}. 
	
	For the remainder of the proof, consider any of the two spirals. Let~${e_1, \dots, e_4}$ denote the first four spiral edges in the outer order.
	Recall that the spiral is contained in~${\bboxGeneral{\Gamma}}$ and observe that~$e_2$ lies above~$g$. 
	Consequently, given the number of free grid points,~${\len{e_1}\le h-4}$ and~${\len{e_2}\le w-4}$. 
	For~$\BigInSpiral$, we even have sharper upper bounds. 
	Observe that~$e_1$ starts on a grid point above~$g$. Given the orientation of~$e_2$ and Observation~\ref{obs:minDistance},~$e_1$ ends two units below the~$x$-monotone part of~$\BigInOuterLadder$; see Fig.~\ref{fig:BBminSpirals-distance}.
	Thus,~${\len{e_1}\le h-6}$. By a similar argument,~${\len{e_2}\le w-5}$.
	
	Now, we show that~${\len{e_i}\ge\len{e_{i+2}}+2}$ holds for~${i \in\{1,2\}}$. 
	By our previous observations and by the winding direction of the spiral, the spiral is contained in the bounding box of its outer ladder.
	Since the outer ladder is connected to~$e_1$, its~$y$-monotone part contains a left-oriented line segment below~$e_4$. Thus, by monotonicity and by Observation~\ref{obs:minDistance},~$e_4$ has to lie at least two units above the bottom endpoint of~$e_1$; see Fig.~\ref{fig:BBminSpirals-distance}. Given that the top endpoints of~$e_1$ and~$e_3$ have the same~$y$-coordinate, the claim holds for~${i=1}$, and, by a similar argument, for~${i=2}$. 
	Given Corollary~\ref{cor:spiralConfigs}, the first claim of the lemma follows.
	
	Regarding the second claim, suppose that, for a spiral edge~$e$,~${\len{e}\ge\lbEdge{e}+2}$.
	Then, by Corollary~\ref{cor:spiralConfigs}, Lemma~\ref{lem:greaterEqualTwo} and Definition~\ref{def:lowerBoundsB}, we have in a cascading manner~${\len{e_1}\ge \lbEdge{e_1}+2}$ if~$e_1$ is parallel to~$e$, and~${\len{e_2}\ge\lbEdge{e_2}+2}$ otherwise.
	Thus, if our spiral is~$\BigOutSpiral$, then we have~${\len{e_1}\ge h-3}$ or~${\len{e_2}\ge w-2}$. 
	If our spiral is~$\BigInSpiral$, then~${\len{e_1}\ge h-5}$ or~${\len{e_2}\ge w-4}$.
	In either case, we have a contradiction to our upper bounds on the spiral edges.
\end{proof}

\begin{lemma}\label{lem:windingAlot}
	The spirals wind at least~${\rojo-\sqrt{\rojo}}$ times around their inner ladders. 
\end{lemma}
\begin{proof}
	
	By Lemma~\ref{lem:area-implies-per} and Lemma~\ref{lem:bbA-windingAndSmallEdgeLen}, 
	we have to consider only the case that the total perimeter is at most~$\Upsilon_p$.
	Consider any of the two spirals. If the	
	spiral winds only around the inner ladder, then we are done. 
	Otherwise, the spiral winds~${\alpha\ge 1}$ times around its outer ladder (see Corollary~\ref{cor:spiralConfigs}). 
	Consider any vertical spiral edge~$e$ of a level~$i$ with~${1\le i\le \alpha}$ in the outer order. 
	Note that its level is at most~${\rojo-i+2}$ in the inner order, hence,~${\lbEdge{e} \le 2(\rojo-i+2) + Y}$ by Definition~\ref{def:lowerBoundsB}.
	Recall that the bounding box of the outer ladder has height at least \[{h-3 = 2\rojo+Y+4}\] (see Lemma~\ref{lem:minBBLadders}). 
	Thus, by Observation~\ref{obs:biggerThanLadder}, we have
	\[\len{e} ~\ge~ 2i + (2\rojo + Y + 4) ~\ge~ 2(\rojo-i+2) + Y + 4i ~\ge~ \lbEdge{e}+4i~.\]
	Consequently, the perimeter of the drawing is at least
	\begin{alignat*}{2}
	&&~&\lbOuterLadders + \lbSpirals+ \sum_{i=1}^{\alpha} 4i 
	\\&\ge&& \lbOuterLadders + \lbSpirals +  2\alpha(\alpha+1)~. 
	\end{alignat*}
	Thus,~${\alpha\le \sqrt{\rojo}}$, as otherwise~${2\alpha(\alpha+1)> \vUBcenter+2\rojo+12}$ (here, recall that we have set~${\rojo=(\vUBcenter+12)^2}$) and the perimeter is greater than~$\Upsilon_p$; a contradiction. 
	We conclude by Corollary~\ref{cor:spiralConfigs} that the spiral winds at least
	\[{\rojo-\alpha \xSpaceSE\ge\xSpaceSE  \rojo-\sqrt{\rojo}}\] times around the inner ladder.
\end{proof}

\begin{lemma}\label{lem:edgeLenSmall}
	For every spiral edge~$e$ of level at most~${\sqrt{\rojo}}$ with respect to the inner order, we have~${\len{e}\le \lbEdge{e} + 2}$.
\end{lemma}
\begin{proof}
	By Lemma~\ref{lem:area-implies-per} and Lemma~\ref{lem:bbA-windingAndSmallEdgeLen}, we have to consider only the case that the total perimeter is at most~$\Upsilon_p$.
	Suppose that there is a horizontal edge~$e$ of a level~${j\le \sqrt{\rojo}}$ for which~${\len{e}\ge \lbEdge{e} + 3}$ holds.
	Then, by Lemma~\ref{lem:windingAlot}, Definition~\ref{def:lowerBoundsB}, and 
	Lemma~\ref{lem:greaterEqualTwo}, we also have~${\len{g}\ge \lbEdge{g} + 3}$ for 
	every horizontal edge~$g$ of the same spiral of a level between~$j$ and~${\rojo-\sqrt{\rojo}}$.
	Hence, the total perimeter of the drawing is at least
	\begin{alignat*}{2}
	&&~&\lbOuterLadders + \lbSpirals+ 3(\rojo-2\sqrt{\rojo}) 
	\\&>&& \lbOuterLadders + \lbSpirals +  \vUBcenter+2\rojo+12
	\\&=&& \Upsilon_p~; 
	\end{alignat*}
	a contradiction to the upper bound~$\Upsilon_p$. 
	In a similar way, we get a contradiction if~$e$ is vertical. 
\end{proof}

Now we will see that the spirals interleave until the first level (with respect to the inner ladders). 
Let~${v_1, \dots, v_{2\rojo+1}}$ be the vertices (including the endpoints) 
and let~${e_1, \dots, e_{2\rojo}}$ be the edges of~$\BigInSpiral$ in the inner order. 
Similarly, let~${w_1, \dots, w_{2\rojo+3}}$ be the vertices (including the endpoints) 
and let~${f_1, \dots, f_{2\rojo+2}}$ be the edges of~$\BigOutSpiral$ in the inner order.
For~${1 \le i <  2\sqrt{\rojo}}$, we define~$\inbboxSpiral{i}$ as the bounding box of~$e_i$ and~${e_{i+1}}$, 
and~$\outbboxSpiral{i}$ as the bounding box of~$f_i$ and~${f_{i+1}}$. 

\begin{lemma}\label{lem:interweavement}
	For~${1\le i < 2 \sqrt{\rojo}}$,~$v_i$ lies in the interior of~$\outbboxSpiral{i}$ and~$w_i$ lies in the interior of~$\inbboxSpiral{i}$. 
\end{lemma}
\begin{proof}
	We show the lemma by induction in two steps. First, we prove the claim for~${i=2 \sqrt{\rojo} - 1}$, and then, by induction, for~${1\le i< 2 \sqrt{\rojo} - 1}$.
	
	Let~${i=2 \sqrt{\rojo} - 1}$. 
	We begin by proving the following observation that will lead us to the first claim:
	The \emph{interiors} of~$\inbboxSpiral{i}$ and~$\outbboxSpiral{i}$ intersect. 
	Recall that both spirals are connected to each other by the polyline realizing~$S$ and the inner ladders. 
	If~$\inbboxSpiral{i}$ and~$\outbboxSpiral{i}$ were interior-disjoint, then
	the polyline, starting inside~$\inbboxSpiral{1}$, had to leave~$\inbboxSpiral{i}$ before entering~$\outbboxSpiral{i}$. 
	However, such a polyline requires\footnote{Proof sketch: The polyline goes through interior-disjoint regions of type~${\inbboxSpiral{i}\setminus\inbboxSpiral{i-1}}$ and in order to visit 
		three consecutive such regions, it needs a separate vertex inside the interior 
		of each of the three regions.}~$i$ vertices just for leaving~$\inbboxSpiral{i}$,
	which is more than the number of vertices provided by~$S$ and the two inner 
	ladders\footnote{Without loss of generality,~$S$ has at most~${(W+1)(H+1)}$ 
		vertices and the inner ladders have at most~${4X+4Y}$ vertices in total. 
		Since~${i\ge \sqrt{\rojo} > \vUBcenter > (W+1)(H+1) + 4X+4Y}$,~$i$ is greater
		than the number of vertices.}; a contradiction. 
	
	Now, suppose that the claim is violated by~$v_i$ not being in the interior of~$\outbboxSpiral{i}$. 
	To ease the description, we temporarily rotate the drawing (if needed) such that~$e_i$ is a right oriented edge.
	Since the interiors of the two bounding boxes intersect and given our assumption,~$v_i$ lies above~$w_{i+3}$ and to the right of~$w_i$ and, consequently, also to the right of~$w_{i+3}$ (note that we have~${\len{f_{i+2}} > \len{f_i}}$ by Lemma~\ref{lem:windingAlot}); see Fig.~\ref{fig:wiNotInside}.
	In particular,~${w_{i+3}}$ lies in~$\inbboxSpiral{i}$.
	Observe that the edge~${f_{i+3}}$ starting at~${w_{i+3}}$ cannot leave~$\inbboxSpiral{i}$ and has distance at least~$1$ to the border of~$\inbboxSpiral{i}$. 
	Also note that the levels of~${e_{i+3}}$ and~${e_{i+1}}$ differ by one.
	Thus, the border edge~${e_{i+1}}$ of~$\inbboxSpiral{i}$ has length
	\begin{alignat*}{2}
	\len{e_{i+1}} &\ge&~& \len{f_{i+3}} + 2
	\\&\ge&& \lbEdge{f_{i+3}} + 2
	\\&=&& \lbEdge{e_{i+3}} + 2
	\\&=&& \lbEdge{e_{i+1}} + 4~.
	\end{alignat*}
	Since the level of~${e_{i+1}}$ is~${\lceil (i+1)/2 \rceil = \sqrt{\rojo}}$, 
	the inequality contradicts Lemma~\ref{lem:edgeLenSmall}.
	In a similar way, we show the case for~$v_i$ not being in the interior of~$\inbboxSpiral{i}$.
	Thus, our claim holds for~${i=2 \sqrt{\rojo} - 1}$; see Fig.~\ref{fig:wiInside}.
	\begin{figure}[t]
		\centering  
		\begin{subfigure}[t]{.5\textwidth}
			\centering
			\includegraphics[page=1]{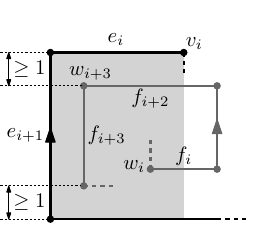}
			\caption{If~${v_i\not\in\outbboxSpiral{i}}$, then~${w_{i+3}\in\inbboxSpiral{i}}$.}
			\label{fig:wiNotInside}
		\end{subfigure}
		\hfill
		\begin{subfigure}[t]{.46\textwidth}
			\centering
			\includegraphics[page=2]{interweavement2}
			\caption{The claim:~${v_i\in\outbboxSpiral{i}}$ and~${w_i\in\inbboxSpiral{i}}$.}
			\label{fig:wiInside}
		\end{subfigure}
		
		\caption{The bounding boxes~$\inbboxSpiral{i}$ (shaded) and~$\outbboxSpiral{i}$ intersect (for~${i=2\sqrt{\rojo} - 1}$).}
	\end{figure}
	
	Now, assume that our claim holds for an~$i$ with~${2 \le i\le 2 \sqrt{\rojo} - 1 }$.
	Temporarily rotate the drawing (if needed) such that~${e_{i-1}}$ and~${f_{i-1}}$ are vertical 
	edges facing downwards; see Fig.~\ref{fig:wiInside}.
	Consider the bounding boxes~${\outbboxSpiral{i-1}}$ and~${\inbboxSpiral{i-1}}$. 
	The vertex~${w_{i-1}}$ lies in the interior of~${\inbboxSpiral{i-1}}$ if and only if~${v_{i-1}}$ lies below the horizontal line through~${w_{i-1}}$.
	Hence, if the induction hypothesis does not hold for~${i-1}$, then~${v_{i-1}}$ does not lie below~${w_{i-1}}$ and, thus, the heights of both~$\inbboxSpiral{i}$ and~$\outbboxSpiral{i}$ are at least~${\len{e_{i-1}}+\len{f_{i-1}}}$. Hence,~${\len{e_{i+1}}\ge \len{e_{i-1}}+\len{f_{i-1}}}$. 
	Therefore, using~${\lbEdge{f_{i+1}}\ge 5}$ (which holds as every spiral edge is longer than~${\min\{W,H\}\ge 5}$), we have
	\begin{alignat*}{2}
	\len{e_{i+1}} &\ge&~& \len{e_{i-1}}+\len{f_{i-1}}
	\\&\ge&& \lbEdge{e_{i-1}} + 5
	\\&=&& \lbEdge{e_{i+1}} + 3~.
	\end{alignat*}
	But this inequality contradicts Lemma~\ref{lem:edgeLenSmall}.
\end{proof}

\begin{corollary}\label{cor:edgeLenVerySmall}
	
	For every spiral edge~$e$ of level at most~${\sqrt{\rojo}}$ with respect to the inner order, we have~${\len{e}\le \lbEdge{e} + 1}$.
\end{corollary}
\begin{proof}
	By Lemmas~\ref{lem:area-implies-per} and~\ref{lem:bbA-windingAndSmallEdgeLen}, 
	we have to consider only the case that the total perimeter is at most~$\Upsilon_p$.
	Suppose that the claim is violated by an edge~$e_j$ of~$\BigInSpiral$ (the argument is similar for~$\BigOutSpiral$). 
	Thus,~${\len{e_j}\ge \lbEdge{e_j} + 2}$. Recall that~${e_j=(v_j,v_{j+1})}$.
	By Lemma~\ref{lem:interweavement},~$v_j$ lies in the interior of~$\outbboxSpiral{j}$ and~${v_{j+1}}$ lies in the interior of~${\outbboxSpiral{j+1}}$.
	Since~${\outbboxSpiral{j}\subset\outbboxSpiral{j+1}}$ (as, by Lemma~\ref{lem:windingAlot}, we have~${\len{f_{j+2}}>\len{f_j}}$),~$e_j$ lies in the interior of~${\outbboxSpiral{j+1}}$
	and both its endpoints have distance at least~$1$ to the border of~${\outbboxSpiral{j+1}}$. 
	Note that the border edge~${f_{j+2}}$ of~${\outbboxSpiral{j+1}}$ and~$e_j$ are parallel and the level of~${f_{j+2}}$ is one more than that of~$e_j$. 
	Consequently,
	\[\len{f_{j+2}} ~\ge~ \len{e_j}+2 ~\ge~ \lbEdge{e_j}+4 ~=~ \lbEdge{f_{j+2}}+2~.\]
	Recall that the level of~$e_j$ is at most~${\sqrt{\rojo}}$. 
	Hence, as in the proof of Lemma~\ref{lem:edgeLenSmall},
	consider any edge~$g$ (of any of the two spirals) that is
	parallel to~$e_j$ and of a level between~${\sqrt{\rojo}+1}$ and~${\rojo-\sqrt{\rojo}}$.
	For such an edge~$g$, we have~${\len{g}\ge \lbEdge{g} + 2}$.
	Then, however, the total perimeter of the drawing is at least
	\begin{alignat*}{2}
	&&~&\lbOuterLadders +\lbSpirals+ 2\cdot 2(\rojo-2\sqrt{\rojo}) 
	\\&>&&\lbOuterLadders + \lbSpirals +  \vUBcenter+2\rojo+13
	\\&=&& \Upsilon_p~; 
	\end{alignat*}
	a contradiction to the upper bound~$\Upsilon_p$. 
\end{proof}

\begin{figure}
	\centering
	\includegraphics[page=1]{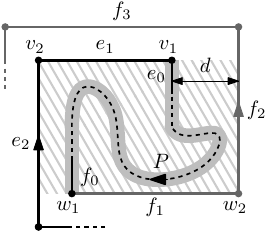}
	\caption{The inner-most levels of the spirals. Their endpoints~$v_1$ and~$w_1$ are connected by a highlighted polyline~$P$ realizing~$S$ and the inner ladders. The polyline has to lie in the bounding box of~$v_2$ and~$w_2$ (hashed area).}
	\label{fig:closedS}
\end{figure} 
\begin{lemma}
	The polyline realizing~$S$ lies completely inside a rectangle of width~${W+10}$ and height~${H+7}$.
\end{lemma}
\begin{proof}
	Let~$e_0$ and~$f_0$ be the vertical edges incident to~$v_1$ and~$w_1$, respectively.
	We claim that the polyline~$P$ connecting~$v_1$ to~$w_1$ (it realizes the inner 
	ladders and~$S$) lies completely in the bounding box of~$v_2$ and~$w_2$; see 
	Fig.~\ref{fig:closedS}. Note that as a consequence of Lemma~\ref{lem:interweavement},~$e_1$ 
	lies in the bounding box of~$f_2$ and~$f_3$, and~$f_1$ lies in the bounding box 
	of~$e_2$ and~$e_3$. Hence, the relative positions of the elements are as depicted
	in Fig.~\ref{fig:closedS}:~$v_2$ lies to the top-left of~$w_2$,~$f_2$ lies to 
	the right of~$e_0$, and~$e_2$ lies to the left of~$f_0$.
	
	First, observe that~$P$ can leave the box only between~$e_0$ and~$f_2$ and between~$e_2$ and~$f_0$.
	Suppose that it leaves the box between~$e_0$ and~$f_2$; the other case is similar.
	Thus,~$P$ contains a vertical line segment between~$e_0$ and~$f_2$ oriented in the same direction as~$f_2$ (to the top).
	Hence, by Observation~\ref{obs:minDistance}, the distance between the vertical line segment and~$f_2$ is at least~$2$. Consequently, the distance~$d$ between~$e_0$ and~$f_2$ is at least~$3$.
	However, given that~$v_2$ is contained in the interior of the bounding box of~$f_2$ and~$f_3$ (Lemma~\ref{lem:interweavement}), 
	we have 
	\[{\len{f_3}~\ge~\len{e_1}+d+1 ~\ge~ \lbEdge{e_1} + 4  ~=~ \lbEdge{f_1} + 4 ~=~ \lbEdge{f_3}+2}~.\]
	This contradicts Corollary~\ref{cor:edgeLenVerySmall}.  
	
	Thus,~$S$ lies completely in the bounding box of~$v_2$ and~$w_2$, which itself is contained in the bounding box of~$f_2$ and~$f_3$. By Corollary~\ref{cor:edgeLenVerySmall}, the width of the box is at most
	\[{\lbEdge{f_3}+1 ~=~ X + 5 ~=~ W + 10}\]
	and the height is at most
	\[{\lbEdge{f_2}+1 ~=~ Y+3 ~=~ H +
		7}~.\tag*{\qedhere}\]
\end{proof}

Hence,~$S$ can be drawn within a~${((W+10)\times (H+10))}$-rectangle such that the first and last edge of~$S$ are horizontal and such that~$S$ can be extended to a simple polygon (given its embedding in~$T$). Hence,~${\langle S, W, H \rangle}$ is a yes-instance. 
This conclusion finishes the second direction of our proof of Theorem~\ref{thm:NPreduction}. 

\section{The Monotone Case: Minimum Area}\label{sec:area-algo}

In this section, we show how to compute, for a monotone angle
sequence, a polygon of minimum bounding box and of minimum area. 
We start with the simple~$xy$-monotone case and then consider the more
general~$x$-monotone case. 

\subsection{The \texorpdfstring{${xy}$}{xy}-Monotone Case}\label{subsec:xymonarea}
An~$xy$-monotone polygon has four \emph{extreme edges}; its leftmost and rightmost vertical edge, and its topmost and bottommost horizontal edge.
Two consecutive extreme edges are connected by a (possible empty)~$xy$-monotone chain that we will call a \emph{stair}.
Starting at the top extreme edge, we let~$\TL$,~$\BL$,~$\BR$, and~$\TR$ denote the four stairs in ccw order; 
see Fig.~\ref{fig:xymonExampleB}.
We say that an
angle sequence consists of~$k$ nonempty \emph{stair sequences} if any~$xy$-monotone polygon that
realizes it consists of~$k$ nonempty stairs; we also call it a 
\emph{$k$-stair sequence}.
The extreme edges correspond to the exactly 
four~${\LS\LS}$-sequences in an~$xy$-monotone angle sequence and are unique up to
rotation. Any~$xy$-monotone angle sequence is of the 
form~${[\LS(\LS\RS)^*]^4}$, where the single~$\LS$ describes the turn before an
extreme edge and~${(\LS\RS)^*}$ describes a stair sequence.
Without loss of generality, we assume that an~$xy$-monotone sequence always begins with~${\LS\LS}$ and that we always draw the first~${\LS\LS}$ as the topmost edge
(the top extreme edge).  Therefore, we can also use~$\TL$,~$\BL$,~$\BR$, and~$\TR$ 
to denote the corresponding stair sequences, namely the first, second,
third and fourth~${(\LS\RS)^*}$ subsequence after the first~${\LS\LS}$ in
cyclic order.
Let~$T$ be the concatenation of~$\TL$, the top extreme edge, and~$\TR$; let~$L$,~$B$, and~$R$ be defined analogously following
Fig.~\ref{fig:xymonExampleB}.  For a chain~$C$, let 
the~\emph{$\rlength$}~$\rnum{C}$ be the number of reflex vertices on~$C$.
If~${C\in\{\TR,\TL,\BL,\BR\}}$, then~$\rnum{C}$ corresponds to the number of horizontal line segments and the number of vertical line segments in~$C$.
When we say that a line segment lies above or below another one, we also require implicitly that both line segments share a grid column.

\begin{figure}[tb]
	\captionsetup[subfigure]{justification=centering}
	\centering
	\subcaptionbox{\label{fig:xymonExampleB}}%
	{\includegraphics[page=1]{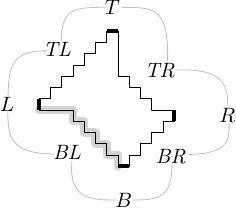}}%
	\hfil
	\subcaptionbox{\label{fig:configA}}%
	{\includegraphics[page=2]{xy-configurations}}%
	\hfil
	\subcaptionbox{\label{fig:configB}}%
	{\includegraphics[page=3]{xy-configurations}}%
	\caption{Extreme edges are bold.  Stair \BL is highlighted. 
		\subref{fig:xymonExampleB} The four stairs \TL, \TR, \BR, and \BL of an~$xy$-monotone polygon.  The sequences~$T$,~$R$,~$B$, and~$L$ are
		unions of neighboring stairs.  \subref{fig:configA} \& \subref{fig:configB}~Two possibly optimum
		configurations of the polygon.}\label{fig:configs}
\end{figure}

In this section, we obtain the following two results.
\begin{theorem}\label{thm:xyarea}
	Given an~$xy$-monotone angle sequence~$S$ of length~$n$, we can find
	a polygon~$P$ that realizes~$S$ and minimizes its
	\begin{inlinelistRoman}
		\item\label{thm:part1bb} bounding box or
		\item\label{thm:part2area} area
	\end{inlinelistRoman}
	in~${\bigOh(n)}$ time,
	and in constant time 
	we can find the optimum objective value  
	if the~$\rlength$s of the stair sequences are given.
\end{theorem}
Part~\ref{thm:part1bb} of Theorem~\ref{thm:xyarea} follows from the following observation:
The bounding box of every polygon that 
realizes~$S$ has width at least \[{\max\{\rnum{T},\rnum{B}\}+1}\] and height
at least \[{\max\{\rnum{L},\rnum{R}\}+1}~.\]
Since we can always draw three
stairs with edges of unit length, we can meet these lower bounds.
\medskip

For part~\ref{thm:part2area}, we first consider angle sequences with at most two nonempty stairs. Here, 
the only non-trivial case is when the angle sequence consists of two opposite 
stair sequences, that is,~$\TL$ and~$\BR$, or~$\BL$ and~$\TR$. Without loss of generality, consider the second case.

A stair has two \emph{delimiters} which are the two vertices outside the stair that are adjacent to the endpoints of the stair; see Figure~\ref{fig:delimeters}.
Note that a delimiter is a convex vertex ($\LS$ vertex).
For each convex vertex of a stair and its delimiters, a~\emph{step} is the polyline consisting of its two adjacent edges. 
For a convex vertex of a stair, its step is~\emph{good} if both edges have the same length.  
For a delimiter, its step is~\emph{good} if the edge adjacent to the stair is shorter by~$1$ than the other edge.
A step that is not good is \emph{bad}.
The~\emph{size} of a step is the minimum of the lengths of its two edges.
\begin{figure}[t]
	\begin{subfigure}[b]{.45\textwidth}
		\centering
		\includegraphics[page=4]{two-stairs-a-is-b}
		\caption{The stair~$\TR$ (bold) with two delimiters
			(white nodes). There are only two good steps
			(highlighted) that belong to~$\TR$ and its
			delimiters.}
		\label{fig:delimeters}
	\end{subfigure}
	\hfill
	\begin{subfigure}[b]{.5\textwidth}
		\centering
		\includegraphics[page=1]{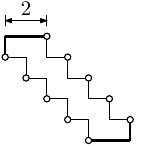}
		\caption{The (only) two optimum polygons realizing
			the~$2$-stair instance
			with~${\rnum{\BL}=\rnum{\TR}=4}$. The nodes depict skew
			convex vertices. The extreme edges are bold.}
		\label{fig:two-stair-a-is-b}
	\end{subfigure}
	
	\caption{A stair with good and bad
		steps~\subref{fig:delimeters}, and two optimum polygons
		realizing a~$2$-stair
		instance~\subref{fig:two-stair-a-is-b}.}
\end{figure}

\begin{lemma}\label{lem:two-stairs-equal}
	Let~$S$ be an~$xy$-monotone angle sequence of length~$n$ consisting
	of exactly two nonempty opposite stair sequences~$\BL$ and~$\TR$. 
	If~${\rnum{\BL}= \rnum{\TR}}$, then we can choose any extreme edge and, in~${\bigOh(n)}$ time,
	we can compute a minimum-area polygon realizing~$S$ such that the chosen extreme edge has length~$1$.
\end{lemma}
\begin{proof}
	Fix a minimum-area polygon~$\Popt$ that realizes~$S$. Let~${\lenTR=\rnum{\TR}}$ and~${\lenBL=\rnum{\BL}}$.
	If~${\lenTR = \lenBL}$, then any two parallel extreme edges have length~$2$ and all other edges have length~$1$; see Fig.~\ref{fig:two-stair-a-is-b}.
	To see this, we use a charging argument. Call a convex vertex~\emph{skew} if it is the top right corner or the bottom left corner of the bounding box of its two adjacent edges. Observe that a grid cell lying in the interior of a polygon can touch at most one skew convex vertex of the polygon, assuming that the polygon has more than four vertices. As each convex vertex is touched by exactly one grid cell from the interior, the number of skew convex vertices is a lower bound on the area. Thus, the two polygons of our construction are optimum as every grid cell is touching a skew convex vertex. Hence, if~${\lenTR=\lenBL}$, the minimum area is
	\[\area{\Popt} ~=~ 2(\lenBL + 1)~.\]
	Also note that these two polygons are the only optimum ones as any other polygon contains at least one grid cell not adjacent to any skew convex vertex.
\end{proof}

\begin{lemma}\label{lem:two-stairs-not-equal}
	Let~$S$ be an~$xy$-monotone angle sequence of length~$n$ consisting
	of exactly two nonempty opposite stair sequences~$\BL$ and~$\TR$. 
	If~${\rnum{\BL}\ne \rnum{\TR}}$, let~${X\in\{\BL,\TR\}}$ be the stair with the smaller number of reflex vertices. 
	Given any priorities on the steps belonging to~$X$ and its delimiters, 
	in~${\bigOh(n)}$ time, we can compute a minimum-area polygon realizing~$S$ that minimizes the sizes of the steps according to the priorities. 
\end{lemma}
\begin{proof}
	Fix a minimum-area polygon~$\Popt$ that realizes~$S$. Let~${\lenTR=\rnum{\TR}}$ and~${\lenBL=\rnum{\BL}}$.
	Assume~${\lenTR < \lenBL}$ (by rotation if necessary). 
	Let~$\exBL$ denote the polyline consisting of~$\BL$ and the bottom and left extreme edge,
	and let~$\exTR$ denote the polyline consisting of~$\TR$ and the top and right extreme edge.
	
	\begin{figure}[t]
		\newcommand{\figurespace}{\quad~}
		\subcaptionbox{\label{fig:bl-unit-1}If two segments share two grid columns (hatched and shaded), contract both by one unit.}{\figurespace\includegraphics[page=1]{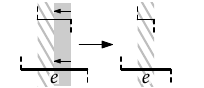}\figurespace}
		\hfill
		\subcaptionbox{\label{fig:bl-unit-2}If a segments shares two grid columns (hatched and shaded) adjacent to a reflex vertex of a long segment, contract both segments by one unit.}{\figurespace\includegraphics[page=2]{BL-unit-new-merged}\figurespace}
		\hfill
		\subcaptionbox{\label{fig:bl-unit-3}If there is a vertex (white node) one unit left and to the bottom of~$v$, then its incident horizontal segment has length at least~$2$.}{\figurespace\includegraphics[page=3]{BL-unit-new-merged}\figurespace}
		\newline
		\subcaptionbox{\label{fig:bl-unit-5}If the vertical segment adjacent to~$v$ has only length~$1$, then decrease the area.}{\figurespace\includegraphics[page=4]{BL-unit-new-merged}\figurespace}
		\hfill
		\subcaptionbox{\label{fig:bl-unit-6}If the vertical segment adjacent to~$v$ has length at least~$2$, then decrease the area by introducing a new reflex vertex (black node).}{\figurespace\includegraphics[page=5]{BL-unit-new-merged}\figurespace}
		\hfill
		\subcaptionbox{\label{fig:bl-unit-7}Then remove one reflex vertex (black node) by removing one unit of the right end of the polygon.}{\figurespace\includegraphics[page=6]{BL-unit-new-merged}\figurespace}
		
		\caption{Forbidden configurations for~$\Popt$ as they allow to decrease the area. In~\subref{fig:bl-unit-3}\textendash\subref{fig:bl-unit-7}, we assume that the only segment in~$\exBL$ of length greater than~$1$ is the bottom extreme edge.}
		\label{fig:bl-unit}
	\end{figure}
	First, we show that 
	all segments of~$\exBL$ are of unit length.
	Suppose that the 
	claim were false and that there is, without loss of generality, a horizontal line segment in~$\exBL$ longer than~$1$.
	Consider the leftmost such segment~$e$ and let~${l(e)}$ and~${r(e)}$ denote its left and right endpoint, respectively. 
	If there were a horizontal segment in~$\exTR$ sharing at least two grid columns with~$e$, we could contract both segments by one unit and decrease the area of~$\Popt$ without causing~$\exBL$ and~$\exTR$ to intersect; a contradiction to the optimality of~$\Popt$; see Fig.~\ref{fig:bl-unit-1}.
	There is also no horizontal segment in~$\exTR$ passing through the two grid columns left and right of~${r(e)}$, as, again, we could contract and obtain a contradiction; see Fig.~\ref{fig:bl-unit-2}.
	
	We will now show that~$e$ is not the bottom extreme edge. If it were, we could modify~$\Popt$ as follows to decrease its area.
	First, we will observe that there is a convex vertex~$v$ of~$\exTR$ whose both incident edges have length at least~$2$ and that there is no vertex of~$\exBL$ one unit to the left and to the bottom of it.
	Given~${\lenTR < \lenBL}$, there is a horizontal line segment in~$\exTR$ of length at least~$2$. Consider the rightmost such segment~$f$ and let~$v$ denote the right endpoint of~$f$.
	If~$f$ is the top extreme edge, then all horizontal edges, with the exception of~$e$ and~$f$, have length~$1$.
	Given~${\lenTR < \lenBL}$, that fact implies~${\len{f}>\len{e}}$. Hence,~${\len{f}\ge 3}$.
	In both cases of whether~$f$ is the top extreme edge or not, if there were a vertex of~$\exBL$ lying one unit to the left and to the bottom of~$v$, 
	then there is an incident horizontal edge of length at least~$2$; see Fig.~\ref{fig:bl-unit-3}. This, however is a contradiction as the only edge in~$\exBL$ of length bigger than~$1$ is~$e$ and its right endpoint is the rightmost vertex in~$\exBL$.
	Suppose that the vertical edge incident to~$v$ had only length~$1$.
	Then we could move the vertical edge by one unit to the left without causing any intersections; see Fig.~\ref{fig:bl-unit-5}.
	This, however, is a contradiction to the optimality of~$\Popt$. 
	Consider the grid cell inside~$\Popt$ that has~$v$ 
	as an endpoint. As argued above, it intersects no vertices of~$\exBL$ and, consequently, no line segments of~$\exBL$.
	Rotate the grid cell, together with the line drawings on its boundary, by~${180^\circ}$; see Fig.~\ref{fig:bl-unit-6}. 
	The resulting polygon~${{\Popt}'}$ has less area than~$\Popt$, but one reflex vertex more.
	To remove one reflex vertex from~${{\Popt}'}$, we contract one unit of~$e$ and we contract the rightmost edge of~$\exTR$, which has length~$1$; see Fig.~\ref{fig:bl-unit-7}.
	Hence, the area decreases again, and we obtain a contradiction to the optimality of~$\Popt$, as our resulting polygon realizes the same angle sequence.
	We conclude that~$e$ is not the bottom extreme edge. 
	
	\begin{figure}[t]
		\begin{subfigure}[b]{.4\textwidth}
			\centering
			\includegraphics[page=7]{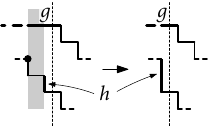}
			\caption{The sweepline (dotted) stabbed a long
				segment~$g$ of~$\exTR$. Contract~$g$ and the
				horizontal segment of~$\exBL$ left to the
				sweepline by one unit. The vertical segment~$h$
				gets longer and we loose one reflex vertex
				(black node).}
			\label{fig:bl-unit-8}
		\end{subfigure}
		\hfill
		\begin{subfigure}[b]{.3\textwidth}
			\centering
			\includegraphics[page=8]{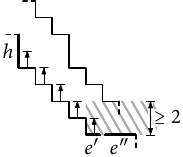}
			\caption{Move all segments that were stabbed by
				the sweepline up by one unit,
				including~${e'}$. There are not intersections as
				they have vertical distance at least~$2$
				to~$\exTR$.}
			\label{fig:bl-unit-9}
		\end{subfigure}
		\hfill
		\begin{subfigure}[b]{.2\textwidth}
			\centering
			\includegraphics[page=9]{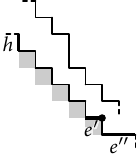}
			\caption{The resulting polygon is simple, has less
				area and contains one new reflex vertex (black
				node).}
			\label{fig:bl-unit-10}
		\end{subfigure}
		\caption{If there is a line segment in~$\exBL$ of
			length greater than~$1$, then we can decrease the
			area of~$\Popt$ in several steps.}
		\label{fig:bl-unit-B}
	\end{figure}
	Next, using the fact that~$e$ is not the bottom extreme edge, we will decrease the area of the polygon by removing a carefully chosen reflex vertex from~$\exBL$. Later, we will restore the angle sequence of~$\exBL$ without increasing the area and thus obtain a contradiction.
	We cut~$e$ one unit right to~${l(e)}$ into two segments,~${e'}$ and~${e''}$, where~${e'}$ denotes the left part. 
	All the facts above imply that the vertical distance between~${e'}$ and~$\exTR$ is at least~$\len{e}$, hence, at least~$2$.
	Place a vertical line through~${e'}$, that we call a \emph{sweepline}, and move the line to the left until, for the first time, one of the two events occurs:
	\begin{inlinelistAlph}
		\item\label{event:TRhor} The horizontal line segment of~$\exTR$ stabbed by the line has length greater than~$1$, or 	
		\item\label{event:BLver} the horizontal line segment of~$\exBL$ stabbed by the line has an (left) incident vertical segment of length greater than~$1$. 
	\end{inlinelistAlph}
	Note that one of the two events will occur since, in our case, the left and top extreme edge cannot simultaneously attend length~$1$.
	Let~$h$ denote the left vertical line segment incident to the last horizontal line segment of~$\exBL$ stabbed by the sweepline; see Fig.~\ref{fig:bl-unit-8}. 
	If the sweep process terminates with event~\ref{event:TRhor}, take the horizontal line segment~$g$ in~$\exTR$ of length at least~$2$ that has been stabbed by the sweep line. Contract one unit of this segment and contract the rightmost horizontal line segment of~$\exBL$ 
	left to~${e'}$ that has not been stabbed. The latter segment has to be a unit-length segment.
	By this operation, we decrease the area of~$\Popt$, we increase the length of~$h$, and we loose one reflex vertex in~$\BL$.
	We proceed similarly if the sweep process does not terminate with event~\ref{event:TRhor}. We take any horizontal line segment~$g$ in~$\exTR$ of length at least~$2$, which exists given~${\lenTR < \lenBL}$, and which lies left to the sweep line or right to~$e$. Then we contract one unit of~$g$ and we contract the leftmost horizontal line segment lying below~$g$. 
	As a result, we decrease the area and loose one reflex vertex.
	
	In both cases, the vertical edge~$h$ has length at least~$2$. 
	Now, in order to reintroduce the missing reflex vertex, we take the subsequence of all segments of~$\exBL$ that where stabbed by the line at some moment, and shift all these segments up by one unit. In the same time, we shrink~$h$ by one unit and connect the right endpoint of~${e'}$ via a vertical segment to~${e''}$; see Fig.~\ref{fig:bl-unit-9} and~\ref{fig:bl-unit-10}. 
	To see that we do not cause any intersections, recall that the distance between~${e'}$ and~$\exTR$ is at least~$2$. 
	Also, recall that all line segments of our subsequence have unit length, the horizontal ones as well as the vertical ones. Together with the fact that all horizontal line segments of~$\exTR$ lying above the subsequence, with possible exception of the last segment, also have unit length, we conclude that every line segment of our subsequence had distance at least~$2$ to~$\exTR$ before the up-shifting. Hence, we have obtained a feasible polygon for the same angle sequence as~$\Popt$ but with smaller area; a contradiction. 
	
	Next, we express the area of~$\Popt$ as a function of the edge lengths of~$\exTR$. 
	We will use the function to find out which values for the edge lengths minimize the area.
	For~${1\le i\le \lenTR+1}$, let~$\segTR_i$ denote the~$i$\thSuffix horizontal segment in~$\exTR$ from the left. 
	Given our assumption that all horizontal segments of~$\exBL$ are of unit-length, 
	we can express the length~${\len{\segTR_i}}$ of~$\segTR_i$ as the number of horizontal segments of~$\exBL$ lying below~$\segTR_i$.
	Thus, we have~${\sum_{i=1}^{\lenTR+1} \len{\segTR_i} = \lenBL + 1}$.
	Let~${\sarea(i)}$ denote the area below~$\segTR_i$ in~$\Popt$, that is, the number of grid cells in~$\Popt$ sharing a grid column with~$\segTR_i$.
	Since the left extreme edge in~$\Popt$ has length~$1$, the area in~$\Popt$ under~$\segTR_1$ is
	\[\sarea(1) ~=~ \sum_{j~=~1}^{\len{\segTR_1}} j ~=~ \frac{\len{\segTR_1}(\len{\segTR_1}+1)}{2}~.\]
	For~${2\le i \le \lenTR+1}$, the distance between~$\segTR_i$ and any horizontal segment below it is~$2$; it cannot be less, and if it were more, we could feasibly shift~$\segTR_i$ to the bottom by at least one unit, contradicting the optimality of~$\Popt$.
	Thus, we have 
	\[\sarea(i) ~=~ \sum_{j~=~1}^{\len{\segTR_i}} (j+1) ~=~ \frac{(\len{\segTR_i}+1)(\len{\segTR_i}+2)}{2} - 1~.\]
	We can overcome the difference between~${i=1}$ and~${i\ge 2}$ by splitting~$\segTR_1$ into~${\segTR'_0}$ and~${\segTR'_1}$, such that~${\len{\segTR'_0}=1}$ and~${\len{\segTR'_1}=\len{\segTR_1} - 1}$ holds. Note that~${\len{\segTR'_1}}$ can be~$0$. 
	For~${2\le i \le \lenTR+1}$, let~${\segTR'_i = \segTR_i}$. 
	Observe that now we have~${\sum_{i=1}^{\lenTR+1} \len{\segTR'_i} = \lenBL}$.
	Thus, 
	\begin{alignat*}{2}
	\area{\Popt} &=&~& 1 + \sum_{i=1}^{\lenTR+1}  \left(\frac{(\len{\segTR'_i}+1)(\len{\segTR'_i}+2)}2 - 1\right)\\
	&=&&1+\sum_{i=1}^{\lenTR+1} \left(\frac12\len{\segTR'_i}^2 +\frac{3}2\len{\segTR'_i}\right)\\
	&=&&1+\frac32b+\frac12\sum_{i=1}^{\lenTR+1} \len{\segTR'_i}^2 ~, 
	\end{alignat*}
	which is minimized if~${\sum_{i=1}^{\lenTR+1}\len{\segTR'_i}^2}$ is minimal.
	By Cauchy-Schwarz, we know that this is the case if, for every~${i\in\{1,\dots,\lenTR+1\}}$, the length~${\len{\segTR'_{i}}}$ is 
	equal to the arithmetic mean; since we have to use integers, the convexity of
	the function tells us that, for every~${i\in\{1,\dots,\lenTR+1\}}$, the length~${\len{\segTR'_{i}}}$ has to be as close to the
	arithmetic mean as possible, that is,
	\[{\len{\segTR'_{i}} \in \{\lfloor \lenBL/(\lenTR+1) \rfloor, \lceil \lenBL/(\lenTR+1) \rceil\}}~.\] 
	Let~$q$ bet the quotient and~$r$ the remainder when~$\lenBL$ is divided by~${\lenTR + 1}$.
	Hence,
	\[\area{\Popt} ~=~ \frac{(\lenTR + 1)(q+1)(q+2)}{2} - \lenTR + r(q+2) ~.\]
	Repeating the same discussion for the vertical segment, we obtain the fact that every line segments of~$\TR$ is of length~${\lfloor \lenBL/(\lenTR+1) \rfloor}$ or~${\lceil \lenBL/(\lenTR+1) \rceil}$, and the top and right extreme edge is of length~${\lfloor \lenBL/(\lenTR+1) \rfloor+1}$ or~${\lceil \lenBL/(\lenTR+1) \rceil+1}$ (the latter fact follows from~${\len{\segTR_1} = \len{\segTR'_1} + 1}$).
	Observe that, in~$\Popt$, all steps belonging to~$\TL$ and its delimiters are good steps. 
	Otherwise, we could take one of the two edges belonging to a bad step 
	and move it towards the interior of the polygon and thus contradict the optimality of~$\Popt$.
	Further, observe that, for~${1\le i \le \lenTR+1}$, the size of the~$i$\thSuffix step from the left corresponds to~${\len{\segTR'_{1}}}$.
	Hence, all steps are of size~${\lfloor \lenBL/(\lenTR+1) \rfloor}$ or~${\lceil \lenBL/(\lenTR+1) \rceil}$.
	We conclude
	that we can arbitrarily assign the values~${\lfloor \lenBL/(\lenTR+1) \rfloor}$ or~${\lceil \lenBL/(\lenTR+1) \rceil}$ to the steps sizes as long as they sum up to~$\lenBL$ and in this way obtain a feasible, and, hence, minimum polygon realizing~$S$. Thereby, we can take into account any priority on the steps given by the input.
	Thus, we can construct a minimum-area polygon realizing~$S$ in~${\bigOh(n)}$ time. 
\end{proof}

Note that the proofs of Lemma~\ref{lem:two-stairs-equal} and~\ref{lem:two-stairs-not-equal} also allow us to obtain, in~${\bigOh(1)}$ time, the exact area of a minimum polygon without having to construct it.
We summarize our results in the following corollary.
\begin{corollary}%\label{cor:two-stairs-complete}
	Let~$S$ be an~$xy$-monotone angle sequence of length~$n$ consisting
	of exactly two nonempty opposite stair sequences~$\BL$ and~$\TR$. 
	We can find a
	minimum-area polygon that realizes~$S$ in~${\bigOh(n)}$ time.  
	If~${\rnum{\BL}}$
	and~${\rnum{\TR}}$ are given, we can compute the area of such a polygon in~${\bigOh(1)}$ time.
\end{corollary}

The proofs of Lemmas~\ref{lem:two-stairs-equal} and~\ref{lem:two-stairs-not-equal} also lead to the following observation.
\begin{observation}\label{obs:two-stairs}
	Let~$P$ be any polygon realizing an angle sequence~$S$ consisting of exactly two nonempty opposite 
	stairs~$\TR$ and~$\BL$ 
	with~${\lenTR=\rnum{\TR}}$ and~${\lenBL=\rnum{\BL}}$.
	The polygon~$P$ is a polygon of minimum area realizing~$S$ if and only if the following holds:
	If~${\lenTR < \lenBL}$, then
	\begin{listRoman}
		\item the steps of~$\TR$ and its delimiters are good and have size~${\floor{\lenBL/(\lenTR + 1)}}$ or~${\ceil{\lenBL/(\lenTR + 1)}}$.
		\item the bottom and right extreme edge and all edges of~$\BL$ have length~$1$. 
	\end{listRoman}
	If~${\lenTR=\lenBL}$, then
	\begin{listRoman}[resume]
		\item two parallel extreme edges have length~$2$, and
		\item all other edges have length~$1$.
	\end{listRoman}
\end{observation}

We now consider the case of four nonempty stairs. 
(The case of three nonempty stairs can be solved analogously.)
\begin{definition}
	Let~$P$ be any~$xy$-monotone polygon~$P$ with four nonempty stairs~$\TL$,~$\TR$,~$\BL$, and~$\BR$. 
	For~${X\in\{\TL, \TR,\BL,\BR\}}$, let~$\BB{X}$ denote the bounding box of~$X$ and its adjacent extreme edges.
	An \emph{interior corner} of~$\BB{X}$ is the corner of~$\BB{X}$ that lies inside~$P$ and not on the extension of any extreme edge adjacent to~$X$. 
	We call~$P$ \emph{half-canonical} if~$P$ has two non-adjacent 
	nonempty stairs~${(X,Y)\in\{(\TL,\BR), (\TR,\BL)\}}$ 
	such that 
	\begin{listRoman}[label=(C\arabic*)]
		\item\label{xycanon:adj}~$\BB{X}$ and~$\BB{Y}$
		do not intersect in more than one point,
	\end{listRoman}   
	and we call it \emph{canonical} if even
	\begin{listRoman}[label=(C\arabic*),resume]  
		\item\label{xycanon:corner}
		each of the two interior corners of~$\BB{X}$ and~$\BB{Y}$ 
		lies on a line segment of~$P$ that also contains an endpoint of one the two stairs in~${\{\TL,\BR,\TR,\BL\}\setminus X\cup Y}$. 
	\end{listRoman} 
\end{definition}
Figure~\ref{fig:example-canonical} depicts some examples for the case that~${X=\Btl}$ and~${Y=\Bbr}$ holds. 
In Property~\ref{xycanon:corner}, the interior corner of the bounding box may coincide with the endpoint of the respective stair; see~$\Btl$ in Fig.~\ref{fig:example-canonical-2}.
Also note that Property~\ref{xycanon:adj} is a necessary condition for Property~\ref{xycanon:corner}. 
\begin{figure}[t]
	\centering
	\subcaptionbox{\label{fig:example-canonical-1}Two half-canonical polygons that are not canonical.}{~~~\qquad\includegraphics[page=3]{example-canonical}\qquad\includegraphics[page=1]{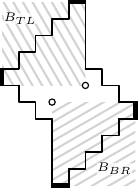}\qquad~~~}
	\hfil
	\subcaptionbox{\label{fig:example-canonical-2}A canonical polygon. }{~\includegraphics[page=2]{example-canonical}~}
	\caption{Examples of half-canonical and canonical polygons. The nodes depict the interior corners of the bounding boxes (hatched).} 
	\label{fig:example-canonical}
\end{figure}

\begin{lemma}\label{lem:xycanonical}
	For every four-stair sequence~$S$ with~${|S|> 36}$, there 
	exists a polygon of minimum area realizing~$S$ that is canonical.
\end{lemma}  
\begin{proof}
	Consider an optimum polygon~$\Popt$ realizing the angle sequence~$S$. 
	Suppose it
	is not canonical.  Observe that all four extreme edges are of
	length~$1$, otherwise the polygon is not optimum.
	
	First, suppose that Property~\ref{xycanon:adj} does
	not hold.  
	Then, for any pair of two opposite stairs, the bounding boxes of their adjacent extreme edges intersect in more than one point. 
	Hence, the (closed)~$x$-ranges of the horizontal extreme edges intersect and the (closed)~$y$-ranges of the vertical extreme edges intersect. 
	Since the extreme edges have length~$1$, and the bounding boxes intersect in more than one point, we even have that  
	the  (closed)~$x$-ranges of the top and bottom extreme edges are the same, or 
	the (closed)~$y$-ranges of the left and right extreme edges are the same.
	Suppose (by rotation if necessary) it is the latter and also suppose
	(by temporary vertical or horizontal reflection and, afterwards, backward reflection) that the stair~$\TR$ 
	has~$\rlength$ greater than~$4$ (since~${|S| > 36}$,  
	this is possible).
	Let~$u$ be the left endpoint of the bottom extreme edge and let~$v$ be the 
	reflex vertex that precedes, in ccw order, the top extreme edge;
	see Fig.~\ref{fig:make-canonical-1}.
	
	\begin{figure}[t]
		\subcaptionbox{\label{fig:make-canonical-1}Shifting the ccw path between~$u$ and~$v$ down by two units.}{\qquad\includegraphics[page=1]{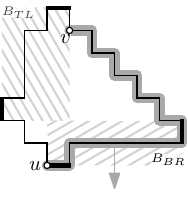}\qquad}
		\hfill
		\subcaptionbox{\label{fig:make-canonical-2}Shifting the ccw path between~$w$ and~$z$ left by three units.}{\qquad\includegraphics[page=2]{make-canonical}\qquad}
		\hfill
		\subcaptionbox{\label{fig:make-canonical-3}The resulting polygon.}{\includegraphics[page=3]{make-canonical}}
		\caption{Transforming an~$xy$-monotone polygon to a polygon that 
			satisfies~\ref{xycanon:adj} and has less area.}
		\label{fig:make-canonical}
	\end{figure}
	We shift the boundary of~$\Popt$ that lies on the ccw walk from~$u$ to~$v$ down 
	by two units, stretching the vertical edges adjacent to~$u$ and~$v$.
	The new polygon~${P'}$ still realizes the angle sequence and its area 
	is larger by two units than the area of~$P$. However, now~$\Btl$ and~$\Bbr$ 
	are intersection-free. Let~$w$ be the reflex vertex that follows, in ccw order, 
	the right extreme edge and let~$z$ be the bottom endpoint of the left
	extreme edge; see Fig.~\ref{fig:make-canonical-2}. We shift the boundary 
	of~${P'}$ that lies on the ccw walk from~$w$ to~$z$ to the left by three units, 
	stretching the horizontal edges adjacent to~$w$ and~$z$.
	The new polygon still realizes the angle sequence and is still simple: 
	The only crossings that can occur by this 
	operation are between~\TR and~\BL. The left extreme edge lies at most
	three rows above the right extreme edge~$\rho$; hence, any crossing must involve 
	the vertical edge~$e_1$ of~\TR in the row above~$\rho$ or
	the vertical edge~$e_2$ of~\TR two rows above~$\rho$; see Fig.~\ref{fig:make-canonical}.
	Let the~$x$-axis go from left to right and let~${x(v)}$ denote the~$x$-coordinate of~$v$ where~$v$ is a vertex or a vertical segment.
	Since~${\rnum{\TR}>4}$, we have after the shift 
	\[x(e_1)~\ge~ x(e_2)~\ge~ x(v)+\rnum{\TR}-2~\ge~ x(v)+3 ~=~ x(u)+1~.\]
	Since each vertical edge of~\BL has~$x$-coordinate at most~${x(u)}$, there can be
	no crossing. However, now the area of the 
	polygon decreased by three units; a contradiction to the fact
	that~$\Popt$ is optimum. Hence, Property~\ref{xycanon:adj} has to hold for~$\Popt$.
	
	Now, assume that there is a bounding box pair having at most one point in common, 
	without loss of generality,~$\Btl$ and~$\Bbr$. Since the optimum polygon~$\Popt$ is not canonical, 
	Property~\ref{xycanon:corner} has to be violated by at 
	least one of the two bounding boxes, say~$\Btl$. 
	Then the interior corner (bottom right corner) of~$\Btl$ does not lie on a line segment that also contains an endpoint of~$\TR$ or~$\BL$.
	Hence, the endpoints of~$\TR$ or~$\BL$ have to lie on the boundary of~$\Btl$ \enquote{behind} the interior corner, that is, they lie on two different edges of~$\Btl$ 
	and, for each one of them, its distance to the 
	closest corner of~$\Btl$ is at least~$1$.
	Then, for at least one of the two edges, it holds that the line going through 
	the edge does not cross the interior of~$\Bbr$ (it can happen that only one such line exists as Fig.~\ref{fig:example-canonical-1} indicates). Without loss of generality, this holds for the 
	line~$g$ that goes through the horizontal edge of~$\Btl$. 
	
	\begin{figure}[t]
		\subcaptionbox{\label{fig:make-C2-1}Shifting a horizontal edge onto~$g$.}{\includegraphics[page=1]{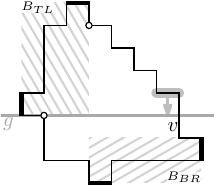}}
		\hfill
		\subcaptionbox{\label{fig:make-C2-2}Shifting the part above~$g$ to the left.}{\includegraphics[page=2]{make-C2}}
		\hfill
		\subcaptionbox{\label{fig:make-C2-3}The resulting polygon satisfies~\ref{xycanon:corner} for~$\Btl$ (not for~$\Bbr$).}
		{\includegraphics[page=3]{make-C2}}
		\caption{Transforming an~$xy$-monotone polygon to a polygon that 
			satisfies~\ref{xycanon:corner} and has the same area.}
		\label{fig:make-C2}
	\end{figure}
	Next, we 
	observe that~$g$ does not cross any vertical line segment of~$\TR$; instead, 
	there is a horizontal line segment of~$\TR$ lying on~$g$.
	To see this, suppose the contrary. Thus, there exists a vertical line segment~$v$ 
	of~$\TR$ that is cut by~$g$; see Fig.~\ref{fig:make-C2-1}. Thus, the two endpoints of~$v$ lie at least one 
	unit above and below~$g$, respectively. Consider the horizontal line segment 
	of~$\TR$ starting at the top endpoint of~$v$. We can move the horizontal segment 
	downwards and place it on~$g$. By this operation, the angle sequence does not change and 
	the polygon remains simple as all line segments of~$\BL$, the only segments that 
	might cross~$\TR$ after his operation, lie below~$g$ by at least one unit.
	Hence, by moving the horizontal edge downwards, we in fact shrink the area of the polygon; a 
	contradiction to its optimality.
	Thus,~$g$ contains a horizontal line segment of~$\TR$. 
	
	Now, we cut the polygon 
	through~$g$ into two parts; see Fig.~\ref{fig:make-C2-2}. Then, we shift the upper part to the left 
	until the endpoint of~$\BL$ coincides with the bottom right corner of~$\Btl$; see Fig.~\ref{fig:make-C2-3}. Hence, Property~\ref{xycanon:corner} is satisfied for~$\Btl$. 
	Moreover, the resulting polygon realizes the same angle sequence as before and has the 
	same area as before. 
	Note that if~$\Bbr$ satisfied Property~\ref{xycanon:corner} before the shift operation, then it also satisfies the property afterwards: If its interior corner (top left corner) lies below~$g$, then any edge containing the corner will remain unchanged as we do not change anything below~$g$. If its interior corner lies on~$g$, then~$\Bbr$ can only satisfy Property~\ref{xycanon:corner} by an endpoint of~$\TR$ which has also to lie on~$g$. During the shift operation, we move this endpoint only to the left, thus the property remains fulfilled for~$\Bbr$.
	
	If the polygon is not yet canonical, then we repeat the procedure with~$\Bbr$ (without losing Property~\ref{xycanon:corner} for~$\Btl$)
	and obtain a canonical optimum polygon. Hence, 
	Property~\ref{xycanon:corner} holds. 
\end{proof}

Let~$\Popt$ be a canonical optimum polygon. Without loss of generality, Property~\ref{xycanon:corner} is satisfied for~$\Btl$ and~$\Bbr$. 
Consider the line segment of~$\TR$ and the line segment of~$\BL$ that connect 
to~$\Btl$ in a canonical polygon. The two line segments are connected to a same edge of~$\Btl$ and are  
\begin{listRoman}
	\item\label{caseFourCanonical:bothHor} both horizontal, %a
	\item\label{caseFourCanonical:bothVer} both vertical, or %b
	\item\label{caseFourCanonical:Perp} perpendicular to each other. %c
\end{listRoman} 
The same holds for~$\Bbr$.
Consequently, there is only a constant number of ways in which the stairs 
outside the two bounding boxes are connected to them. 
Even more, the three cases cannot appear arbitrarily in an optimum polygon as we will see below.

We cut the optimum polygon~$\Popt$ along the edge 
of~$\Btl$ to which~$\BL$ and~$\TR$ are connected. We also cut along the 
respective edge of~$\Bbr$. We get three polygons~$\Popt_1$,~$\Popt_2$ and~$\Popt_3$. 
The polygons~$\Popt_1$ and~$\Popt_3$, which lie on the outside, 
realize the~$1$-stair sequence defined by~$\TL$ and~$\BR$ (including adjacent
extreme edges), respectively, whereas the middle polygon~$\Popt_2$ realizes the~$2$-stair sequence defined by the concatenation of~$\BL$,~$\TR$, and the edge segments 
of~$\Btl$ and~$\Bbr$ that connect them.

\begin{figure}
	\centering
	\includegraphics[page=1]{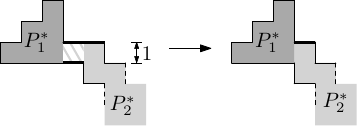}
	\caption{If~$\Popt_2$ connects to~$\Popt_1$ via two horizontal segments (bold) of distance~$1$, then the top one of them has length at least~$2$. Thus, we can contract both by one unit and reduce the area by one grid cell (hatched).}
	\label{fig:four-stairs-connections}
\end{figure}
Let~${\lenTR=\rnum{\TR}}$ and~${\lenBL=\rnum{\BL}}$. 
If~${\lenTR=\lenBL}$, then, for at least one of the two bounding boxes~$\Btl$ and~$\Bbr$, Case~\ref{caseFourCanonical:Perp} holds.
To see this, suppose the contrary. Then, for~$\Popt_1$ and~$\Popt_3$, the two parallel segments of~$\TR$ an~$\BL$ attached to it have distance at least~$2$, as otherwise 
we could shrink the area; see Fig.~\ref{fig:four-stairs-connections}. This fact implies that the extreme edges of~$\Btl$ and~$\Bbr$ to which we attached~$\Popt_2$ have length at least~$3$.
Let~$e$ and~$f$ denote the extreme edges in the angle sequence of~$\Popt_2$ to which we attached~$\Popt_1$ and~$\Popt_3$, respectively.
By Observation~\ref{obs:two-stairs}, we compute a minimum-area polygon~$P_2$ for the angle sequence of~$\Popt_2$ such that~$e$ has length~$1$ and~$f$ has length at most~$2$. Then, we can feasibly attach~$\Popt_1$ and~$\Popt_3$ to~$P_2$ yielding a polygon for~$S$ of area at most~${\area{\Popt}}$. However, now the two parallel segments of~$\TR$ an~$\BL$ touching~$\Popt_1$ have only distance~$1$. As discussed above, we can shrink the area; a contradiction to the optimality of~$\Popt$.

This observation leads to the following algorithm: For~${|S|\le 36}$, we 
find a solution in constant time by exhaustive search. For larger~${|S|}$, 
we guess which pair of opposite bounding boxes in~${\{(\TL,\BR), (\TR,\BL)\}}$ is intersection-free in the canonical optimum polygon~$\Popt$
that we want to compute. Without loss of generality, we guessed~$\Btl$ and~$\Bbr$ (the other case 
is symmetric). Then, we guess how~$\TR$ and~$\BL$, the two stairs outside~$\Btl$ 
and~$\Bbr$, are connected to each of the two bounding boxes (see Cases~\ref{caseFourCanonical:bothHor}--\ref{caseFourCanonical:Perp}).
The guessed information gives us two~$1$-stair instances and a~$2$-stair instance.
We solve the instances independently and then put the solutions
together to form a solution to the whole instance. 

Whereas the~$1$-stair instances are trivial to solve, we apply Lemmas~\ref{lem:two-stairs-equal} and~~\ref{lem:two-stairs-not-equal} to obtain a solution to the middle instance. For this purpose, we will also fix some edge lengths and assign priorities to steps as follows.  
Let~${\lenTR=\rnum{\TR}}$ and~${\lenBL=\rnum{\BL}}$. Without loss of generality,~${\lenTR\le\lenBL}$ and~${\rnum{\TL}\le \rnum{\BR}}$ (the other cases are symmetric).
Assume~${\lenTR=\lenBL}$. If we guessed Case~\ref{caseFourCanonical:Perp} for both~$\Btl$ and~$\Bbr$, then we choose an arbitrary extreme edge to have length~$1$. 
Otherwise, exactly one of the two bounding boxes is in Case~\ref{caseFourCanonical:bothHor} or~\ref{caseFourCanonical:bothVer}. When its corresponding instance has been solved, we have to attach the solution to a particular extreme edge of the solution of the middle instance. We choose this extreme edge to have length~$1$ in the solution (see Lemma~\ref{lem:two-stairs-equal}).
Next, assume~${\lenTR < \lenBL}$. 
Recall that for this case, the algorithm of Lemma~\ref{lem:two-stairs-not-equal} takes any priorities into account that we have assigned to the steps. The algorithm  guarantees that steps of higher priority are not smaller than steps of lower priority. We will assign the priorities in the following way.
If we guessed Case~\ref{caseFourCanonical:bothVer} for~$\Btl$, 
then we assign the highest priority to the step of the left delimiter of~$\TR$, and the second-highest priority to the step of the right delimiter of~$\TR$. In all other cases, we give the highest priority to the step of the right delimiter of~$\TR$.

In detail, we put our three solutions together as follows. Let~$P_1$ denote our solution to the
instance corresponding to~$\Btl$, let~$P_2$ denote our solution to the middle instance, and let~$P_3$ denote our solution to the
instance corresponding to~$\Bbr$; see Fig.~\ref{fig:xy-decomposition-1}.
If we guessed Case~\ref{caseFourCanonical:bothVer} for~$\Btl$, then we put~$P_1$ and~$P_2$ together along 
their corresponding horizontal extreme edges.
If the bottom extreme edge of~$P_1$ is too short, we make it sufficiently longer 
by shifting the left extreme edge of~$P_1$ to the left; see Fig..~\ref{fig:xy-decomposition-2}. Case~\ref{caseFourCanonical:bothHor} works symmetrically.
If we guessed Case~\ref{caseFourCanonical:Perp} for~$\Btl$, then note that the left or top extreme edge 
of~$P_2$ has length at least~$2$ (independently of that we are in Case~\ref{caseFourCanonical:Perp}). 
We glue~$P_1$ and~$P_2$ together along this extreme edge
and the corresponding extreme edge of~$P_1$; see Fig.~\ref{fig:xy-decomposition-1}.
We repeat the same process with~$P_2$ and~$P_3$.

\begin{figure}[t]
	\subcaptionbox{\label{fig:xy-decomposition-1}The three solutions. The arrows indicate how to attach~$P_1$ if we are in Case~\ref{caseFourCanonical:Perp}.}{\includegraphics[page=1]{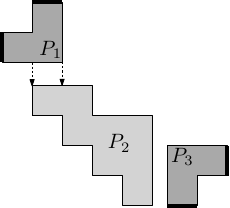}}
	\hfill
	\subcaptionbox{\label{fig:xy-decomposition-2}We use Case~\ref{caseFourCanonical:bothVer} for~$P_1$ and Case~\ref{caseFourCanonical:bothHor} for~$P_3$, both have to be stretched.}{\includegraphics[page=2]{xy-decomposition}}
	\hfill
	\subcaptionbox{\label{fig:xy-decomposition-3}The alternative cut yields three instances optimally solved by~${P'_1}$,~${P'_2}$ and again~$P_3$.}{\includegraphics[page=3]{xy-decomposition}}
	\caption{Putting the three solutions~$P_1$,~$P_2$ and~$P_3$ together.}
	\label{fig:xy-decomposition}
\end{figure}
All in all, we obtain a canonical polygon~$P$ which realizes the given angle sequence. 
To show that it has minimum area, we cut it into three smaller parts and show that the area of each part is upper bounded by the corresponding part of~$\Popt$.
Our choice of the parts will depend on the following cases.
\begin{listArabic}
	\item
	First, assume that we did not prolong any extreme edges.
	Consider the optimum polygon~$\Popt$ realizing~$S$ and our guesses and cut it accordingly to obtain three polygons~$\Popt_1$,~$\Popt_2$, and~$\Popt_3$, corresponding to the three instances of~$P_1$,~$P_2$ and~$P_3$, respectively. 
	Note that, by construction,~$P_2$ is a minimum-area polygon. Since we did not prolong the edges of~$P_1$ and~$P_2$, these polygons are also of minimum area.
	Hence, for~${1\le i\le 3}$, we obtain~${\area{P_i}\le\area{\Popt_i}}$, implying~\[\area{P}~\le~ \area{\Popt}~.\]
	
	\item\label{casecut:pone}
	Secondly, assume that we did prolong only an extreme edge of~$P_1$. Then we guessed Cases~\ref{caseFourCanonical:bothHor} or~\ref{caseFourCanonical:bothVer} for~$\Btl$.
	Note that we did not prolong any edge of~$P_2$ if~${\lenTR=\lenBL}$
	as otherwise we would have solved~$P_2$ such that the extreme edge on~$P_2$ to which we attached~$P_1$ would have unit length; contradicting the necessity to prolong the corresponding extreme edge of~$P_1$.
	Thus, we have~${\lenTR < \lenBL}$.
	Further, observe that we did not guess Case~\ref{caseFourCanonical:bothHor} for~$P_1$, as otherwise we would attach~$P_1$ to the left extreme edge of~$P_2$.
	This would, however, contradict the necessity to prolong as the left extreme edge has unit length by Observation~\ref{obs:two-stairs} and the fact~${\lenTR < \lenBL}$.  
	We conclude that we prolonged the bottom extreme edge~$e$ of~$P_1$ in 
	Case~\ref{caseFourCanonical:bothVer}; see Fig.~\ref{fig:xy-decomposition-2}. 
	We now cut~$\Popt$ in a slightly different way. 
	Our first cut goes horizontally through the top endpoint of the left extreme edge (before, we cut through the bottom endpoint),
	and our second cut is the same as before. 
	Hence, by the second cut, we again obtain~$\Popt_3$.
	The two other polygons that we get,~${{\Popt_1}'}$ and~${{\Popt_2}'}$, realize a~$1$- and a~$2$-stair instance, respectively.
	We cut~$P$ in the same way and obtain three polygons~${P'_1}$,~${P'_2}$, and~$P_3$, where~${P'_2}$ is the polygon realizing the~$2$ star instance; see Fig.~\ref{fig:xy-decomposition-3}.
	Whereas~$P_1$ is not a minimum-area polygon due to the prolongation of its extreme edge, we have that~${P'_1}$ as well as~$P_3$ is a minimum-area polygon. 
	Hence,~${\area{P'_1}\le \area{\Popt_1}}$ and~${\area{P_3}\le\area{\Popt_3}}$.
	We now show that~${\area{P'_2}\le\area{{\Popt_2}'}}$ holds by proving that~${P'_2}$ is a minimum-area polygon.
	The three inequalities will imply~${\area{P}\le\area{\Popt}}$.
	
	Let~${\BL'}$ be the stair sequence that we obtain by adding one reflex and one convex vertex to~$\BL$. Thus, we have~${\rnum{\BL'}=\lenBL+1}$.
	Observe that the~$2$-stair instance realized by~${P'_2}$ and~${{\Popt_2}'}$ consists of the two stairs~$\TR$ and~${\BL'}$.
	Given that all line segments belonging to~$\BL$ in~$P_2$ had unit length, so do all the line segments in~${P'_2}$ belonging to~${\BL'}$.
	The same holds for the left and bottom extreme edge of~${P'_2}$.
	Then, note that the step~${s'}$ of the left delimiter of~$\TR$ in~${P'_2}$ is good.
	Also note that it is bigger by~$1$ when compared to the step~$s$ of the left delimiter of~$\TR$ in~$P_2$.
	Given our priorities on the steps when we computed~$P_2$, 
	the size of~$s$ is~${\floor{\lenBL/(\lenTR+1)}}$.
	Consequently, the size of~${s'}$ is~${\floor{\lenBL/(\lenTR+1)}+1}$. 
	Let~$S$ and~${S'}$ denote all the steps belonging to~$\TR$ with its delimiters in~$P_2$ and~${P'_2}$, respectively.
	As the sizes of the steps in~${S\setminus\{s\}}$ did not change, all steps in~${S'}$ have sizes in~${\{\floor{\lenBL/(\lenTR+1)}, \ceil{\lenBL/(\lenTR+1)}, \floor{\lenBL/(\lenTR+1)}+1\}}$. 
	If~$s$ was the only stair of size~${\floor{\lenBL/(\lenTR+1)}}$, then, given the total size~$\lenBL$ of all steps in~$S$, all steps in~${S'}$ must have the same size~${(\lenBL+1)/(\lenTR+1)}$. 
	Otherwise, if~$s$ was not the only stair in~$S$ of size~${\floor{\lenBL/(\lenTR+1)}}$, then only two different step sizes occur for~${S'}$ and, in particular, we have~${\floor{\lenBL/(\lenTR+1)}=\floor{(\lenBL+1)/(\lenTR+1)}}$.
	Hence, in every case, all steps in~${S'}$ have size in~${\left\{\floor{(\lenBL+1)/(\lenTR+1)},\ceil{(\lenBL+1)/(\lenTR+1)}\right\}}$.
	Given all these facts, Observation~\ref{obs:two-stairs} implies that~${P'_2}$ is a minimum-area polygon.
	
	\item
	Thirdly, assume that we also prolonged an extreme edge of~$P_3$. By a similar argument that we used for~$P_1$, one can show that this may happen only if we guessed Case~\ref{caseFourCanonical:bothHor} for~$\Bbr$ and~${\lenTR < \lenBL}$ holds.
	In what follows, let~$r$ be the step of the delimiter of~$\TR$ in~$P_2$. If we did not prolong any extreme edge of~$P_1$ and if~$r$ has size~${\floor{\lenBL/(\lenTR+1)}}$, 
	then we can conduct a similar discussion as in Case~\ref{casecut:pone} and obtain~${\area{P}\le \area{\Popt}}$.
	
	To this end, we therefore assume that
	\begin{inlinelistAlph}
		\item\label{assumption:prolong} we did prolong an extreme edge of~$P_1$ (it has to be the bottom one), or that
		\item\label{assumption:bigsize} the step~$r$ has size~${\floor{\lenBL/(a+1)}+1}$.
	\end{inlinelistAlph}  
	We cut~$\Popt$ and~$P$ in the same way as in Case~\ref{casecut:pone} and we define, for~${1\le i\le 3}$, the variables~${P_i'}$ and~${{\Popt_i}'}$, as well as~${\BL'}$,~$S$,~${S'}$,~$s$ and~${s'}$ in the same way as in Case~\ref{casecut:pone}.
	Note that~$r$ is in~$S$ and in~${S'}$. 
	Also note that~$r$ has the smallest size among all steps in~${S\setminus\{s\}}$ as it received at least the second-highest priority when computing~$P_2$.
	We claim that~${P'_2}$ is a minimum-area polygon and that~$r$ 
	is a smallest step in~${S'}$. Given this claim, we can conduct a similar discussion as in Case~\ref{casecut:pone} and obtain~${\area{P}\le \area{\Popt}}$.
	
	If we did prolong the bottom extreme edge~$e$ of~$P_1$ (which can happen only in Case~\ref{caseFourCanonical:bothVer} for~$\Btl$), then the polygon~${P_2'}$ is of minimum area by our discussion of Case~\ref{casecut:pone}. 
	Given that~$r$ has the smallest size among all steps in~${S\setminus\{s\}}$ and given that~${s'}$ is greater than~$s$, we conclude that~$r$ is a smallest step in~${S'}$.
	
	Otherwise, assume that we did not prolong~$e$. 
	There are two immediate consequences. First,~$e$ is at least one unit longer than the top extreme edge of~$P_2$. Thus, Observation~\ref{obs:two-stairs} implies~${\len{e}\ge \floor{\lenBL/(\lenTR + 1)}+1}$.
	Secondly,~${\floor{\lenBL/(a+1)}+1}$ is the size of the step~$r$ as one of the two assumptions~\ref{assumption:prolong} or~\ref{assumption:bigsize} must hold.
	Hence, given the size of~$r$, the left extreme edge~$g$ of~$P_3$ has length~${\floor{\lenBL/(a+1)}+2}$ after its prolongation.
	Recall our assumption~${\rnum{\TL}\le \rnum{\BR}}$ and observe~${\len{e}=\rnum{\TL}+1}$ and (after prolongation)~${\len{g}>\rnum{\BL}+1}$.
	Thus,
	\[\floor{\lenBL/(\lenTR + 1)}+1~\le~ \len{e}~<~\len{g}~=~\floor{\lenBL/(\lenTR + 1)}+2~,\]
	and so we have~${\len{e}= \floor{\lenBL/(\lenTR + 1)}+1}$.
	Further, observe that 
	all line segments belonging to~${\BL'}$ as well as the left and bottom extreme edge are of unit length in~${P'_2}$.
	Since the top extreme edge of~${P_2'}$ coincides with~$e$, we conclude that~${s'}$ is a good step that is bigger than~$s$ by exactly one unit. 
	Given that~$r$ has the smallest size among all steps in~${S\setminus\{s\}}$
	and size greater than~$s$, all steps in~${S\setminus\{s\}}$ are of size~${\floor{\lenBL/(a+1)}+1}$. Thus, all steps in~${S'}$ are good and of the same size, hence, of size~${(\lenBL+1)/(a+1)}$. 
	Therefore, by Observation~\ref{obs:two-stairs},~${P'_2}$ is a minimum-area polygon and~$r$ a smallest step in~${S'}$.
	
\end{listArabic}
We conclude that we computed a polygon of minimum area. The run time is linear in~$n$
since our algorithm computes only constantly many~$1$-stair and~$2$-stair instances 
which are each solvable in linear time.
Given the number of reflex vertices for the four stairs, we can even compute the minimum 
area in constant time since this is true for instances with two or less stairs. 
This observation completes our proof of Theorem~\ref{thm:xyarea}.

\subsection{The \texorpdfstring{${x}$}{x}-Monotone Case}\label{subsec:xmonarea}

For the~$x$-monotone case, we first give an algorithm that minimizes the
bounding box of the polygon, and then an algorithm that minimizes the
area.

An~$x$-monotone polygon consists of two \emph{vertical extreme} edges, that is,
the leftmost and the rightmost vertical edge, and at least two \emph{horizontal extreme} 
edges, which are defined to be the horizontal edges of locally maximum or minimum height.
The vertical extreme edges divide the polygon into an upper and a lower hull, 
each of which consists of~$xy$-monotone chains 
that are connected by the horizontal 
extreme edges. We call a horizontal extreme edge of type~${\RS\RS}$ an \emph{inner
	extreme edge}, and a horizontal extreme edge of type~${\LS\LS}$ an \emph{outer extreme 
	edge}; see Fig.~\ref{fig:bbarea-canonical1}. 
Similar to the~$xy$-monotone case,
we consider a \emph{stair} to be an~$xy$-monotone chain between any two consecutive extreme edges (outer and inner extreme
edges as well as vertical extreme edges)  
and we let \emph{stair sequence} denote the corresponding angle subsequence~${(\LS\RS)^*}$.  
Without loss of generality, at least one inner extreme edge exists, otherwise the polygon is~$xy$-monotone and we refer to Section~\ref{subsec:xymonarea}. Given an~$x$-monotone sequence, we always draw the first~${\RS\RS}$-subsequence as the leftmost inner extreme edge of the lower hull. By this, the correspondence between the angle subsequences and the stairs and extreme edges is unique.

\begin{figure}[tb]
	\subcaptionbox{An~$x$-monotone polygon.\label{fig:bbarea-canonical1}}
	{\includegraphics[page=1,scale=.92]{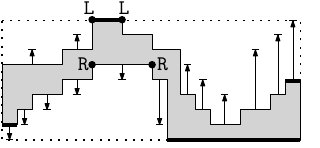}}%
	\hfill
	\subcaptionbox{Conditions~\ref{xcan:outer}--\ref{xcan:vert} are satisfied.\label{fig:bbarea-canonical2}}%
	{\includegraphics[page=2,scale=.92]{x-canonical}}%
	\hfill
	\subcaptionbox{Conditions~\ref{xcan:outer}--\ref{xcan:hor} are satisfied.\label{fig:bbarea-canonical3}}%
	{\includegraphics[page=3,scale=.92]{x-canonical}}%
	
	\caption{Illustration of how to make a polygon canonical. The bold horizontal 
		edges are outer extreme edges, the hashed area marks double stairs (see definition in proof of Theorem~\ref{thm:xbbox}). Note that the illustrating drawing is not optimal.}
	\label{fig:bbarea-canonicals}
\end{figure}
\begin{definition}
	An~$x$-monotone polygon is \emph{canonical} if 
	\begin{listRoman}[label=(D\arabic{*})]
		\item\label{xcan:outer} all outer extreme edges are lying on the border of the
		bounding box,
		\item\label{xcan:vert} each vertical non-extreme edge that is not incident to an inner
		extreme edge has length~$1$, and
		\item\label{xcan:hor} each horizontal edge that is not an outer extreme edge 
		has length~$1$.
	\end{listRoman}
\end{definition}

The following lemma states that it suffices to find a canonical~$x$-monotone polygon of minimum bounding box; see 
Fig.~\ref{fig:bbarea-canonicals} for an illustration.

\begin{lemma}\label{lem:transform}
	Any~$x$-monotone polygon can be transformed into 
	a canonical~$x$-monotone polygon without changing its bounding 
	box.
\end{lemma}
\begin{proof}
	Let~$P$ be an~$x$-monotone polygon. We transform it into a canonical polygon in two steps without changing its bounding box.
	
	First, we move all horizontal edges on the upper hull as far up as possible
	and all horizontal edges on the lower hull as far down as possible; see
	Fig.~\ref{fig:bbarea-canonical1} and~\ref{fig:bbarea-canonical2}. This 
	establishes Condition~\ref{xcan:outer}.
	Furthermore, assume that there is a vertical edge~${(u,v)}$ on the upper hull 
	with~${y(u)>y(v)+1}$. If the (unique) horizontal edge~${(v,w)}$ is not an
	inner extreme edge, then it can be moved upwards until~${y(u)=y(v)+1}$, which
	contradicts the assumption that all horizontal edges on the upper hull are 
	moved as far up as possible. This argument applies symmetrically to 
	the edges on the lower hull. Hence, Condition~\ref{xcan:vert} is established.
	
	Second, we move all vertical edges on a stair as far as possible in the 
	direction of the inner extreme edge bounding the stair, for instance, if the stair
	lies on the upper hull and is directed downwards, then all vertical edges are
	moved as far right as possible; see 
	Fig.~\ref{fig:bbarea-canonical2} and~\ref{fig:bbarea-canonical3}.
	This movement stretches the outer extreme edges while simultaneously contracting all
	other horizontal edges to length~$1$, which satisfies Condition~\ref{xcan:hor}.
	
	Note that in neither step the bounding box changed. Since all conditions are satisfied, the resulting polygon is canonical. 
\end{proof}

We observe that the length of the vertical extreme edges depends on the height
of the bounding box, while the length of all other vertical edges is fixed by
the angle sequence. Thus, a canonical~$x$-monotone polygon is fully 
described by the height of its bounding box and the length of its outer extreme
edges. Furthermore, the~$y$-coordinate of each vertex depends solely on the
height of the bounding box.

We use a dynamic program that constructs a canonical polygon of minimum
bounding box in time~${\bigOh(n^3)}$. For each possible height~$h$ of the bounding 
box, the dynamic program populates a table that contains an entry 
for any pair of an extreme vertex~$p$ (that is, an endpoint of an outer extreme 
edge) and a horizontal edge~$e$ of the opposite 
hull. The value of the entry~${T[p,e]}$ is the minimum width~$w$ such that the 
part of the polygon left of~$p$ can be drawn in a bounding box of height~$h$ and
width~$w$ in such a way that the edge~$e$ is intersecting the interior of the 
grid column left of~$p$. 

\begin{theorem}\label{thm:xbbox}
	Given an~$x$-monotone angle sequence~$S$ of 
	length~$n$, we can find a polygon~$P$ that realizes~$S$ and minimizes its
	bounding box in~${\bigOh(n^3)}$ time.
\end{theorem}
\begin{proof}
	To prove the theorem, we present an algorithm that constructs a canonical 
	polygon of minimum
	bounding box in time~${\bigOh(n^3)}$. The height of any minimum bounding box is
	at most~$n$; otherwise, as there are only~$n$ vertices, there is a~$y$-coordinate on the grid that contains no vertex and can be ``removed''.
	For any height~$h$ of the~$n$ possible heights of an optimum
	polygon, we run the following dynamic program in~${\bigOh(n^2)}$ time. 
	
	We call the left and right endpoint of an outer extreme edge the \emph{left 
		extreme vertex} and the \emph{right extreme vertex}, respectively. The dynamic
	program contains an entry
	for any pair of an extreme vertex~$p$ and a horizontal edge~$e$ of the opposite 
	hull. 
	Consider the part of the polygon between~$p$ and~$e$ that includes the left vertical extreme edge, that is, the chain that goes from~$p$ to~$e$ over the left vertical extreme edge. 
	The value of the entry~${T[p,e]}$ is the minimum width~$w$ of a bounding box of height~$h$ in which 
	this part of the polygon 
	can be drawn in such a way that edge~$e$ is intersecting the interior of the 
	grid column left of~$p$ and such that~$e$ has the same~$y$-coordinate as it has in a canonical drawing of the whole polygon in a bounding box of height~$h$; see Fig.~\ref{fig:bbarea-DP}. We call~${(p,e)}$ an
	\emph{extreme column pair}.
	
	\begin{figure}[tb]
		\centering
		\includegraphics{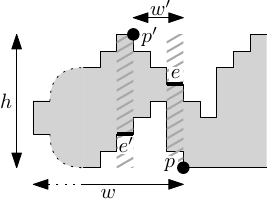}%
		\caption{Two extreme column pairs~${(p,e)}$ and~${(p',e')}$ with~${T[p,e]=T[p',e']+w'=w}$. 
			The part of the polygon left of~$p$ can be drawn in 
			the bounding box of size~${h\times w}$.}
		\label{fig:bbarea-DP}
	\end{figure}
	
	We compute~${T[p,e]}$ as follows. 
	Consider a drawing of the part of the polygon between~$p$ and~$e$
	that includes the left vertical extreme edge in a bounding box of height~$h$ and
	minimum width. Let~${p'}$ be the rightmost extreme vertex in this drawing to the left of~$p$,
	let~${(p',e')}$ be the corresponding extreme column pair, and let~${w'}$ be the
	horizontal distance between~$p$ and~${p'}$; see Fig.~\ref{fig:bbarea-DP}. 
	
	We can find~${(p',e')}$ and~${w'}$ from the angle sequence as follows.
	If~$p$ is a left extreme 
	vertex, then, by Condition~\ref{xcan:hor}, 
	the 
	pair~${(p',e')}$ and the distance~${w'}$ is fully determined. Otherwise, if~$p$ is a right extreme vertex, then~${p'}$ is 
	either the left extreme vertex incident to~$p$, or~${p'}$ is the horizontally 
	closest extreme vertex on the opposite hull; we test both cases. Again, by 
	Condition~\ref{xcan:hor},  
	edge~${e'}$ and distance~${w'}$ is
	fully determined.
	
	When determining~${(p',e')}$ and~${w'}$, we also test, as we will describe in the next paragraph, whether we can canonically draw
	the part of the polygon between~${(p',e')}$ and~${(p,e)}$ in the given space constraints. If we can, then we call~${(p',e')}$ a feasible pair for~${(p,e)}$.
	We find a feasible pair~${(p',e')}$ for~${(p,e)}$ with the smallest value of~${T[p',e'] + w'}$ and set
	\[T[p,e] = T[p',e'] + w'~.\] 
	If all pairs for~${(p,e)}$ are infeasible, we set~${T[p,e]=\infty}$.
	
	First, we will argue that if there is such a canonical drawing, then it is unique. 
	We assume~${T[p',e']<\infty}$.
	We group each pair of stairs that share an inner extreme edge as a
	\emph{double stair}; see Fig.~\ref{fig:bbarea-canonical3}. Each remaining stair forms a double stair by itself.
	Let~$\partHullTop$ denote the part of the upper hull between~${(p',e')}$ and~${(p,e)}$.
	Given the choice of~${p'}$, it does not contain any endpoint of an outer extreme edge 
	in its interior.  Hence, there are only two cases. Either~$\partHullTop$ consists of
	a single horizontal line segment belonging to an outer extreme edge, or it is a 
	subchain belonging to a double stair. 
	In the first case, by Condition~\ref{xcan:outer}, we have to draw~$\partHullTop$ on 
	the top boundary of the bounding box. Further, its left endpoint has~$x$-coordinate 
	equal to~${T[p',e']}$ and the length of the segment is~${w'}$. Hence, the drawing is unique.
	In the second case, note that conditions~\ref{xcan:outer}--\ref{xcan:hor} determine 
	the lengths and~$y$-positions of all edges with the exception of the lengths of the outer extreme edges. 
	Thus, given the~$x$-position of any vertex of a double stair, there is only one 
	canonical way to draw the double stair. In our case, the value of~${T[p',e']}$ is equal 
	to the~$x$-position of the leftmost vertex of~$\partHullTop$. Hence, the drawing 
	of~$\partHullTop$ is unique. By the same arguments, we have to draw the part~$\partHullBot$ 
	of the lower hull between~${(p',e')}$ and~${(p,e)}$ in a unique way. 
	
	Now, given the unique drawings of~$\partHullTop$ and~$\partHullBot$, we check for 
	every~$x$-coordinate whether~$\partHullTop$ is lying above~$\partHullBot$. If and only 
	if this is the case, then the two drawings together form a feasible canonical drawing 
	and~${(p',e')}$ is a feasible pair for~${(p,e)}$.
	
	In the last step, we compute the minimum width~$w$ of the bounding box assuming height~$h$.
	Consider an optimum canonical drawing of the whole polygon in a bounding box of height~$h$. 
	Let~${p^*}$ be a rightmost (right) extreme vertex. Note that for~${p^*}$ there are only two candidates, one from the upper hull and one from the lower hull.
	Since~${p^*}$ is a rightmost extreme vertex, all horizontal 
	edges to the right of~${p^*}$ (on the upper and on the lower hull) are segments of length~$1$. 
	Thus, given~${p^*}$, we can compute the distance~${r^*}$ between~${p^*}$ and  
	the right vertical extreme edge. 
	Let~${e^*}$ be the~${r^*}$\thSuffix horizontal edge from the right on the hull opposite to~${p^*}$.
	Observe that edge~${e^*}$ is the edge that forms an extreme column pair with~${p^*}$. Hence, the width of the polygon is~${w=T[p^*,e^*]+r^*}$.
	
	We compute width~$w$ as follows. For each one of the two candidates for~${p^*}$, we determine~${r^*}$ and~${e^*}$. Then we check whether the candidate is feasible.
	For this, recall that \mbox{Conditions~\ref{xcan:outer}--\ref{xcan:hor}} determine the~$y$-positions of all edges. Also recall that all horizontal edges to the right of~${(p^*,e^*)}$ are of length~$1$.
	Hence, there is only one way to canonically draw the edges right to~${(p^*,e^*)}$. If the upper hull always stays above the lower hull, candidate~${p^*}$ is feasible.
	Thus, we get the width by \[w~=~\min_{\text{feasible candidate }p^*}\{T[p^*,e^*]+r^*\} \cup \{\infty\}~.\]
	
	For every height~$h$, we compute the minimum width~$w$ and find the bounding box of minimum area~${w \cdot h}$.
	
	It remains to show the run time of the algorithm. The table~$T$ consists 
	of~${\bigOh(n^2)}$ entries. To find the value of an entry~${T[p,e]}$, we have to
	find the closest column pair~${(p',e')}$ to the left, the distance~${w'}$, and we 
	have to test whether we can canonically draw
	the polygon between~${(p',e')}$ and~${(p,e)}$. We now show 
	that each of these steps is possible in~${\bigOh(1)}$ time by precomputing some values 
	for each point. 
	
	\begin{enumerate}[label=(\roman*)]
		\item\label{precomp:ycoord}
		For each point, we store its~$y$-coordinate. As observed above, the~$y$-coordinate is fixed, and it can be computed in~${\bigOh(n)}$ time in total by
		traversing the stairs.
		\item\label{precomp:next}
		For each point~$p$, 
		we store the next extreme
		point~${\lambda(p)}$ to the left on the same hull, as well as the 
		distance~${\delta(p)}$ to it. These values
		can be computed in~${\bigOh(n)}$ time in total by traversing the upper and the lower hull from
		left to right.
		\item\label{precomp:array}
		For each left extreme vertex~$q$, we store an array that contains all
		horizontal edges between~$q$ and~${\lambda(q)}$ ordered by their appearance on a walk
		from~$q$ to~${\lambda(q)}$ on the same hull. We also store the index of the inner 
		extreme edge in this array. These arrays can be computed altogether in~${\bigOh(n)}$ time 
		by traversing the upper and the lower hull from right to left.
	\end{enumerate}
	
	The precomputation takes~${\bigOh(n)}$ time in total. Given an extreme column 
	pair~${(p,e)}$, let~$l_e$ be the left endpoint of~$e$. We can use 
	the precomputation of step~\ref{precomp:next} to find in~${\bigOh(1)}$ time the closest extreme 
	vertex~${p'}$ to the left of~$p$, since it is either~${\lambda(p)}$ 
	or~${\lambda(l_e)}$, as well as the distance~${w'}$, which is either~${\delta(p)}$ 
	or~${\delta(l_e)}$. To test whether we can canonically draw the polygon 
	between~${(p',e')}$ and~${(p,e)}$, we make use of the fact that there is no
	outer extreme edge between them. Hence, we only have to test whether a pair
	of opposite double stairs intersects. To this end, we observe that a pair of
	double stairs can only intersect if the inner extreme edge of the lower hull 
	lies (partially) above the upper hull or the inner extreme edge of the upper 
	hull lies (partially) below the lower hull. With the array precomputed in 
	step~\ref{precomp:array}, we can find the edge opposite of the inner extreme 
	edges, and by step~\ref{precomp:ycoord}, each point (and thus each 
	edge) knows its~$y$-coordinate, which we only have to compare to find out 
	whether an intersection exists. Hence, we can compute each table entry in~${\bigOh(1)}$
	time after a precomputation step that takes~${\bigOh(n)}$ time.
	
	Since we call the dynamic program~${\bigOh(n)}$ times---once for each candidate for
	the height of the bounding box---the algorithm takes~${\bigOh(n^3)}$ time in total.
	Following Lemma~\ref{lem:transform}, this proves the theorem. 
\end{proof}

For the area minimization, we make two key observations. 
First, since the polygon is~$x$-monotone, each grid 
column (properly) intersects 
either no or exactly two horizontal edges: one edge
from the upper hull and one edge from the lower hull. Secondly, a pair of
horizontal edges share at most one column; otherwise, the polygon could be
drawn with less area by shortening both edges. With the same argument as for
the bounding box, the height of any minimum-area polygon is at most~$n$.

We use a dynamic program to solve the problem. To this end, we fill a
three-dimensional table~$T$ as follows. Let~$e$ be a horizontal edge on 
the upper hull, let~$f$ be a horizontal edge of the lower hull, and 
let~$h$ be an integer satisfying~${1\le h\le n}$. The entry~${T[e,f,h]}$ specifies the minimum
area required to draw the part of the polygon to the left of (and 
including) the unique common column of~$e$ and~$f$ under the condition that~$e$ 
and~$f$ share a column and have vertical distance~$h$. 

Let~${e_1,\dots,e_k}$ be the horizontal edges on the upper hull from left to right
and let~${f_1,\dots,f_m}$ be the horizontal edges on the lower hull from left to 
right. For each~$h$ with~${1\le h\le n}$, we initialize the table with~${T[e_1,f_1,h]=h}$. 
To compute any other entry~${T[e_i,f_j,h']}$, we need to find the correct entry from the
column left of the column shared by~$e_i$ and~$f_j$. There are three 
possibilities: this column either intersects~${e_{i-1}}$ and~${f_{j-1}}$, it 
intersects~$e_i$ and~${f_{j-1}}$, or it intersects~${e_{i-1}}$ and~$f_j$. For each
of these possibilities, we check which height can be realized if~$e_i$
and~$f_j$ have vertical distance~${h'}$ and search for the entry of
minimum value.  We set
\[T[e_i,f_j,h']~=\min_{h'' \text{ valid}}\{T[e_{i-1},f_{j-1},h''],
T[e_i,f_{j-1},h''], T[e_{i-1},f_j,h'']\}+h'~.\]
Finally, we can find the optimum solution by 
finding~${\min_{1\le h\le n}\{T[e_k,f_m,h]\}}$. Since the table has~${\bigOh(n^3)}$ 
entries each of which we can compute in~${\bigOh(n)}$ time, the algorithm runs 
in~${\bigOh(n^4)}$ time. This proves the following theorem.

\begin{theorem}\label{thm:xbarea}
	Given an~$x$-monotone angle sequence~$S$ of length~$n$, we can find
	a minimum-area polygon that realizes~$S$ in~${\bigOh(n^4)}$ time.
\end{theorem}

\section{The Monotone Case: Minimum Perimeter}\label{sec:peri-algo}

In this section, we show how to compute a polygon of minimum perimeter
for any~$xy$-monotone or~$x$-monotone angle sequence~$S$ of length~$n$.

Let~$P$ be an~$x$-monotone polygon realizing~$S$.
Let~$e_L$ be the leftmost vertical edge and let~$e_R$ be
the rightmost vertical edge of~$P$.
Recall that~$P$ consists of two~$x$-monotone chains; an upper chain~$T$
and a lower chain~$B$ connected by~$e_L$ and~$e_R$.
For every~${e\in T}$, let~${T(e_R,e)}$ denote the subchain of~$T$ consisting of all segments between~$e_R$ and~$e$ (without~$e_R$ and~$e$).
Similarly, for every~${e'\in B}$, let~${B(e',e_R)}$ denote the subchain of~$B$ consisting of all segments between~${e'}$ and~$e_R$ (without~${e'}$ and~$e_R$).
Without loss of generality, we assume that the number of reflex vertices of~$T$ and~$B$ satisfies~${\rnum{T}\geq \rnum{B}}$.

\begin{definition}
	An~$x$-monotone polygon is \emph{perimeter-canonical} if
	\begin{listArabic}
		\item\label{pericanon:vert}
		every vertical edge except~$e_R$ and~$e_L$ has unit length, and
		\item\label{pericanon:hor}
		every horizontal edge of~$T$ has unit length. 
	\end{listArabic}
\end{definition}

We show that it suffices to find a perimeter-canonical polygon of minimum perimeter.
\begin{figure}
	\centering
	\includegraphics[page=6]{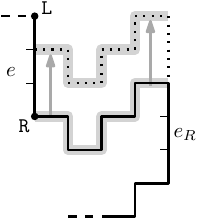}
	\caption{First step of transforming~$P$ into a canonical form. We decrease~$\len{e}$ by increasing~${\len{e_R}}$.}
	\label{fig:mono-hori-edge-1}
\end{figure}

\begin{lemma}\label{lem:mono-unit-length}
	Any~$x$-monotone polygon can be transformed into a perimeter-ca\-no\-ni\-cal~$x$-mo\-no\-tone polygon without increasing its perimeter.
\end{lemma}
\begin{proof}
	We transform any minimum-perimeter polygon into a perimeter-canonical
	form without sacrificing its perimeter in two steps as follows.
	First, we shorten every \emph{long} vertical edge~${e \in T\cup B}$ 
	with~${\len{e}>1}$ so that~${\len{e}=1}$ holds. 
	This shortening is always possible:
	For any long vertical edge~${e\in T\cup B}$, say~${e \in T}$, 
	if its end vertices have turns~${\RS\LS}$
	in ccw order, then we proceed as follows; see Fig.~\ref{fig:mono-hori-edge-1}. 
	We move the subchain~${T(e_R,e)}$ 
	upward by~${\len{e}-1}$ units by shortening~$e$ 
	and by simultaneously stretching~$e_R$. 
	This movement guarantees that~$\len{e}$ decreases and~${\len{e_R}}$ increases by
	the same amount of~${\len{e}-1}$, so the perimeter remains the same.
	We can also shorten any long vertical edge whose end vertices have turns~${\LS\RS}$ in a symmetric way.
	
	Secondly, we shorten every long horizontal edge~${e \in T}$ with~${\len{e}>1}$
	so that its length becomes~$1$. 
	Suppose that~$e$ is the rightmost long horizontal edge~$e$ in~$T$.
	Since~${\rnum{T}\geq \rnum{B}}$, there must be a long horizontal edge~${e'}$ in~$B$.
	We shorten both~$e$ and~${e'}$ by one unit, and 
	move the two subchains~${T(e_R, e)}$ and~${B(e', e_R)}$
	together with~$e_R$ one unit left. 
	This move may cause two vertical edges,~${f\in T}$ 
	and~${f'\in B}$, to intersect; 
	see Fig.~\ref{fig:mono-hori-edge-2}. 
	Note that exactly one of both vertical edges did not move, say~${f'}$, 
	as otherwise there would be no intersection between them. 
	This means~${f'}$ is to the left of~${e'}$, that is,~${f'\in B\setminus B(e', e_R)}$. 
	We also know that the~$x$-distance between~$f$ and~${f'}$ prior to the move was one, 
	otherwise they would not intersect. 
	Since~$f$ and~${f'}$ are of unit length, 
	the lower end vertex of~$f$ has the same~$y$-coordinate 
	as the upper end vertex of~${f'}$. 
	To avoid the intersection, we first move the whole upper chain~$T$ 
	one unit upward by stretching~$e_R$ and~$e_L$ each by one unit, 
	as in Fig.~\ref{fig:mono-hori-edge-3}. 
	Then we can move~${T(e_R, e)}$,~${B(e', e_R)}$, and~$e_R$ one unit to the left 
	without causing any intersection. 
	We get rid of two units by shortening~$e$ and~${e'}$, and receive two units 
	by stretching~$e_R$ and~$e_L$, 
	so the total perimeter remains unchanged. 
	We repeat this second step until~${\len{e}=1}$.
\end{proof}
\begin{figure}
	\centering
	\subcaptionbox{\label{fig:mono-hori-edge-2}}{\includegraphics[page=2]{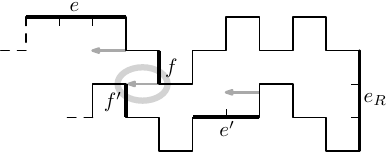}}
	\hfill
	\subcaptionbox{\label{fig:mono-hori-edge-3}}{\includegraphics[page=3]{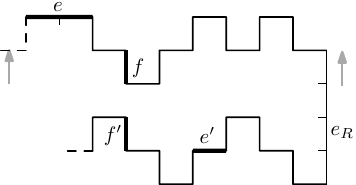}}
	\caption{Steps two and three of transforming~$P$ into a canonical form. 
		We decrease 
		the length of~$e$ and~${e'}$ (bold)
		by increasing~${\len{e_L}}$ and~${\len{e_R}}$. Stretching~$e_l$ and~$e_R$ prevents the crossing of~$f$ and~${f'}$ (bold).}
\end{figure}

Assume that~$P$ is a minimum-perimeter canonical polygon 
that realizes~$S$. Assume further that~${\rnum{T}\geq \rnum{B}}$ holds. 
Let~$\peri{P}$ denote the perimeter of~$P$. 
By Conditions~\ref{pericanon:vert}--\ref{pericanon:hor}, every edge in~$T$ is of
unit length, so the length of~$T$ is~${2\rnum{T}+1}$. 
This property implies that the width of~$B$ should be~${\rnum{T}+1}$.
By Condition~\ref{pericanon:vert}, the length
of the vertical edges in~$B$ is~$\rnum{B}$, so the total length of~$B$
is~${\rnum{T}+\rnum{B}+1}$. Thus, we can observe the following property.
\begin{lemma}
	\label{lem:peri-equation} 
	Given an~$x$-monotone angle sequence~$S$, 
	there is a canonical minimum-perimeter polygon~$P$ realizing~$S$
	with~${\rnum{T}\geq \rnum{B}}$ such that~${\peri{P} = 3\rnum{T}+\rnum{B}+2+\len{e_L}+\len{e_R}}$ holds.
\end{lemma}%\qed

The first three terms of~$\peri{P}$ in Lemma~\ref{lem:peri-equation} 
are constant, so we need to minimize the sum of the last two 
terms,~${\len{e_L}}$ and~${\len{e_R}}$, to get a minimum perimeter. 
However, once one of them is fixed, the other is automatically determined 
by the fact that all vertical edges in~$T$ and~$B$ are unit-length segments.
Even more, minimizing one of them is equivalent to 
minimizing their sum, consequently minimizing the perimeter.
We call the length of the left vertical extreme edge of a polygon 
the \emph{height} of the polygon.

\subsection{The \texorpdfstring{${xy}$}{xy}-Monotone Case}
Let~$P$ be a minimum-perimeter canonical~$xy$-monotone polygon  
that realizes an~$xy$-mo\-no\-tone angle sequence~$S$ of length~$n$. 
As before, we assume that~${\rnum{T}\ge \rnum{B}}$ holds. 
When~${n = 4}$, that is, the number~$r$ of reflex vertices is~$0$,
then a unit square~$P$ achieves 
the minimum perimeter. Therefore, we assume in the following that we have~${r > 0}$.
Recall that the boundary of~$P$ consists of four stairs,~${\TR, \TL, \BL}$, and~$\BR$. 
Let~${(r_1, r_2, r_3, r_4)}$ be a quadruple of the numbers of 
reflex vertices of~${\TR, \TL, \BL}$, and~$\BR$, respectively. 
Then~${r = r_1+r_2+r_3+r_4}$, where~${r_i \geq 0}$ for each~$i$ with~${1 \le i \le 4}$. 
Again, we define~$L$ as the chain consisting of~$\TL$,~$e_L$ and~$\BL$ and~$R$ as the chain consisting of~$\BR$,~$e_R$ and~$\TR$.
In~$P$, let~${w(T)}$ and~${w(B)}$ denote the widths of~$T$ and~$B$, respectively, and~${h(L)}$ and~${h(R)}$ the heights of~$L$ and~$R$, respectively. 
Hence, the perimeter of~$P$ 
is~\[{\peri{P}~=~w(T)+w(B)+h(L)+h(R)}~.\] 

Note that~${w(T)=w(B)}$ holds and, by Condition~\ref{pericanon:hor},
\[{w(T)~=~r_1+1+r_2}~.\] 
Thus,~${w(T)+w(B)=2(r_1+r_2)+2}$. 
Similarly,~${h(L)=h(R)}$, and, by Condition~\ref{pericanon:vert},
\[{h(L)~=~r_2+\len{e_L}+r_3} \textrm{\quad and \quad} {h(R)~=~r_4+\len{e_R}+r_1}~.\] 
Thus, if~${\len{e_L}=1}$, then
\[{h(L)+h(R)~=~2(r_2+r_3)+2}~,\] 
and, if~${\len{e_R}=1}$, then
\[{h(L)+h(R)~=~2(r_1+r_4)+2}~.\]
Furthermore observe
that~${\len{e_L}=1}$ implies
\[{r_2+r_3 ~\ge~ r_1+r_4}~,\] and that~${\len{e_R}=1}$ implies
\[{r_2+r_3 ~\le~ r_1+r_4}~.\]
Hence, if~${\len{e_L}=1}$ or~${\len{e_R}=1}$,
then
\[{h(L)+h(R)~=~r+|r_2 + r_3- r_1- r_4|}\] 
and eventually
\begin{align}\label{eq:per-when-height-zero}
\peri{P}=3(r_1+r_2)+(r_3+r_4)+|r_2 + r_3- r_1- r_4|+4 ~.
\end{align}

Now, consider the remaining case when~${\len{e_L}\ge 2}$ and~${\len{e_R}\ge 2}$. We will observe that this case can occur only if~${(r_1, r_2, r_3, r_4)}$ is~${(r_1,0,r_1,0)}$ or~${(0, r_2, 0,r_2)}$. We will also observe that then~${\len{e_L}=\len{e_R}=2}$. 
Hence, we obtain that~${\peri{P}=2r_1+6}$ for case~${(r_1,0,r_1,0)}$, and~${\peri{P}=2r_2+6}$ for case~${(0, r_2, 0,r_2)}$. For all other cases, Equation~\ref{eq:per-when-height-zero} holds.

To make these observations, we first apply the same contraction step as depicted in Fig.~\ref{fig:bbarea-canonical2} of 
Lemma~\ref{lem:transform}. That is, 
we contract all horizontal segments of~$\BL$ to length~$1$ by moving all their right endpoints as far as possible to the left,
and we contract all horizontal segments of~$\BR$ to length~$1$ by moving all their left endpoints as far as possible to the right. 
By this, all edges of~$B$ except the bottom extreme edge have length~$1$, and 
the perimeter does not change. 
Next, note that~$T$ and~$B$ have vertical distance~$1$ to each other. Otherwise, we could move~$B$ at least one unit to the top by simultaneously shrinking~$e_L$ and~$e_B$, and thus shrinking the perimeter of~$P$, a contradiction to the minimality of~$\peri{P}$.
As~$T$ consists only of unit-length segments (Conditions~\ref{pericanon:vert}--\ref{pericanon:hor}), there is a vertex~$p$ in~$T$ having distance~$1$ to~$B$. 

First assume that~$p$ belongs to~$\TR$. We choose the rightmost such~$p$.
If~$p$ were a convex vertex, then it would be the top endpoint of~$e_R$, and, hence, we would have~${\len{e_R}=1}$; a contradiction to~${\len{e_R}\ge 2}$.
Thus,~$p$ is a reflex vertex and therefore a left endpoint of a horizontal edge~${pp'}$. Hence, the right endpoint~${p'}$ of~${pp'}$ is convex. Let~$e$ be the edge in~$B$ below~${pp'}$, that is, the edge that crosses the same grid column as~${pp'}$. 
Observe that the distance between~${pp'}$ and~$e$ is at least~$2$. If it were~$1$, then the vertical edge~${p'p''}$ incident to~${p'}$ would connect to~$e$ (recall that~${p'}$ is convex). Hence,~${pp'}$ and~$e$ would be incident to~${e_R=p'p''}$, and again we would have~${\len{e_R}=1}$; a contradiction. 
Thus, the distance between~$p$ and~$e$ is at least~$2$.
Let~$q$ be the point of~$B$ directly one unit below~$p$. Then~$e$ lies at least one unit below~$q$.
Hence,~$q$ has to connect to~$e$ via a vertical edge, and, consequently,~$q$ has to be a reflex vertex and belong to~$\BL$. 
By Condition~\ref{pericanon:vert}, the vertical edge connecting~$q$ and~$e$ has length~$1$, hence, the distance between~${pp'}$ and~$e$ is exactly~$2$.
But now, either the bottom endpoint~${p''}$ of~${p'p''}$ has distance~$1$ to~$B$, or~${p''}$ lies on~$B$, that is,~${p'p''=e_R}$. The former case contradicts our assumption that~$p$ is the rightmost vertex of~$T$ having distance~$1$ to~$B$.  Thus, the latter case holds and~${pp'}$ and~$e$ are incident to~$e_R$. 
Hence,~${\len{e_R}=2}$,~$e$ is the bottom extreme edge and has length~${\len{e}=1}$, and~$\BR$ is empty, that is,~${r_4=0}$.  
Thus, 
all horizontal edges in~$B$ have unit length. 
This property allows us to use the same argument as above to show~${r_2=0}$ and~${\len{e_L}=2}$.
Given~${r_1+1=w(T)=w(B)=r_3+1}$, we get~${r_1=r_3}$.

Finally, assume that~$p$ belongs to~$\TL$. Then we can show in a similar way as above that we are in case~${(0, r_2, 0,r_2)}$, and, again,~${\len{e_L}=\len{e_R}=2}$.
Thus, our observation follows.

\begin{theorem}
	Given an~$xy$-monotone angle sequence~$S$ of length~$n$, we can find a 
	polygon~$P$ that realizes~$S$ and minimizes its perimeter in~${\bigOh(n)}$ time.
	Furthermore, if the lengths of the stair sequences 
	are given as above as a tuple~$\ell$ where~${\ell={(r_1,r_2,r_3,r_4)}}$, then~$\peri{P}$ can be expressed as:
	\[
	\peri{P} ~=~ 
	\begin{cases}
	4r_1+6 & \text{if } \ell~=~(r_1,0,r_1,0),\\
	4r_2+6 & \text{if } \ell~=~(0,r_2,0,r_2),\\
	3(r_1+r_2)+(r_3+r_4)+|r_3-(r_1-r_2+r_4)|+4 & \text{otherwise.}
	\end{cases}	\]
\end{theorem}

\subsection{The \texorpdfstring{${x}$}{x}-Monotone Case}

A minimum height polygon~$P$ that realizes~$S$ can be computed
in~${\bigOh(n^2)}$ time using dynamic programming. Recall that a perimeter-canonical polygon of minimum height is a polygon of minimum perimeter.

From right to left, let~${t_1,\ldots, t_{\rnum{T}}}$ be the horizontal
edges in~$T$ and~$b_1$,~$b_2$,~${\ldots ,}$~${b_{\rnum{B}}}$ be the horizontal edges
in~$B$.
Recall our assumption~${\rnum{T}\geq \rnum{B}}$. 
For~${i \geq j \geq 1}$, let~${A[i, j]}$ be
the minimum height of the subpolygon formed 
with the first~$i$ horizontal edges from~$T$ and 
the first~$j$ horizontal edges from~$B$.
Note that the leftmost vertical edge of the subpolygon whose minimum
height is stored in~${A[i,j]}$ joins the left endpoints of~$t_i$ and~$b_j$. 
To compute~${A[i,j]}$, we attach edges~$t_i$ and~$b_j$
to the upper and lower chains of the subpolygon constructed so far.
Since~$t_i$ has unit length,  
either~$t_i$ and~$b_j$ are attached to the subpolygon 
with height of~${A[i-1, j-1]}$ or just~$t_i$ is attached to
the subpolygon with height of~${A[i-1, j]}$. Figure~\ref{fig:DP-polygon-perimeter} shows that 
there are four cases, Cases~\subref{fig:DP-polygon-perimeter-1}--\subref{fig:DP-polygon-perimeter-4}, for the first attachment 
and two cases, Cases~\subref{fig:DP-polygon-perimeter-5}--\subref{fig:DP-polygon-perimeter-6}, for the second attachment, 
according to the turns formed at the attachments. 

\begin{figure}[tb]
	\centering
	\subcaptionbox{\label{fig:DP-polygon-perimeter-1}}{\includegraphics[page=1]{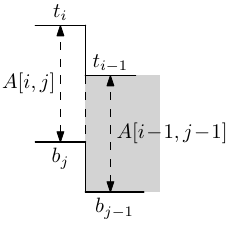}}
	\hfill
	\subcaptionbox{\label{fig:DP-polygon-perimeter-2}}{\includegraphics[page=2]{DP-polygon-perimeter-wide}}
	\hfill
	\subcaptionbox{\label{fig:DP-polygon-perimeter-3}}{\includegraphics[page=3]{DP-polygon-perimeter-wide}}
	\\
	\subcaptionbox{\label{fig:DP-polygon-perimeter-4}}{\includegraphics[page=4]{DP-polygon-perimeter-wide}}
	\hfill
	\subcaptionbox{\label{fig:DP-polygon-perimeter-5}}{\includegraphics[page=5]{DP-polygon-perimeter-wide}}
	\hfill
	\subcaptionbox{\label{fig:DP-polygon-perimeter-6}}{\includegraphics[page=6]{DP-polygon-perimeter-wide}}
	\caption{Six situations when~$t_i$ and~$b_j$ are considered to fill~${A[i, j]}$.}
	\label{fig:DP-polygon-perimeter}
\end{figure}

Let~$u$ and~$v$ be the left end vertex of~${t_{i-1}}$ and the right end
vertex of~$t_i$, respectively. Let~${u'}$ and~${v'}$ be the right end vertex
of~$b_{j}$ and the left end vertex of~${b_{j-1}}$, respectively. 
Notice that both vertical edges~${(u,v)}$ and~${(u', v')}$ have unit length. 
As an example, let us explain how to calculate~${A[i, j]}$ 
when~${uv = \LS\RS}$ and~${u'v'=\LS\RS}$, which corresponds to Fig.~\ref{fig:DP-polygon-perimeter-2} and Fig.~\ref{fig:DP-polygon-perimeter-6}. 
We set~${A[i, j]}$~to the minimum height of the two possible attachments of Cases~\subref{fig:DP-polygon-perimeter-2} and~\subref{fig:DP-polygon-perimeter-6}.
For now, consider the height for Case~\subref{fig:DP-polygon-perimeter-2}. 
If~${A[i-1,j-1]>1}$, then~$t_i$ and~$b_j$ are attached to 
the subpolygon as illustrated in Fig.~\ref{fig:DP-polygon-perimeter-2}. Since edges~${(u,v)}$ and~${(u', v')}$ have unit length,~${A[i,j] = A[i-1,j-1]}$. In the other case, if~${A[i-1,j-1] = 1}$, 
then we can move the upper chain of the subpolygon one unit upward 
without intersection so that~$t_i$ and~$b_j$ are safely attached to the
subpolygon with~${A[i,j] = 2}$. 
Note that this is the smallest possible value for~${A[i,j]}$ given~${uv=\LS\RS}$ and~${u'v'=\LS\RS}$.
Thus,~${A[i,j] = \maxStyle{A[i-1,j-1], 2}}$.
The height for Case~\subref{fig:DP-polygon-perimeter-6} should be at least~$1$, so it is expressed as~${\maxStyle{A[i-1,j]-1, 1}}$. 
Therefore, \[A[i,j] = \minStyle{\maxStyle{A[i-1,j-1],2},\maxStyle{A[i-1,j]-1,1}}~.\] 
For the other turns at~$uv$ and~${u'v'}$, 
we can similarly define the equations as follows:
{\[A[i,j]~=\left\{ 
	\begin{array}{ll}
	\text{undefined} &  \text{if~${i = 0}$,~${j = 0}$ or~${i < j}$},\\
	1 &  \text{if~${i=1}$,~${j = 1}$},\\
	A[i-1,j]+1 &  \text{if~${uv=\RS\LS}$,~${j = 1}$},\\
	\maxStyle{A[i-1,j]-1, 1} &  \text{if~${uv=\LS\RS}$,~${j = 1}$},\\
	\minStyle{\maxStyle{ A[i-1,j-1], 2}, A[i-1,j]+1} & \text{if~${uv=\RS\LS}$,~${u'v'=\RS\LS}$},\\
	\minStyle{A[i-1,j-1]+2, A[i-1,j]+1} & \text{if~${uv=\RS\LS}$,~${u'v'=\LS\RS}$},\\
	\min\left\{\maxStyle{ A[i-1,j-1], 2},\right.\\ 
	\phantom{\min\{}\left.\!\maxStyle{A[i-1,j]-1,1}\right\} &  \text{if~${uv=\LS\RS}$,~${u'v'=\LS\RS}$},\\
	\min\left\{\maxStyle{ A[i-1,j-1]-2, 1},\right.\\%
	\phantom{\min\{}\left.\!\maxStyle{A[i-1,j]-1,1}\right.\} &  \text{if~${uv=\LS\RS}$,~${u'v'=\RS\LS}$}~.
	\end{array} \right.\]}

Evaluating each entry takes constant time, so the total time to
fill~$A$ is~${\bigOh(n^2)}$.  Using~$A$, a minimum-perimeter polygon can be
reconstructed within the same time bound.

\begin{theorem}
	Given an~$x$-monotone angle sequence~$S$ of length~$n$, we can find
	a polygon~$P$ that realizes~$S$ and minimizes its perimeter in~${\bigOh(n^2)}$ time.
\end{theorem}

\section{Conclusion}

In this paper, we considered the problem of drawing a polygon satisfying a given angle sequence on a rectilinear grid such that its area, its bounding box, or its perimeter is minimized.
We have seen several efficient algorithms for~$x$-monotone and~$xy$-monotone variants of the problem
and have shown that the general variant is \NP-hard for all three objectives.
These results raise the question about the approximability of the
general problem.  Step by step, one could consider more and more
complicated objects than polygons.  Eventually, one would arrive at
the following general question: Given an orthogonal representation of
a graph that specifies an angle sequence for each edge and an angle
for each vertex, draw the graph without crossings on an integer grid
such that the orthogonal representation is realized and the bounding
box or the perimeter is minimized.
For these optimization versions of the problem, no
  approximation results are known---apart from a subpolynomial
  inapproximability bound for the non-planar case~\cite{bannister12}.
  It would therefore be interesting to study approximability in this
  context.  Are the optimization versions of the problem indeed
  \APX-hard as originally claimed by
  Patrignani~\cite{WADS99Titto,PrivComTitto}?  Or do they admit
  approximation schemes?

\subsection*{Acknowledgements}

We thank Titto Patrignani for inspiring
Theorem~\ref{thm:polylineNPhard} and for informing us about the status
of the compaction problem.

\end{document}